\documentclass[11pt]{article} 
	\usepackage[margin=1.0 in]{geometry} 
	\usepackage[english]{babel} 
	\usepackage{amsmath, amssymb, textcomp, amsthm, mathtools,mathrsfs} 
	\usepackage{booktabs}
    \usepackage{wasysym} 
	\usepackage{stmaryrd}
	\usepackage{dsfont}
	\usepackage{hyperref}
   	\usepackage{natbib}
   	\usepackage{enumitem}
   	\usepackage{inputenc} 
	\usepackage{csquotes}
	\usepackage{footnote,footmisc} 
	\usepackage{setspace} 
	\usepackage{threeparttable}
	\usepackage{tikz}
	\usepackage{adjustbox}
	\usepackage{pgfplots}
	\pgfplotsset{compat=1.17}
	\usetikzlibrary{shapes.misc}
	\usetikzlibrary{patterns}
	\usepackage{xcolor}
	\hypersetup{
    	colorlinks,
        linkcolor={red!50!black},
    	citecolor={blue!50!black},
    	urlcolor={blue!80!black}
	}
	\usetikzlibrary{arrows}
	\usepackage{float}
	\usepackage{algorithm}
	\usepackage{algpseudocode}
    \usetikzlibrary{arrows.meta}

	\usepackage[toc,page]{appendix}
	
	\newtheorem{theorem}{Theorem}
	\newtheorem{lemma}{Lemma} 
	\newtheorem{proposition}{Proposition}

	\newtheorem{assumption}{Assumption}
	\theoremstyle{definition}
	\newtheorem{example}{Example}
	\newtheorem*{example*}{Example}
	\theoremstyle{plain}
	
	\usepackage{graphicx}  
	\usepackage{marvosym}

    \newtheorem*{assumptions*}{\assumptionnumber}
\providecommand{\assumptionnumber}{}

\usepackage{caption}
\usepackage{subcaption}
\usepackage{mwe}

\definecolor{cyan}{cmyk}{1, 0.4, 0, 0}

\definecolor{mypink}{RGB}{219, 48, 122}

\makeatletter

 \def\R{\mathbb{R}}
 \def\N{\mathbb{N}}
 \def\W{\Omega}
 \def\w{\omega}
 
 \def\X{\mathcal{X}}
 
 \def\Z{\mathcal{Z}}
 \def\Pr{\mathbb{P}}
 \def\P{\mathcal{P}}

\begin{document}

\title{Partial identification with entropy regularized optimal transport}

\author{Bruno Costa\\ University of Michigan \and Florian Gunsilius\\ Emory University}

\date{\today}

\maketitle

\begin{abstract}
  In many statistical settings, the available data and maintained assumptions do not suffice to uniquely identify the model parameters of interest. In such cases, one can only identify sets which are guaranteed to contain the true parameters. These are often characterized through linear programs that optimize over models compatible with the observed data. These programs can be infinite-dimensional in the optimizer and the number of constraints. We provide a unified way to characterize and solve such optimization problems by phrasing them as optimal transport problems on path spaces. This allows us to regularize the problem with an entropy penalty, recasting it as a multi-marginal entropic optimal transport problem, which can be solved efficiently via Sinkhorn iterations. In addition, it allows us to establish convergence of the regularized value to the sharpest bound, derive consistency rates for a plug-in estimator, and obtain asymptotic distribution for approximate bounds. The method is general and accommodates settings ranging from instrumental variable models with continuous variables to welfare estimation in heterogeneous demand models. We verify the statistical and computational properties in simulations and provide an application to demand estimation. \vspace{0.2cm}

\end{abstract}

\section{Contribution and related literature}

Identification in statistical and economics models is often an all-or-nothing concept: either a parameter of interest can be uniquely determined from the model assumptions and the data-generating process, or it cannot. Partial identification relaxes this dichotomy by characterizing the parameter values that are compatible with the observed distribution and maintained assumptions. The argument for partial identification has been that it offers greater flexibility and robustness under weaker assumptions \citep[e.g.][]{manski2003partial, tamer2010partial, kline2023recent, molinari2020microeconometrics}. There are two main methodological approaches in this area, one that mainly draws on random set theory \citep[e.g.][]{molchanov2005theory, molchanov2014applications, beresteanu2012partial, chesher2017generalized} and one that mainly works with linear programming \citep[e.g.][]{balke1994counterfactual, balke1997bounds, laffers2015bounding, tebaldi2019nonparametric}.

Both approaches face substantial computational challenges when unobserved heterogeneity is unrestricted. Finite representations can grow exponentially with the observable supports -- the ``curse of cardinality'' \citep{beresteanu2012partial, russell2019sharp} -- while continuous variables can require optimization over distributions of entire response functions \citep{gunsilius2019path}. We develop a general method for approximating and estimating sharp bounds on partially identified averages that accommodates this functional heterogeneity and data restrictions across a continuous range of observed conditions. Under our regularity conditions, we consider all distributions that satisfy the model's assumptions and reproduce the full observable distributions. Building on \citet{gunsilius2019path}, we show that this problem is equivalent to constrained optimal transport over admissible response paths (Theorem~\ref{thm:strong-duality}). In the binary IV model, these paths encode the familiar compliance types \citep{angrist1996identification}.

  We make this representation computable through deterministic discretization, relaxation of the structural support constraint, and relative entropy regularization. The resulting multi-marginal entropic optimal transport problem \citep{peyre2019computational} has a smooth dual that can be solved using generalized Sinkhorn or IPFP iterations \citep{deming1940least, sinkhorn1967diagonal, cuturi2013sinkhorn, nutz2021entropic}. This construction avoids sampling latent response paths as in \citet{gunsilius2019bounds} and permits a joint analysis of computational approximation and statistical estimation. Simulations and an empirical application assess its computational performance and the
 informativeness of the resulting bounds.

Our statistical analysis tracks the errors introduced by discretization, support relaxation, entropy regularization, and numerical optimization. Drawing on regularization theory \citep{weed2018explicit, nutz2021entropic}, we establish convergence to the sharp bound and consistency under joint conditions on the approximation parameters and sample size (Theorem~\ref{thm:consistency}). For fixed discretization and regularization, with sufficiently accurate numerical optimization, we derive a central limit theorem and a consistent variance estimator for the regularized population value (Theorem~\ref{thm:asymptotic_distribution}). 

Our approach relates to several advances in econometric identification and computation. \citet{henry2015combinatorial} use a maximum-flow representation to avoid enumerating all inequalities defining incomplete models. \citet{torgovitsky2019partial} obtains finite linear characterizations by extending subdistributions, while \citet{gu2023partial, gu2024counterfactual} develop response-type and latent-space enumeration methods for discrete outcome models. \citet{mogstad2018using} bound policy-relevant parameters by optimizing marginal treatment response functions consistent with identified IV-like estimands; their binary-treatment framework allows continuous instruments. Related work studies multiple instruments under partial monotonicity \citep{mogstad2024policy}; \citet{mogstad2024iv} survey this literature. Our formulation instead works with structural response-path laws and the full observable conditional distributions, allowing continuous treatments.

Existing computational methods for continuous treatments include sampling response paths \citep{gunsilius2019path}, imposing selected moment restrictions \citep{kilbertus2020class}, and optimizing generative or neural representations \citep{balazadeh2022partial, tan2024consistency}. The latter can accommodate continuous variables, and \citet{tan2024consistency} establish consistency under explicit conditions, although computation involves non-convex optimization. Our method provides a convex entropic approximation whose discretization and regularization errors can be analyzed jointly with sampling error.

The connection between transport and partial identification also has important precedents \citep[see][for a broader discussion]{gunsilius2025primer}. \citet{galichon2009test, galichon2011set} relate compatibility in incomplete models to transport and random-set representations, and \citet{ekeland2010optimal} study falsifiability using optimal transport. Entropy has an established role in latent-variable econometrics: \citet{schennach2014entropic} uses entropy maximization to transform latent-variable moment restrictions into observable restrictions and simulation to integrate out the latent variables. In our construction, relative entropy regularizes an extremal functional over structurally admissible path laws, with computation based on deterministic discretization.

Recent transport approaches address complementary counterfactual and identification problems. \citet{gu2024wasserstein} study sensitivity to Wasserstein departures from a baseline latent distribution. \citet{oberreynolds2023estimating} derives transport bounds for joint potential-outcome functionals and develops estimation and bootstrap inference. \citet{fan2025partial} characterize incomplete-data moment models through conditional transport. \citet{voronin2025generalized} uses characteristic kernels and polynomial approximation to compute bounds over distributions of finite-dimensional observed and latent vectors. Closely related computationally, \citet{franguridi2025inference} combine entropic transport, Sinkhorn computation, and Gaussian and directional-bootstrap inference for partially identified moment models. Our contribution brings together structural response-function spaces, continuously indexed distributional constraints, and a convergent deterministic approximation with joint statistical and numerical error analysis. 

Also closely related, \citet{gao2025bridging} use multi-marginal optimal transport and entropic computation to bound causal functionals given finitely many identified potential-outcome marginals. Our framework accommodates distributions over entire structural response functions and a continuum of observable distributional restrictions, with approximation guarantees connecting the computational problem to the sharp bounds and a central limit theorem for estimation at fixed regularization and discretization. 

Another close contribution is \citet*{tan2026partial}, who obtain sharp transport bounds on policy-relevant effects in the generalized Roy model. They preserve full distributional information and exploit the model's selection structure to reduce the problem to separable one-dimensional transport problems. Their framework allows continuous instruments and includes treatment extensions and statistical inference. Our approach supplies a complementary numerical method for general structural response-path models under the regularity conditions developed below, without requiring such an analytical reduction and a specific Roy model setting.

Our inference results also connect to the literature on partially identified functions and optimization problems. \citet{santos2012inference} studies nonparametric IV models, and \citet{chernozhukov2023constrained} develop inference under shape restrictions in conditional moment models. \citet{fang2023inference} study large linear systems with known coefficients, while \citet{bai2026inference} allow unknown coefficients. Directionally differentiable functionals and linear-program perturbations provide other routes to inference \citep{fang2018inference, cho2024simple}. Our Gaussian limit uses the smooth dependence of the regularized transport value on the observable marginals, building on statistical theory for entropic transport \citep{delbarrio2023improved}.

Finally, the path formulation connects partial identification to trajectory inference from population snapshots \citep{schiebinger2019optimal, lavenant2024mathematical}. Both settings identify marginals while leaving dependence across coordinates of each path unobserved. Here instrument values index potential responses in place of time. After Gibbs reweighting, the regularized objective has a multi-marginal Schr\"odinger bridge interpretation \citep{schrodinger1931, schrodinger1932, leonard2013survey}. Recent trajectory methods accommodate smooth reference processes \citep{hong2025trajectory} and iterative reference refinement \citep{shen2025multimarginal}. We use this common entropy-minimization structure to approximate extremal causal functionals over paths that respect the structural model.

The leading application is to IV bounds with continuous treatments and instruments, extending formulations developed for binary treatments \citep{balke1997bounds, kitagawa2009identification}, and related problems arise in welfare analysis \citep{hausman2016individual}. Our simulations compare exact linear-program benchmarks and examine computation and inference. An application to food and leisure budget shares illustrates the role of distributional information and shape restrictions.

\textbf{Notation.}
We refer to $\N_0$ as the set containing all natural numbers and $0$, whereas $\overline{\R}$ is used to denote the real numbers with $+\infty$. For any $n\in\N$, we denote the set of all natural numbers smaller than or equal to $n$ by $[n]$, i.e. $[n]=\{1,2,\dots,n\}$. Unless otherwise stated, we reserve calligraphic letters for sets, capital letters $X,Y,Z,\Lambda$ for random variables, and their minuscule is used for elements of a set.

Given two metric spaces  $(\mathcal{X}_i,d_{\mathcal{X}_i})$, $i=1,2$, the set $C(\mathcal{X}_1,\mathcal{X}_2)$ ($C_b(\mathcal{X}_1,\mathcal{X}_2)$) stands for the set of all continuous functions (and bounded) from $\mathcal{X}_1$ to $\mathcal{X}_2$. We endow $C_b(\mathcal{X}_1,\mathcal{X}_2)$ with the sup-norm, $\|\cdot\|_\infty$. The set of all Lipschitz functions is denoted by $\textrm{Lip}(\mathcal{X}_1,\mathcal{X}_2)$. When the domain is compact, this set is endowed with the norm $\|f \|_{\textrm{Lip}}=\|f\|_\infty+\sup_{x\neq y\in\mathcal{X}_1} \frac{d_{\mathcal{X}_2}(f(x),f(y))}{d_{\mathcal{X}_1}(x,y)}$. Under this norm, the closed ball of radius $M>0$ is denoted by $\textrm{Lip}_M(\mathcal{X}_1,\mathcal{X}_2)$.

For any metric space $(\mathcal{X},d_\mathcal{X})$, the Borelian sets of $\mathcal{X}$ are referred by $\mathcal{B}(\mathcal{X})$. 
The set $\P(\mathcal{X})$ ($\mathcal{M}(\mathcal{X})$) stands for the set of all probability (signed) measures defined over $\mathcal{B}(\mathcal{X})$. We analyze this set with two different topologies: weak and $W_1$ topologies. We refer the reader to \citet{billingsley2013convergence, peyre2019computational} for their definition and properties.

\section{Setup and approach}\label{sec:setup}

\subsection{Setup}\label{sec:setup_objects}

\textbf{The baseline framework.} We observe a pair $(Z,X)$: an exogenous variable $Z$ and an endogenous variable $X$, taking values in compact Polish spaces $\Z$ and $\X\subseteq\mathcal{H}$, respectively. The underlying space $\mathcal{H}$ is assumed to be a separable Hilbert space, and $\X$ is either discrete or a convex subset.
The observed data reveal the exogenous variable's (unconditional) distribution $\lambda\in\P(\Z)$, as well as the conditional law of the endogenous variable given the realization $Z=z$, $P_z\in\P(\X)$. We write $\mathbf{P}=(P_z)_{z\in\Z}$ for this \emph{marginal profile} of the data, and denote the joint distribution of the exogenous and endogenous variables by $P=\int_{\Z}P_z\,d\lambda(z)\in \P(\Z\times \X)$.

Unobserved heterogeneity is modeled as a latent variable $\omega$ in a Polish space $\Omega$, distributed according to some non-identified law $Q\in\P(\W)$. Its interaction with the observed variable is governed by a known mechanism function $g:\Omega\to\mathbb{H}$, mapping each latent type $\omega\in\Omega$ to an $\X$-valued function of the exogenous variable, $g(\omega)\in\mathbb{H} = L^1(\Z,\lambda;\X)$. Here $L^1(\Z,\lambda;\mathcal{H})$ stands for the Bochner space of $\mathcal{H}$-valued, $\lambda$-integrable functions on $\Z$ or, equivalently, the set of $\mathcal{H}$-valued $\lambda$-integrable stochastic processes on $\Z$. Equivalently, $\mathbb{H}$ contains all stochastic processes whose value lie in the smaller set $\X$. Intuitively, $g(\omega)$ collects the potential endogenous values, as a stochastic process, that the type $\omega$ would produce at every realization of the exogenous variable. Measurability and integrability are the only regularity this path is required to satisfy, with jumps being allowed. Writing $g_z(\omega):=g(\omega)(z)$ for the path evaluated at exogenous value $z$, the endogenous variable is generated by evaluating the mechanism at the realized instrument: $X = g_Z(\omega)$.

The exogenous variable plays the role of an \emph{instrument}: it induces variation in the endogenous variable that is observed and independent of the unobserved heterogeneity. This is the \emph{strict exogeneity} condition, $Z\perp\omega$, and it may be a source of increasing in the identification power. Figure~\ref{fig:baseline} represents this structure as a directed acyclic graph (DAG).

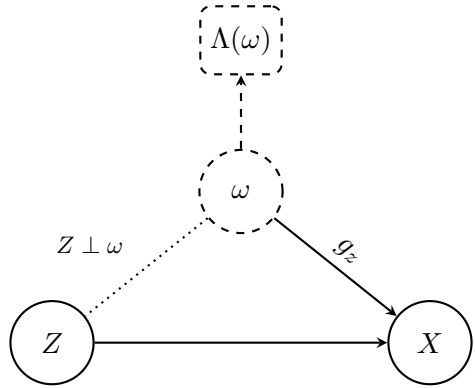
\begin{figure}[H]
\centering
\begin{tikzpicture}[
  latent/.style={circle, draw, dashed, minimum size=1.1cm},
  obs/.style={circle, draw, minimum size=1.1cm},
  target/.style={rectangle, draw, dashed, rounded corners, minimum size=0.9cm},
  >=stealth, thick
]
  \node[latent] (om) at (0,0) {$\omega$};
  \node[target] (lam) at (0,2) {$\Lambda(\omega)$};
  \node[obs] (z) at (-2.5,-2) {$Z$};
  \node[obs] (x) at (2.5,-2) {$X$};
  \draw[->, dashed] (om) -- (lam);
  \draw[->] (om) -- (x) node[midway, above, sloped] {$g_z$};
  \draw[->] (z) -- (x);
  \draw[dotted] (om) -- (z);
  \node[align=center] at (-2,-0.7) {\footnotesize $Z\perp\omega$};
\end{tikzpicture}
\caption{Solid nodes are observed, the dashed node is latent, and the dashed box marks the target. The dotted line marks the absence of an edge between $\omega$ and $Z$: exogeneity. The solid arrows into $X$ depict the mechanism, jointly determined by the latent type and the instrument.}
\label{fig:baseline}
\end{figure}

The object of interest is the mean of an effect function $\Lambda:\Omega\to[0,1]$ of the latent type, $E_Q[\Lambda]$.\footnote{The restriction to $[0,1]$ is a normalization. Any bounded effect function $\Lambda:\Omega\to[a,b]$ is accommodated by the affine change $\tilde\Lambda=(\Lambda-a)/(b-a)$, since the feasible set $\mathcal{Q}(\mathbf{P},\lambda)$ defined below does not depend on the effect function; boundedness is the only substantive requirement.} The effect function may encode counterfactual outcomes, marginal effects, or other features of the latent type. The data restrict the latent law $Q\in\P(\W)$ to the sharp identified set
\begin{equation}\label{eq:identified-set}
\mathcal{Q}(\mathbf{P},\lambda) = \big\{Q\in\P(\Omega) : (g_z)_\#Q = P_z \text{ for }\lambda\text{-a.e. } z\in\Z\big\},
\end{equation}
that is, the latent laws whose instrument-wise pushforwards reproduce every observed conditional endogenous distribution. Since this set need not be a singleton, the target $E_Q[\Lambda]$ is, in general, only partially identified.

Our main objective is the identification and estimation of the sharpest lower bound on this average effect, the value of the following infinite-dimensional linear program:
\begin{equation}\label{eq:sharp-bound}
v(\mathbf{P},\lambda) = \inf_{Q\in\mathcal{Q}(\mathbf{P},\lambda)} E_Q[\Lambda].
\end{equation}
Focusing on the lower bound is without loss of generality: the upper bound is obtained by solving the same program with $-\Lambda$ in place of $\Lambda$ and reversing the sign of the resulting value.

The key step toward computing this bound is a novel connection to optimal transport. The identification problem can be equivalently formulated as a constrained optimal transport problem, minimizing the expected cost of transporting the instrument distribution to the profile of marginal endogenous distributions. Leveraging this connection, we develop a series of simplifications to the associated optimal transport problem that yield a tractable identification and estimation routine.

Before turning to this connection, we show that covariates, often used in applied work to justify the exogeneity restriction above, can be accommodated without leaving the baseline framework. The covariate framework introduced next is just a special case in which the sharp bound reduces to a covariate-average of baseline bounds. This lets us develop the theory for the baseline framework without loss of generality.

\textbf{The covariate framework.} Suppose the researcher additionally observes a covariate vector $W$ in a compact Polish space $\mathcal{W}$, with known marginal $m\in\P(\mathcal{W})$. The covariate $W$ interacts with both the instrument and the endogenous variable. For each realization $W=w$, the instrument is conditionally distributed as the known law $\lambda_w\in\P(\Z)$, and the endogenous variable given $(Z,W)=(z,w)$ has law $P_{z,w}\in\P(\X)$. The mechanism mediating the interaction of observed and unobserved variables becomes covariate-indexed as well: for each covariate $w\in\mathcal{W}$, a map $g_w:\Omega\to\mathbb{H}$, with $g_{w,z}(\omega)\in\X$ denoting its value at exogenous realization $z$. The structural equation is then $X=g_{W,Z}(\omega)$. As with the mechanism, the effect function may as well depend on the covariate, $\Lambda:\Omega\times\mathcal{W}\to[0,1]$, for instance, to capture the effect of a program for a specific subpopulation.

A distinctive feature of this framework is that the instrument need not be exogenous unconditionally. Unlike in the baseline case, $Z$ and $\omega$ may be marginally dependent, provided the covariate $W$ is the sole channel of that dependence. Once this channel is closed by conditioning on $W$, exogeneity is restored. This is the standard assumption of \emph{conditional exogeneity} maintained in applied work: $Z\perp\omega \mid W$.

This triangular dependence, $W$ feeding into $Z$, $\omega$, and $X$, is displayed in Figure~\ref{fig:covariate}.

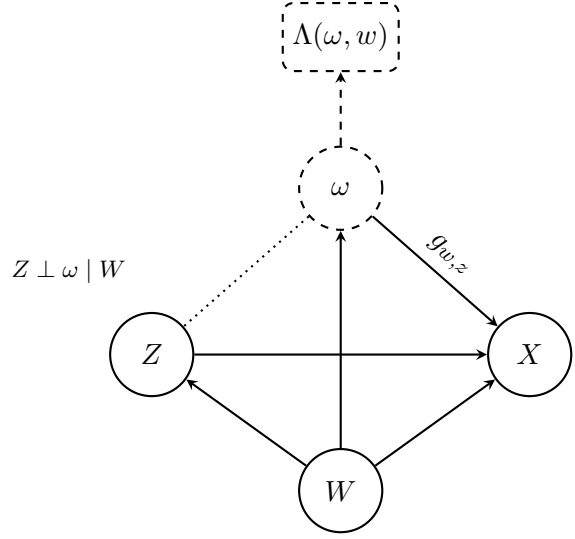
\begin{figure}[H]
\centering
\begin{tikzpicture}[
  latent/.style={circle, draw, dashed, minimum size=1.1cm},
  obs/.style={circle, draw, minimum size=1.1cm},
  target/.style={rectangle, draw, dashed, rounded corners, minimum size=0.9cm},
  >=stealth, thick
]
  \node[obs] (w) at (0,-4) {$W$};
  \node[latent] (om) at (0,0) {$\omega$};
  \node[obs] (z) at (-2.5,-2.2) {$Z$};
  \node[obs] (x) at (2.5,-2.2) {$X$};
  \node[target] (lam) at (0,2) {$\Lambda(\omega,w)$};
  \draw[->, dashed] (om) -- (lam);
  \draw[->] (om) -- (x) node[midway, above, sloped] {$g_{w,z}$};
  \draw[->] (z) -- (x);
  \draw[->] (w) -- (om);
  \draw[->] (w) -- (z);
  \draw[->] (w) -- (x);
  \draw[dotted] (om) -- (z);
  \node[align=center] at (-3.6,-1.1) {\footnotesize $Z\perp\omega\mid W$};
\end{tikzpicture}
\caption{The covariate framework. Deleting $W$ returns Figure~\ref{fig:baseline}. Both $W\to Z$ and $W\to\omega$ must be present for this to strictly generalize the baseline framework.}
\label{fig:covariate}
\end{figure}

Conditional exogeneity, together with the fact that the covariate law $m$ and the conditional instrument law $\lambda_w$ are both known, disintegrates the state law into a free part and a known part:
\[
Q(d\omega,dw,dz) = \underbrace{Q_w(d\omega)}_{\text{free}}\ \underbrace{\lambda_w(dz)m(dw)}_{\text{known}}\ ,
\]
with $Q_w$, the latent law given $W=w$, the only unknown component.

The data-matching restriction again applies fiberwise: for each covariate realization, the pushforward of the conditional latent law under the covariate-indexed mechanism must match the conditional distribution of the endogenous variable, $(g_w)_\#Q_w = P_w$ for $m$-a.e. $w$, where $P_w=\int_{\Z}P_{w,z}\,d\lambda_w(z)$ is the covariate-conditional endogenous law. The sharp identified set of state laws is therefore
\[\mathcal{Q}(\mathbf{P},\lambda) = \big\{(Q_w)_{w\in\mathcal{W}}\in\P(\Omega) : (g_w)_\#Q_w = P_w \text{ for } m\text{-a.e. } w\in\mathcal{W}\big\}.\]

The target effect disintegrates in the same way, as the integral of covariate-fiber averages,
\[E_Q[\Lambda]=\int_{\mathcal{W}}E_{Q_w}[\Lambda(\cdot,w)]\,dm(w).\]

The sharp lower bound on the average effect is therefore
\[v(\mathbf{P},\lambda)=\inf_{Q\in\mathcal{Q}(	\mathbf{P},\lambda)} E_Q[\Lambda].\]
Because both the state law and the target disintegrate fiberwise, and under mild additional regularity on the conditional profile $\mathbf{P}_w=(P_{z,w})_{z\in\Z}$, the conditional instrument law $(\lambda_w)_{w\in\mathcal{W}}$, and the effect function $\Lambda$, a measurable-selection argument shows that the sharp bound is itself a covariate average of fiberwise bounds:
\begin{equation}\label{eq:covariate-bound}
v(\mathbf{P},\lambda) = \int_{\mathcal{W}} \left(\inf_{\substack{Q_w\in \P(\W)\\ \,(g_w)_\#Q_w=P_w}} E_{Q_w}[\Lambda(\cdot,w)]\right)\,dm(w) = \int_{\mathcal{W}} v(\mathbf{P}_w,\lambda_w)\,dm(w).
\end{equation}

Framework 2 is thus Framework 1 run in every covariate fiber and averaged: conditional on $W=w$, it is the baseline problem with data $(\mathbf{P}_w,\lambda_w)$, mechanism $g_w$, and effect $\Lambda(\cdot,w)$. Deleting $W$ from Figure~\ref{fig:covariate} returns Figure~\ref{fig:baseline}. We therefore develop identification and estimation for the bound of the baseline framework and recover the covariate framework through the reduction \eqref{eq:covariate-bound}, without loss of generality under the stated regularity.

\subsection{Examples}\label{sec:examples}

\textbf{A common two-stage structure.} All examples that follow share a common structure: a two-stage model in which the latent type describes a complete potential-response profile at each stage. The first stage is a selection or treatment phase, in which the instrument influences the treatment the unit takes. The second is the outcome equation, in which the realized treatment produces the response.

Write $\mathcal{D}$ for the set of treatments or selections and $\mathcal{R}$ for the set of responses, so that the endogenous variable records both stages as the pair $X=(D,R)\in\X=\mathcal{D}\times\mathcal{R}$.

The latent type of each stage is a complete description of how the unit would respond to every instrument or treatment value, that is, a map
\[
\omega_1:\Z\to\mathcal{D} \quad\text{ and }\quad\omega_2:\mathcal{D}\to\mathcal{R}
\]
assigning to each instrument value the treatment the unit would take, and to each treatment the outcome the unit would then realize. The complete latent profile is the concatenation of the two,
\[
\omega=(\omega_1,\omega_2)\in\Omega=\Omega_1\times\Omega_2.
\]

These maps are to be read as counterfactually complete descriptions of behavior. 
A unit's type is therefore its entire response profile across all possible realizations, and the data reveal only the single coordinate that the realized instrument selects. Unobserved heterogeneity in this framework is precisely heterogeneity in these profiles, and it is unrestricted: no assumption ties how a unit responds in the first stage to how it responds in the second.

Additional assumptions enter as shape restrictions on the form of the response functions, that is, as restrictions on the sets $\Omega_1$ and $\Omega_2$ rather than on the observed data. Requiring both maps to be continuous is a natural and comparatively weak such condition. In models with continuous treatment where the targets are marginal effects, one restricts attention to continuously differentiable types, so that the derivatives defining those effects exist. Exclusions of particular response patterns, such as the removal of defiers in Example~\ref{ex:compliance} below, are of the same nature. Restrictions of this kind shrink the feasible set $\mathcal{Q}(\mathbf{P},\lambda)$ and hence tighten the bound, but leave the structure of the problem intact.

Finally, the mechanism function translates the latent type into the instrument-response profile. Given a type $\omega=(\omega_1,\omega_2)$ and an instrument realization $z$, the unit takes treatment $\omega_1(z)$ and realizes the outcome $\omega_2(\omega_1(z))$, so that the mechanism is the two-stage composition:
\begin{equation}\label{eq:mechanism-2stage}
g_z(\omega) = \big(\omega_1(z),\ \omega_2(\omega_1(z))\big)\in\mathcal{D}\times\mathcal{R}.
\end{equation}
The induced path $z\mapsto g_z(\omega)$ is thus vector-valued, tracking jointly how selection and outcome respond as the instrument ranges over $\Z$: it is the element of $\mathbb{H}$ that the type $\omega$ determines, and the data constrain only its pushforwards $(g_z)_\#Q=P_z$.

What varies across the examples that follow is the economic content and the richness of the underlying spaces. We begin with the fully binary case and progress to settings in which the instrument and the treatment are continuous. All share the two-stage structure above, with targets that are means of type-level effects, $E_Q[\Lambda]$.

\begin{example}[Imperfect compliance]\label{ex:compliance}
The instrument $Z\in\{0,1\}$ is an assigned treatment. The endogenous variable bundles the treatment received and the realized outcome, $X=(D,R)\in\{0,1\}\times\{0,1\}$. All observable spaces are binary. With no restrictions imposed, the latent type collects two response functions,
\[
\Omega=\Omega_1\times\Omega_2, \qquad \Omega_1=\Omega_2=\{0,1\}^{\{0,1\}},
\]
where $\omega_1$ is the potential-treatment map and $\omega_2$ the potential-outcome map. Each factor has four elements: never-taker ($\mathrm{NT}$), always-taker ($\mathrm{AT}$), complier ($\mathrm{C}$), defier ($\mathrm{D}$). Writing a response path as the pair of endogenous values realized under each instrument value, $(\mathrm{treat}_{z=0},\mathrm{out}_{z=0};\mathrm{treat}_{z=1},\mathrm{out}_{z=1})$, mechanism \eqref{eq:mechanism-2stage} evaluates to
\[
g(\omega) = \big(\omega_1(0),\,\omega_2(\omega_1(0));\ \omega_1(1),\,\omega_2(\omega_1(1))\big).
\]
The set of possible paths for the potential endogenous variables is $\mathbb{H}=\left(\{0,1\}\times\{0,1\}\right)^{\{0,1\}}$, i.e. the set of all maps from $\{0,1\}$ to $\{0,1\}\times\{0,1\}$, since integrability imposes no further restriction in the finite case. The distribution $\lambda$ is the assignment law, and $P_z=(g_z)_\#Q$ is the observed distribution of $(D,R)$ under assignment $z$. The identified set $\mathcal{Q}(\mathbf{P},\lambda)$ is the set of type-distributions consistent with the two observed tables, and $v(\mathbf{P},\lambda)$ is the sharp lower Balke--Pearl bound \citep{balke1994counterfactual, balke1997bounds}.

Two features of this simple framework already illustrate the core issues of the problem. First, the mechanism is generally many-to-one, which is often the source of partial identification. For instance, for a never-taker in the first-stage ($\omega_1\equiv 0$) the outcome type enters the path only through its response to non-treatment, $\omega_2(0)$: a unit whose second-stage type is analogously a never-taker (recovers under neither treatment) and a unit whose second-stage type is a complier (recovers only if treated) produce the identical path,
\[
g(\mathrm{NT},\mathrm{NT}) = g(\mathrm{NT},\mathrm{C}) = (0,0;0,0).
\]
Thus $Q$ cannot be recovered from the data: only its pushforwards $(g_z)_\#Q$ are constrained, so the average effect of $\Lambda$ is left partially identified.

Second, sign restrictions on the outcome stage may act only as shape restrictions on $\Omega$, not necessarily as point-identifying devices. Requiring treatment to have a weakly beneficial causal effect removes the outcome-defiers, the type with $\omega_2(d)=1-d$ (a treated unit fails to recover while an untreated unit recovers), by imposing $Q(\{\omega_2=\mathrm{D}\})=0$. This shrinks $\mathcal{Q}(\mathbf{P},\lambda)$, and hence the bound, without collapsing it to a point.

Taking $\Lambda(\omega)=\omega_2(1)$, the bound $v(\mathbf{P},\lambda)$ gives the smallest treated-state outcome probability $\Pr(R_{d=1}=1)$ compatible with the data. The original Balke--Pearl target, the average treatment effect $E_Q[\omega_2(1)-\omega_2(0)]$, is accommodated after renormalization and translation, taking
\[
\Lambda(\omega) = \tfrac{1}{2}\big(\omega_2(1)-\omega_2(0)+1\big)\in[0,1],
\]
and the identified set $\mathcal{Q}(\mathbf{P},\lambda)$ is unchanged.
\end{example}

\begin{example}[Returns to schooling]\label{ex:schooling}
The instrument $Z$ is a natural-experiment shifter of schooling: quarter of birth interacted with compulsory-schooling laws \citep{angrist1991compulsory}, proximity to a four-year college \citep{card1993geographic}, or exposure to a school-construction program \citep{duflo2001schooling}. The endogenous variable pairs attained schooling and earnings, $X=(D,R)$, with $D$ in a discrete ladder $\mathcal{D}$ and $R\in\mathcal{R}\subset\mathbb{R}$ log-earnings. The type $\omega=(\omega_1,\omega_2)$ consists of a schooling-response map $\omega_1:\Z\to\mathcal{D}$ and a potential-earnings map $\omega_2:\mathcal{D}\to\mathcal{R}$. The effect function is the individual return to schooling, $\Lambda(\omega)=\omega_2(1)-\omega_2(0)$, and the target parameter is the average return $E_Q[\Lambda]$. Absent restrictions on how returns $\omega_2$ co-vary with selection $\omega_1$, the average effect is only partially identified, with sharp lower bound $v(\mathbf{P},\lambda)$.
\end{example}

\begin{example}[A program under imperfect take-up]\label{ex:takeup}
Let $Z$ be a randomized offer of a program: the Oregon Health Insurance Experiment lottery \citep{finkelstein2012oregon}, a Moving to Opportunity voucher \citep{katz2001moving}, or a JTPA training slot \citep{abadie2002instrumental}. Here, strict exogeneity $Z\perp\omega$ holds by design. The endogenous variable records take-up and the downstream outcome, $X=(D,R)$, with $D\in\{0,1\}$ enrollment and $R$ the outcome (e.g. health utilization, earnings, neighborhood quality). The type factors into an enrollment-response map $\omega_1$ and a potential-outcome map $\omega_2$, with the same mechanism \eqref{eq:mechanism-2stage}, and the target is the average program effect $\Lambda(\omega)=\omega_2(1)-\omega_2(0)$. Because offers are complied with heterogeneously and effects vary across the population, the average effect is partially identified and $v(\mathbf{P},\lambda)$ delivers its sharp lower bound directly from the offer-conditional distributions of $(D,R)$, with no restriction on who complies or on how compliance co-varies with the effect.
\end{example}

\begin{example}[Demand with a continuous instrument and a continuous treatment]\label{ex:demand}
Nothing in the framework requires discreteness, and the instrument $\Z$ may be a continuum, scalar or multidimensional. Consider nonparametric demand estimation with endogenous prices. The endogenous variable is the price-quantity pair $X=(D,R)$, with $D\in\mathcal{D}\subset\mathbb{R}_+$ the log-price faced and $R\in\mathcal{R}\subset\mathbb{R}_+$ the log-quantity demanded (both normalized), and the instrument is a continuous cost shifter. In the gasoline application of \citet{blundell2012measuring,blundell2017nonparametric}, the instrument is the distance between an oil platform in the Gulf of Mexico and the capital of the household's state, a scalar continuum. Taking several shifters jointly (e.g. transport distance, refinery capacity, state excise taxes) gives a multidimensional set $\Z\subset\mathbb{R}^k$.

The latent type $\omega=(\omega_1,\omega_2)$ now consists of a price-response map $\omega_1:\Z\to\mathcal{D}$, recording the price the household faces as the cost shifter varies, and an individual demand function $\omega_2:\mathcal{D}\to\mathcal{R}$. Individual demand may be further restricted by economic theory conditions, such as monotonicity, concavity, or the Slutsky condition. 

The mechanism is the same two-stage composition \eqref{eq:mechanism-2stage}, and the response path $z\mapsto g_z(\omega)$ is now a genuine element of $\mathbb{H}=L^1(\Z,\lambda;\X)$ rather than a finite list. Because $\omega_2$ is an entire demand curve, heterogeneity is infinite-dimensional, in the spirit of the multidimensional unobserved heterogeneity considered by \citet{hausman2016individual}.

Continuity of the treatment makes marginal effects the natural targets. One may be interested in the sensitivity of demand to a price change, such as the elasticity of demand, evaluated at price $d$. This is a marginal effect of the second-stage type, and the effect function
\[
\Lambda(\omega) = \partial_d\,\omega_2(d),
\]
suitably normalized,\footnote{This target requires $\omega_2$ differentiable at $d$, i.e.\ a smoothness restriction on $\Omega$ (not on the data), consistent with the shape-restriction discussion above.} delivers the average elasticity of demand, and welfare targets such as average deadweight loss or equivalent variation from a tax-induced price change are means of the same kind. In the absence of restrictions on how the demand curve $\omega_2$ co-varies with price exposure $\omega_1$, these averages are partially identified, and $v(\mathbf{P},\lambda)$ is their sharp lower bound.
\end{example}

\textbf{Covariates in applied practice.} Strict exogeneity of the instrument is often implausible in applied work, and the standard response is to assert exogeneity only after conditioning on observed covariates. In the examples that follow, we take the resulting conditional-exogeneity restriction $Z\perp\omega\mid W$ as given, as the applied literature does, and note only that the covariate framework then delivers the bound as the covariate-average $\int_{\mathcal{W}}v(\mathbf{P}_w,\lambda_w)\,dm(w)$ of baseline bounds, so working fiber by fiber is without loss of generality. We do not take a stand on whether a given covariate vector in fact delivers conditional exogeneity in a given application. We assume it, as the applied literature does.

The examples below extend the two-stage model above with one additional feature. We now have a covariate $W$ that interacts with every observed and unobserved object in the model, so that conditioning on its realization restores independence between the instrument and the heterogeneity. Concretely, each stage type is now the $W$-fiber of a deeper, jointly-indexed map. The first-stage type $\omega_1$ is the evaluation at $W$ of $\tilde{\omega}_1:\Z\times\mathcal{W}\to\mathcal{D}$, $\omega_1=\tilde{\omega}_1(\cdot,W):\Z\to \mathcal{D}$, and the second-stage type $\omega_2$ is the evaluation at $W$ of $\tilde{\omega}_2:\mathcal{D}\times\mathcal{W}\to\mathcal{R}$, $\omega_2=\tilde{\omega}_2(\cdot,W):\mathcal{D}\to \mathcal{R}$. This representation makes precise the sense in which conditioning on $W$ restores exogeneity. All dependence between $Z$ and the type runs now through the covariate cell, so that, under the standard assumption that the instrument is as good as randomly assigned within cells, the conditional-exogeneity restriction $Z\perp\omega\mid W$ holds.

For each realization $W=w$, the mechanism is the covariate-indexed map $g_w$, the two-stage composition \eqref{eq:mechanism-2stage} evaluated at the realized covariate, so that the endogenous variable is generated by evaluating the mechanism at the realized covariate and instrument,
\[X=g_{W,Z}(\omega)=\left(\omega_1(Z),\omega_2\left(\omega_1(Z)\right)\right).\]
The target may likewise vary with the covariate, $\Lambda(\omega,w)$, and the covariate framework returns the sharp bound $v(\mathbf{P}_w,\lambda_w)$ cell by cell, averaged over $m$, with no restriction on how responsiveness to the instrument co-varies with the effect of the treatment.

\begin{example}[Examiner and judge designs]\label{ex:examiner}
A large applied literature exploits the assignment of cases to decision-makers who differ systematically in how readily they impose a treatment. Because which decision-maker is assigned to a given case is often either deliberately random, to ensure fairness, or as good as random due to idiosyncrasies of the assignment process such as shift rotations, a measure of the assigned decision-maker's leniency serves as an instrument for the decision itself. \citet{kling2006incarceration} pioneered the design, estimating the effect of incarceration length on post-release earnings using the rules that assign offenders to judges. Leading examples include \citet{dobbie2018effects}, who use variation in arraignment judges' tendency to set bail as an instrument for pre-trial detention, \citet{aizer2015juvenile} on juvenile incarceration, \citet{doyle2015measuring}, who exploit the pseudo-random assignment of ambulance companies to patients, and \citet{dobbie2017consumer} on the leniency of randomly assigned bankruptcy judges.

The instrument is unconditionally exogenous only in the rare case of a single undifferentiated assignment pool. In practice assignment is randomized within administrative cells, and examiners may be only conditionally randomly assigned. For instance, defendants charged with felonies may be assigned to a different set of judges than those charged with misdemeanors, in which case the analysis must control for whatever covariates make assignment as good as random. Applied practice reflects this fact, as leniency is typically residualized on cell dummies, such as year-by-department, to account for randomization occurring within these pools. This is exactly the conditional-exogeneity restriction $Z\perp\omega\mid W$ of the covariate framework.

The design maps into the framework directly. Let $W$ be the assignment cell -- e.g. court by time period, office by charge category, hospital by shift -- with $m$ its population distribution. Let $Z$ be the leniency of the assigned decision-maker, so that $\lambda_w$ is the within-cell distribution of leniency. Note that $Z$ is naturally continuous, since it is a rate (in fact, the leave-out leniency measure is constructed to be highly predictive of the decision while uncorrelated with case and defendant characteristics). The treatment space $\mathcal{D}$ records the decision (detained or released, sentence length, benefit granted or denied) and $\mathcal{R}$ the downstream outcome (recidivism, earnings, subsequent filings, health).

As in the general construction above, the first-stage type $\omega_1$ is the $W$-fiber of a deeper map jointly indexed by instrument and cell, recording how the unit's treatment would vary across decision-makers of every leniency level. The second-stage type $\omega_2$ is the $W$-fiber describing the unit's potential outcomes under each decision. The target is the average effect of the decision, $\Lambda(\omega,w)$, permitted to vary across cells, with no restriction on how responsiveness to leniency co-varies with the effect of the decision.
\end{example}

\begin{example}[Demand with covariates]\label{ex:demand-covariates}
Continuing Example~\ref{ex:demand}, cost shifters are plausibly exogenous only conditional on household characteristics. In this direction, \citet{blundell2012measuring} condition on household income together with controls such as population density, urbanization and demographics. Setting $W$ to those characteristics gives $Z\perp\omega\mid W$, with $g_w$ the covariate-indexed demand mechanism and $\Lambda(\omega,w)$ possibly varying with $w$, allowing, for example, price responsiveness to differ across the income distribution, as that literature emphasizes. The bound on the average elasticity is the income-weighted average of the fiberwise bounds.
\end{example}


\subsection{The method in six steps}\label{sec:approach}

Having formulated the sharp bound $v(\mathbf{P},\lambda)$ in \eqref{eq:sharp-bound}, we turn to the two questions the rest of the paper answers: how to \textbf{identify} this bound, and how to \textbf{estimate} it from data. Both answers rest on a single device: a reformulation of \eqref{eq:sharp-bound} as a constrained optimal transport problem. This reformulation replaces the abstract, unobserved latent space $\Omega$ with the concrete space of potential endogenous paths, the trajectories describing how the endogenous variable would respond to every instrument value. It is this concrete representation, rather than the original one over $\Omega$, that both identification and estimation exploit. We proceed in six steps. The first, the optimal transport connection, secures identification. The remaining five progressively reshape the resulting transport problem through discretization, a relaxation of its support constraint, entropic regularization, and endogenous space discretization, turning it into one that is computationally tractable and statistically well-behaved, culminating in an estimator whose asymptotic properties we characterize.

\textbf{Step 1: Optimal transport connection (Section~\ref{sec:ot_connection}).} The bound \eqref{eq:sharp-bound} optimizes over laws on the latent space $\Omega$, which carries no exploitable structure. The latent space is unobserved, and nothing restricts its dimension or geometry. We show that \eqref{eq:sharp-bound} is equivalent to an optimal transport problem posed over the concrete space of potential endogenous paths $\mathbb{H}=L^1(\Z,\lambda;\X)$:
\[
    u(\mathbf{P},\lambda) = \inf_{\substack{\mu\in\mathcal{P}(\mathbb{H})\\ \mu_z = P_z,\ \lambda\text{-a.e. } z\in\Z\\ \mathrm{supp}(\mu)\subseteq g(\Omega)}} E_\mu[\Psi],
\]
where the cost $\Psi(e) = \inf\{\Lambda(\omega): g(\omega)=e\}$ assigns to each achievable path the smallest effect among the latent types capable of generating it (Theorem~\ref{thm:strong-duality}). The correspondence runs both ways, as every feasible latent law $Q$ induces a feasible transport plan $\mu=g_\#Q$ of equal average value, and every optimal transport plan can conversely be disintegrated into a latent law attaining the same value in \eqref{eq:sharp-bound}, so that solving one problem solves the other. 

The gain is concrete. Rather than searching over distributions on the unobserved $\Omega$, it suffices to search over couplings of the observed marginal profile $\mathbf{P}=(P_z)_{z\in\Z}$, supported on the achievable path set $g(\Omega)$. This is an optimization over objects the researcher actually knows, and about which they may hold informed views, such as the shape restrictions of Section~\ref{sec:examples}, which translate directly into restrictions on which paths are achievable. The only genuine unknown is how the marginal profile is jointly coupled across instrument values. This is both a new interpretation of the bound, as the cost of transporting the instrument's distribution to the profile of endogenous variables, and, more importantly, the structure the remaining four steps exploit.

\textbf{Step 2: Discretization (Section~\ref{sec:discretization}).} The transport problem of Step 1 remains infinite-dimensional in two respects: the coupling $\mu$ ranges over the infinite-dimensional path space $\mathbb{H}$, and the marginal constraints $\mu_z=P_z$ are indexed by the continuum $\Z$. Both can be reduced to a single finite resolution $n\in\mathbb{N}$. Indeed, partitioning $\Z$ into $n$ cells and replacing each path by its cell-wise average defines a projection $\pi_n$ onto the $n$-dimensional space of piecewise-constant paths with at most $n$ distinct values. Discretizing the marginal constraints along the same partition, requiring only the projected coupling $\pi_n\#\mu$ to match the resulting cell-wise marginals $P_{k,n}$ and to be supported on the projected achievable set, gives the constrained optimal transport problem
\[
    u_n(\mathbf{P},\lambda) = \inf_{\substack{\mu\in\mathcal{P}(\mathbb{H})\\ (\pi_n\#\mu)_k = P_{k,n},\ k\in[n]\\ \mathrm{supp}(\pi_n\#\mu)\subseteq \pi_n(g(\Omega))}} E_\mu[\Psi_n].
\]
This discretized transport problem has a natural identification counterpart. It is itself the sharp lower bound $v_n(\mathbf{P},\lambda)$ of the partially identified average effect of $\Lambda$ when the mechanism $g$ is replaced by the discretized mechanism $\pi_n\circ g$ and the marginal profile is replaced consistently by $(P_{k,n})_{k\in[n]}$, bridging the discretization to the coarsening the model requirements, not only a numerical device (Proposition~\ref{prop:discretized-ot-pi}).

\textbf{Step 3: Relaxing the support constraint (Section~\ref{sec:moment_penalization}).} The discretized problem of Step 2 still carries the requirement that the projected coupling be supported on the achievable set $\pi_n(g(\Omega))$, a constraint that is numerically costly to enforce and interacts poorly with the entropic penalty introduced next, whose minimizer has full support. We remove it by penalizing incompatible paths directly in the cost rather than excluding them by constraint. For a control parameter $\delta>0$, the regularized cost $\Psi_{\delta,n}(e) = \inf_{\omega\in\Omega}\{\Lambda(\omega) + \delta^{-1}\|\pi_n(e)-\pi_n(g(\omega))\|_{\mathbb{H}}\}$ charges every path, achievable or not, the cost of the latent type that comes closest to generating it, plus a penalty proportional to that shortfall. This Tikhonov-style regularization yields an \emph{unconstrained} transport problem,
\[
    u_{\delta,n}(\mathbf{P},\lambda) = \inf_{\substack{\mu\in\mathcal{P}(\mathbb{H})\\ (\pi_n\#\mu)_k = P_{k,n},\ k\in[n]}} E_\mu[\Psi_{\delta,n}],
\]
easier to implement in practice, whose value $u_{\delta,n}(\mathbf{P},\lambda)\le u_n(\mathbf{P},\lambda)$ converges to the constrained one as $\delta\to0$. Beyond tractability, penalizing rather than constraining is what makes the objective Lipschitz in the path, a property the original cost $\Psi_n$ satisfies only under restrictive conditions, and it is this Lipschitz continuity that proves pivotal for the dual analysis and for the asymptotic distribution of the estimator in Step 5.

\textbf{Step 4: Entropic regularization (Section~\ref{sec:entropic_perturbation}).} The unconstrained problem of Step 3 is, for fixed $\delta$ and $n$, still a linear program, potentially large-scale and, as high-dimensional linear programs typically are, possessing many solutions. Penalizing deviations from a reference coupling by a relative entropy term turns it into a strictly convex problem with a unique solution. For a control parameter $\varepsilon>0$,
\[
    u_{\varepsilon,\delta,n}(\mathbf{P},\lambda) = \inf_{\substack{\mu\in\mathcal{P}(\mathbb{H})\\ (\pi_n\#\mu)_k = P_{k,n},\ k\in[n]}} E_\mu[\Psi_{\delta,n}] + \varepsilon\, H(\mu\,|\,R_{n}),
\]
where $R_{n}$ is the independent coupling of the discretized marginals $(P_{k,n})_{k\in[n]}$ and $H(\cdot\,|\,R_{n})$ is the relative entropy, or Kullback--Leibler divergence, with respect to it. This is a multi-marginal entropic optimal transport problem \citep{cuturi2013sinkhorn, nutz2021entropic}, a currently very active area of applied mathematics and statistics whose tools this reformulation lets us bring to bear on the identification problem. As $\varepsilon\to0$, the value $u_{\varepsilon,\delta,n}$ decreases to the value $u_{\delta,n}$ of Step 3.

Applying duality arguments to the entropic regularized problem delivers more than existence, identifying the unique solution explicitly. Based on it, we show that there exist dual potentials $(\phi_1,\ldots,\phi_n)$, functions of the endogenous value alone, solving a system of first-order conditions, such that the optimal coupling is the Gibbs measure with log-density proportional to $\varepsilon^{-1}\left(\sum_k\phi_k(e_k)\lambda_{k,n}-\Psi_{\delta,n}(e)\right)$ with respect to $R_n$. The same potentials recover the value of the program from their expectation under the marginals. This dual characterization is also what makes the problem computable without sampling through its first-order system, solved by iteratively updating each potential in turn, the Sinkhorn algorithm.

\textbf{Step 5: Discretizing the endogenous space (Section~\ref{sec:x_discretization}).} The problem of Step 4 is computable in principle, but the marginals $P_{k,n}$ entering it, and the potentials $(\phi_1,\ldots,\phi_n)$ solving it, are still, in general, continuous objects on $\X$. This turns the problem into an infinite-dimensional optimization whose plug-in analog is subject to overfitting issues, as its solution is given by a system of $N$ equations for $N$ unknowns within a sample of size $N$. We close this last gap by partitioning $\X$ itself into $m$ cells, controlling the effective number of equations/solutions, obtained by collapsing each marginal $P_{k,n}$ onto the finitely supported measure obtained by moving all its mass within a cell to a single representative point. The resulting entropic optimal transport problem is a genuine finite-dimensional problem, whose solution is determined by a system of $n\times m$ equations and $n\times m$ unknowns. This systems is solvable by the same Sinkhorn iteration as Step 4, and converging to $u_{\varepsilon,\delta,n}(\mathbf{P},\lambda)$ as $m\to\infty$. As numerically we generally do not solve a system of non-linear equations exactly, we introduce a final estimate based on the number of iterations of the Sinkhorn Algorithm, $I\in\N$. This is the deterministic, population-level counterpart of the sampling step of Step 6 below, and the two share the same structure.

\textbf{Step 6: Estimation from data (Section~\ref{sec:estimation}).} The value in Step 5 is still a population object, defined through the finitely supported and regularized marginals, $P_{k,n,m}$. Replacing them with their empirical counterparts from a sample of size $N$ gives the plug-in estimator $\hat{v}_N:=\hat{u}_{\varepsilon,\delta,n,m, I}$ that is actually computed in practice. We show in this section how to construct the marginal estimates, organize the data, and implement the Sinkhorn Algorithm to estimate $\hat{v}_N$. Because the five controls $\beta=(\varepsilon,\delta,n,m,I)$ each introduce their own approximation error on top of sampling error, consistency requires choosing all five jointly as functions of $N$. We derive $L^1$  consistency of this estimator towards the true value $v(\mathbf{P},\lambda)$, by suitably choosing $\beta$ as function of the sample size $N$, in Theorem~\ref{thm:consistency}.

When, in addition, $(\varepsilon,\delta,n,m)$ is held fixed, the estimator is shown to be asymptotically normal at the rate $\sqrt{N}$ by correctly selecting the number of Sinkhorn iterations $I$.  The asymptotic variance is also characterized, given in closed form by the instrument weighted variance of the dual potentials of Step 5 under their respective marginals (Theorem~\ref{thm:asymptotic_distribution}), and consistently estimated from their empirical counterparts. This delivers confidence intervals and hypothesis tests for the sharp lower bound $v(\mathbf{P},\lambda)$ constructed directly from data, discussed in Section~\ref{sec:discussion}, together with the limitations and extensions of the method. 

Simulations in Section~\ref{sec:simulations} illustrate the finite-sample performance of the estimator and the coverage of the confidence intervals. An application in Section~\ref{sec:application} demonstrates the method in practice.

\section{From infinite-dimensional linear programs to the multi-marginal entropic optimal transport}\label{sec:EOT}

This section establishes the connection between the partially identified bound \eqref{eq:sharp-bound} and optimal transport, and shows how a sequence of controlled regularizations turns it into a multi-marginal entropic optimal transport problem that can be solved without sampling. To keep the exposition light, we introduce the main objects only informally and state the central results. The precise definitions are collected in Appendix~\ref{app:aux}, whereas all proofs are collected in Appendix~\ref{app:proofs_all}. The development proceeds in five steps: (i) we recast the bound as a \emph{constrained} optimal transport problem on the space of potential endogenous paths (Section~\ref{sec:ot_connection}); (ii) we project this infinite-dimensional problem onto a finite resolution $n$ through a single projection map $\pi_n$ (Section~\ref{sec:discretization}); (iii) we relax the requirement that the transport plan be supported on achievable paths through an objective function regularization that penalizes deviations from the feasible set (Section~\ref{sec:moment_penalization}); (iv) we add a relative-entropy penalty, obtaining an entropic optimal transport problem amenable to the Sinkhorn algorithm (Section~\ref{sec:entropic_perturbation}); and (v) we discretize the endogenous space $\X$ itself, collapsing each marginal onto a finite grid so that the entropic problem becomes a genuine finite-dimensional problem (Section~\ref{sec:x_discretization}). The subsequent sections then use this representation to compute the bound in practice and to derive its statistical properties.

\subsection{Constrained Optimal Transport Bound}\label{sec:ot_connection}

In this section we establish the connection between the lower bound on the partially identified average, \eqref{eq:sharp-bound}, and an associated constrained optimal transport problem posed over the space of potential endogenous paths, in the sense of \citet{EkrenSonner}. The connection replaces optimization over laws on the abstract, unobserved latent space $\Omega$ with optimization over couplings of the observed marginal profile $\mathbf{P}=(P_z)_{z\in\Z}$. This represents a genuine simplification, since the latter are concrete objects the researcher can reason about directly, while the former is not even observed. We build the connection in three steps: we first express the data-matching restriction defining the identified set $\mathcal{Q}(\mathbf{P},\lambda)$ through a linear operator. We then show that this restriction, transported through the mechanism, characterizes a corresponding set of feasible path distributions. Finally, we define a cost function on paths and show, by means of an appropriate right-inverse of the mechanism function, that the resulting constrained optimal transport problem is equivalent, in value, to the original bound.

It is convenient to express the data-matching restriction defining $\mathcal{Q}(\mathbf{P},\lambda)$ through a linear operator. Let $T^\ast$ denote the adjoint operator mapping a distribution $\mu\in\mathcal{P}(\mathbb{H})$ over endogenous paths to the joint distribution it induces over the observable variables $(Z,X)$. Informally, for any test function $\phi$,
\[E_{T^\ast\mu}[\phi]=E_{\mu}\Big[\int_\Z \phi(z,e_z)\,d\lambda(z)\Big],\]
where $e_z$ denotes the value of the path $e\in\mathbb{H}$ at $z$. The convenience of this operator is that it lets us restate the data-matching restriction compactly: a latent law $Q\in\mathcal{P}(\Omega)$ reproduces the observed data whenever the pushforward of $Q$ through the mechanism, mapped forward again by $T^\ast$, equals the observed joint law $P$, that is, $T^\ast(g_\#Q)=P$. The pre-adjoint operator $T$ admits a complementary interpretation: for a function $\phi$ of the exogenous and endogenous variables and a path distribution $\mu\in\mathcal{P}(\mathbb{H})$, the value $T\phi(e)$ records the conditional expectation of $\phi$ under $h_\#(\lambda\otimes \mu)$ given that the endogenous path equals $e$, where $h:\Z\times \mathbb{H}\to \Z\times \X$, $h(z,e)=(z,e_z)$, is the map that pairs each instrument realization with the corresponding observation of the path. 

Adopting this convention throughout, the identified set of Section~\ref{sec:setup_objects} reads
\begin{equation}\label{eq:identified-set-adjoint}
\mathcal{Q}(\mathbf{P},\lambda) = \big\{Q\in\mathcal{P}(\Omega): T^\ast(g_\#Q)=P\big\},
\end{equation}
an equivalent, and for what follows more informative, restatement of \eqref{eq:identified-set}.

Note that the mechanism function pushes forward any feasible latent-type distribution $Q\in\mathcal{Q}(\mathbf{P},\lambda)$ to a distribution over potential endogenous paths, $\mu:=g_\#Q\in\mathcal{P}(\mathbb{H})$, recording the likelihood of occurrence of each potential path. Importantly, this path distribution is not unrestricted. Because it is the pushforward of a feasible latent-type distribution, it inherits two restrictions, one on its support and one on its marginals. Denote by $\mathbb{K}:=g(\Omega)$ the set of potential endogenous paths achievable by the mechanism function. Since $\mu$ is the pushforward of $Q$ under $g$, its support is automatically contained in $\mathbb{K}$: no unachievable path can receive positive probability. The second restriction comes from the requirement that the $z$-marginals of $\mu$ match the observed data, $P_z$, for $\lambda$-a.e.\ $z$. As just described, this condition is expressed compactly through the adjoint operator, as $T^\ast\mu=P$. These two restrictions define the set of feasible potential-endogenous-path distributions,
\begin{align}
\mathcal{M}(\mathbf{P},\lambda) &= \big\{\mu\in\mathcal{P}(\mathbb{H}): T^\ast\mu=P \text{ and } \mathrm{supp}(\mu)\subseteq\mathbb{K}\big\}\notag\\
&= \big\{\mu\in\mathcal{P}(\mathbb{H}): \mu_z=P_z\ \lambda\text{-a.e.}\ z, \text{ and } \mathrm{supp}(\mu)\subseteq\mathbb{K}\big\}.\label{eq:M-set}
\end{align}
By construction, every feasible latent law $Q\in\mathcal{Q}(\mathbf{P},\lambda)$ is mapped by $g$ into a path distribution $\mu=g_\#Q\in\mathcal{M}(\mathbf{P},\lambda)$. The converse also holds under mild structural conditions: if the mechanism admits a measurable right-inverse $i:\mathbb{K}\to\Omega$ satisfying $g(i(e))=e$ for every achievable path $e\in\mathbb{K}$, then every feasible path distribution $\mu\in\mathcal{M}(\mathbf{P},\lambda)$ can likewise be mapped back to a feasible latent law.

As just discussed, the mechanism function naturally transports feasible latent laws into feasible path distributions. But how should we assess the effect of an endogenous path, when the effect function $\Lambda$ is defined only on types? If the mechanism were invertible, the answer would be immediate: assign to each achievable path $e\in\mathbb{K}$ the effect of the unique type that generates it, $\Lambda\circ g^{-1}(e)$. The core difficulty of partial identification, however, is precisely that the mechanism is not invertible. Many latent types are typically mapped to the same endogenous path, so this naive effect function is ill-defined.

When the mechanism fails to be injective, there is nonetheless an intuitive way to define an effect function on paths that reduces to the one above whenever the mechanism happens to be invertible. Because every effect is a number in $[0,1]$, the worst effect among the types generating a given achievable path is always well-defined. This number can be read as ``the worst outcome we could expect from a type mapped to this path.'' Moreover, since feasibility of latent laws is determined only by pushforwards of the mechanism, at an optimum of the lower bound problem, the least favorable feasible latent law can give positive probability only to the worst types consistent with each achievable path. If that was not the case, mass could be shifted toward a worse type without violating feasibility, strictly lowering the average effect. Together with the fact that, under mild conditions, every feasible path distribution corresponds to a feasible latent law, this delivers that the worst-case effect is a sensible measure of effect over paths, and minimizing its expectation over feasible path distributions is equivalent to minimizing the average effect over feasible latent laws. These are the ideas behind the connection between the lower bound and the constrained optimal transport problem developed below.

Formally, define the cost function $\Psi:\mathbb{H}\to\mathbb{R}_+$, assigning to each achievable path the worst effect among the latent types that generate it, and a maximal cost to every other path\footnote{Since every feasible path distribution assigns zero probability to unachievable paths, the value that $\Psi$ takes there is immaterial to the expected cost. }:
\begin{equation}\label{eq:Psi}
\Psi(e) = \begin{cases}
    \inf\{\Lambda(\omega): g(\omega)=e\}, & \text{if } e\in \mathbb{K},\\
    \sup_{\omega\in\Omega}\Lambda(\omega)+1, & \text{otherwise.}
\end{cases}
\end{equation}

The choice of cost \eqref{eq:Psi} is transparent in the binary instrumental variable model of Example~\ref{ex:compliance} \citep{balke1994counterfactual, balke1997bounds}. For instance, in this example, we have seen that types $(\mathrm{NT},\mathrm{NT})$ and $(\mathrm{NT},\mathrm{C})$ were both mapped into the same path. As treatment effect depends on second-stage type only, the treatment effect of each type is different: the former effect is $0$, while the latter is $1$. Now, because the bound is a worst-case over all latent distributions consistent with the data, and the data depends only on paths constraints, an optimal distribution placing positive mass in the path generated by these types must be concentrated in the most adverse one, $(\mathrm{NT},\mathrm{NT})$. Consequently, the effective cost assigned to $e$ is the \emph{smallest} value of $\Lambda$ among the types compatible with it, and the optimal coupling then concentrates mass on the least favorable type for each achievable path.

Associated with this cost, we define the constrained optimal transport problem
\begin{equation}\label{eq:OT_value}
    u(\mathbf{P},\lambda) = \inf_{\mu\in\mathcal{M}(\mathbf{P},\lambda)} E_\mu[\Psi].
\end{equation}
This problem is posed entirely in terms of objects the researcher knows, and it represents the constrained optimal transport problem over the space of potential endogenous paths. We interpret it in the following way. The objective of the researcher is to determine the feasible endogenous path distribution (a coupling of $z$ and $x$) compatible with the data that minimizes the expected cost of transporting the exogenous distribution $\lambda$ to the marginal profile $	\mathbf{P}$, where the cost of each path is given by the worst-case effect. 

To formalize this idea, we impose a set of mild regularity assumptions on the objects of the baseline framework, maintained throughout the paper.

\begin{assumption}\label{ass:PI}
    The following holds true.
    \begin{enumerate}
        \item The sets $\Omega$ and $\Z$ are compact metric spaces. Moreover, the set $\mathcal{H}$ is a Hilbert space, and $\X\subseteq \mathcal{H}$ is either a compact, convex subset or a finite set.
        \item The mechanism function $g:\Omega\rightarrow \mathbb{H}$ is continuous.
        \item The objective function $\Lambda:\Omega\rightarrow [0,1]$ is continuous and strictly positive.
        \item The feasible set of latent models $\mathcal{Q}(\mathbf{P},\lambda)$ is non-empty.
    \end{enumerate}
\end{assumption}

These conditions impose only simple regularity on the underlying spaces and functions. In particular, requiring the mechanism to be continuous in the latent type does not require the induced paths themselves to be continuous functions of the instrument as jumps in $z$ are allowed, so long as small perturbations of the type do not produce large jumps in the induced path, as measured in the norm of $\mathbb{H}$. Likewise, requiring $\Lambda$ to be strictly positive is a harmless renormalization of boundedness, meant only to keep the bound from being trivially zero. Together with the feasibility condition of item 4, which merely restates that the data are compatible with the model, these assumptions ensure that the problem is worth studying, rather than a vacuous one.

Under these conditions, the constrained optimal transport problem and the lower bound on the partially identified average are, in a precise sense, the same problem. This happens for two reasons. First, as the following proposition shows, the feasible sets of the two problems are in a direct correspondence:

\begin{proposition}\label{prop:co-mechanism}
    Under Assumption~\ref{ass:PI}, there exists a measurable right-inverse $i:\mathbb{K}\to \Omega$ of the mechanism function such that:
    \[\Lambda(i(e))=\Psi(e),\quad \forall e\in\mathbb{K}.\]
    Moreover, its pushforward map $i_\#$ maps feasible path distributions into feasible latent laws, while the pushforward map $g_\#$ maps feasible latent laws onto feasible path distributions, that is:
    \[i_\#(\mathcal{M}(\mathbf{P},\lambda))\subseteq \mathcal{Q}(\mathbf{P},\lambda)\quad \text{and}\quad g_\#(\mathcal{Q}(\mathbf{P},\lambda))=\mathcal{M}(\mathbf{P},\lambda).\]
\end{proposition}

\begin{proof}
    See Appendix~\ref{app:proof_co_mechanism}.
\end{proof}

This proposition shows that a ``worst-effect'' latent type can be selected, measurably, for each achievable endogenous path, by means of an appropriate right-inverse of the mechanism function. We refer to this map as the \emph{co-mechanism}, since it undoes the mechanism while preserving the properties of the original problem. In fact, as the second half of the result also shows, the co-mechanism gives a natural way to map feasible path distributions back to feasible latent laws, a map that is reversed by the mechanism itself.

Second, the co-mechanism also lets us conclude that the values of both programs coincide, and that the mechanism and co-mechanism pushforwards map optimal solutions of one problem onto those of the other, creating a conjugacy between them. To see this, let $\mathcal{M}^\ast(\mathbf{P},\lambda)\subseteq \mathcal{M}(\mathbf{P},\lambda)$ and $\mathcal{Q}^\ast(\mathbf{P},\lambda)\subseteq\mathcal{Q}(\mathbf{P},\lambda)$ denote the sets of optimal solutions to the respective problems. Both results are stated in the following theorem.

\begin{theorem}\label{thm:strong-duality}
Under Assumption~\ref{ass:PI}, the value of the lower bound on the partially identified average \eqref{eq:sharp-bound} is equal to the value of the constrained optimal transport problem \eqref{eq:OT_value}:
\[v(\mathbf{P},\lambda) = u(\mathbf{P},\lambda).\]
Moreover, the optimal solutions are conjugate by the mechanism and co-mechanism pushforwards:
\[g_\#(\mathcal{Q}^\ast(\mathbf{P},\lambda))=\mathcal{M}^\ast(\mathbf{P},\lambda)\quad \text{ and }\quad i_\#(\mathcal{M}^\ast(\mathbf{P},\lambda))=\mathcal{Q}^\ast(\mathbf{P},\lambda)\cap\{Q\in\P(\W): \mathrm{supp}(Q)\subseteq i(\mathbb{K})\}.\]
\end{theorem}

\begin{proof}
    See Appendix~\ref{app:proof_strong_duality}.
\end{proof}

Theorem~\ref{thm:strong-duality} is the central connection underlying the paper. Rather than searching over laws on the unobserved space $\Omega$, it suffices to search over couplings of the observed marginal profile $\mathbf{P}=(P_z)_{z\in\Z}$ supported on the achievable path set $\mathbb{K}$. Moreover, the theorem shows that the optimal solutions of both problems are conjugate by the mechanism and co-mechanism pushforwards, so that the least favorable latent law supported on $i(\mathbb{K})$ can be recovered from the least favorable path distribution, and vice versa. Figure~\ref{fig:conjugacy} summarizes this relationship among the spaces, the feasible sets, the optimal sets, and the effect functions.

\begin{figure}[H]
\centering
\begin{minipage}{\textwidth}
\centering
\begin{subfigure}[b]{0.48\textwidth}
\centering
\begin{tikzpicture}[>=stealth, thick]
  \node (L) at (0,0) {$\Omega$};
  \node (R) at (6,0) {$\mathbb{K}\subseteq\mathbb{H}$};
  \draw[->] (L) to[bend left=35] node[midway, above] {$g$} (R);
  \draw[->] (R) to[bend left=35] node[midway, below] {$i$} (L);
\end{tikzpicture}
\caption{Spaces}
\end{subfigure}
\hfill
\begin{subfigure}[b]{0.48\textwidth}
\centering
\begin{tikzpicture}[>=stealth, thick]
  \node (L) at (0,0) {$\mathcal{Q}(\mathbf{P},\lambda)$};
  \node (R) at (6,0) {$\mathcal{M}(\mathbf{P},\lambda)$};
  \draw[->] (L) to[bend left=35] node[midway, above] {$g_\#$} (R);
  \draw[->] (R) to[bend left=35] node[midway, below] {$i_\#$} (L);
\end{tikzpicture}
\caption{Feasible sets}
\end{subfigure}

\vspace{1.5em}

\begin{subfigure}[b]{0.48\textwidth}
\centering
\begin{tikzpicture}[>=stealth, thick]
  \node (L) {$\mathcal{Q}^\ast(\mathbf{P},\lambda)$};
  \node (R) at ([xshift=4.2cm]L) {$\mathcal{M}^\ast(\mathbf{P},\lambda)$};
  \draw[->] (L) to[bend left=35] node[midway, above] {$g_\#$} (R);
  \draw[->] (R) to[bend left=35] node[midway, below] {$i_\#$} (L);
\end{tikzpicture}
\caption{Optimal sets}
\end{subfigure}
\hfill
\begin{subfigure}[b]{0.48\textwidth}
\centering
\begin{tikzpicture}[>=stealth, thick]
  \node (L) {$\Lambda$};
  \node (R) at ([xshift=4.2cm]L) {$\Psi$};
  \draw[->] (L) to[bend right=35] node[midway, below] {$i$} (R);
  \draw[->] (R) to[bend right=35] node[midway, above] {$g$} (L);
\end{tikzpicture}
\caption{Effect functions}
\end{subfigure}
\end{minipage}
\caption{The conjugacy between the latent and path-space sides of the problem. Top arrows push forward via the mechanism (or its pushforward map). Bottom arrows recover the latent-space side via the co-mechanism $i$ (or its pushforward), which in panel (d) goes in the opposite direction.}
\label{fig:conjugacy}
\end{figure}
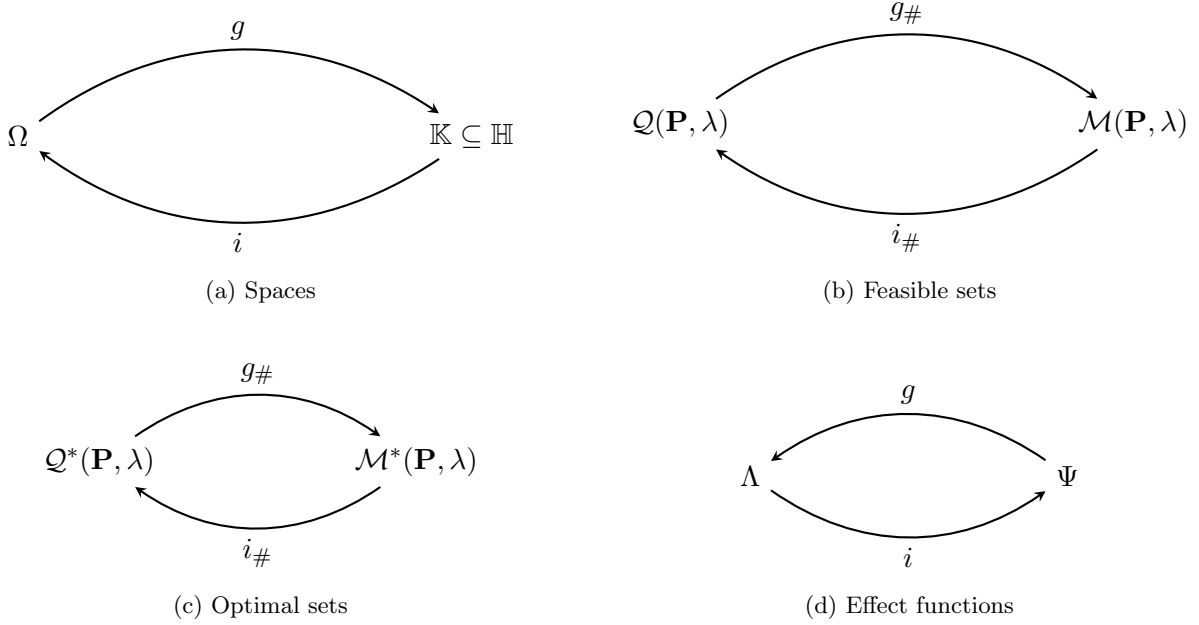

The above connection is more than an interpretive gain. Because the transport formulation is posed over concrete, finite- or function-valued objects, it can be regularized in ways that have no natural counterpart for the original problem over $\Omega$. The sections that follow discretize the path space, relax the support constraint $\mathrm{supp}(\mu)\subseteq\mathbb{K}$, and add an entropic penalty, each step trading a controlled approximation error for tractability. The result, developed over the remainder of this section, is a computable and, ultimately, statistically well-behaved route to identifying and estimating the sharp lower bound \eqref{eq:sharp-bound}.

\subsection{Discretization}\label{sec:discretization}

In many applied settings, instruments arise as genuinely continuous quantities (e.g. distance, time, tax rates, capacity) rather than a handful of discrete categories. When $\Z$ is a continuum, both sides of the constrained transport problem \eqref{eq:OT_value} inherit an infinite character from it. First, because a potential endogenous path $e\in\mathbb{H}$ records a value for every instrument realization, so the transport plan $\mu$ ranges over distributions on the infinite-dimensional space $\mathbb{H}$. Second, as the marginal profile $\mathbf{P}=(P_z)_{z\in\Z}$ records a conditional endogenous distribution for every instrument realization as well, so the data-matching constraint $T^\ast\mu=P$ is itself indexed by the continuum $\Z$. The uncountable support of the instrument thus burdens identification and estimation along two distinct dimensions, the infinite dimensionality of paths, and the uncountability of the marginal constraints.

A natural remedy is to approximate the instrument by a discretized version of itself, taking no more than $n$ distinct values that serve as a control for its resolution. The simplest such approximation partitions the instrument space into $n$ disjoint cells, $\mathcal{A}_n=(A_{k,n})_{k\in[n]}$, each visited with positive probability, and represents every instrument realization within a cell by a single value.

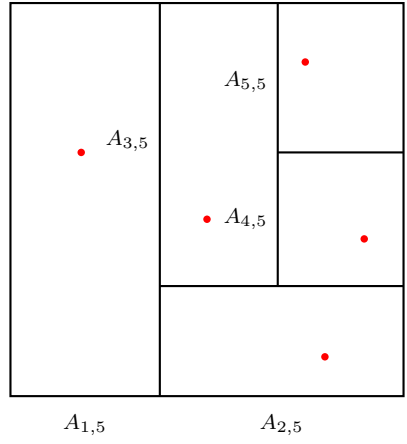
\begin{figure}[H]
\centering
\begin{tikzpicture}[x=5.2cm,y=5.2cm]
  \draw[thick] (0,0) rectangle (1,1);
  \draw[thick] (0.38,0) -- (0.38,1);
  \draw[thick] (0.38,0.28) -- (1,0.28);
  \draw[thick] (0.68,0.28) -- (0.68,1);
  \draw[thick] (0.68,0.62) -- (1,0.62);
  \fill[red] (0.18,0.62) circle (1.4pt);
  \fill[red] (0.80,0.10) circle (1.4pt);
  \fill[red] (0.50,0.45) circle (1.4pt);
  \fill[red] (0.90,0.40) circle (1.4pt);
  \fill[red] (0.75,0.85) circle (1.4pt);
  \node[below] at (0.19,-0.02) {\scriptsize $A_{1,5}$};
  \node[below] at (0.69,-0.02) {\scriptsize $A_{2,5}$};
  \node[left] at (0.38,0.65) {\scriptsize $A_{3,5}$};
  \node[left] at (0.68,0.45) {\scriptsize $A_{4,5}$};
  \node[left] at (0.68,0.80) {\scriptsize $A_{5,5}$};
\end{tikzpicture}
\caption{Discretizing a bivariate instrument $Z\in[0,1]^2$. The instrument space is partitioned into $n=5$ cells of positive probability, each represented by a single highlighted value, not necessarily its center.}
\label{fig:instrument-partition}
\end{figure}

This coarsening bears on both paths and marginal profiles. On paths, it projects each endogenous path into the set of piecewise-constant functions. Indeed, if along every cell of the partition the discretized instrument no longer varies, the path it induces should not vary either. A constant that represents well the endogenous value while the instrument is confined to a cell is its conditional average over that cell. This leads to the replacement of the path $e\in\mathbb{H}$ by its projection $\pi_n(e)$,
\[\pi_n(e)_z=\frac{1}{\lambda(A_{k,n})}\int_{A_{k,n}}e_{z'}\,d\lambda(z'),\qquad\text{for all } z\in A_{k,n}.\]
Section~\ref{subsec:projections} states this projection precisely and collects its properties. The effect of different resolutions in the projection is illustrated below:

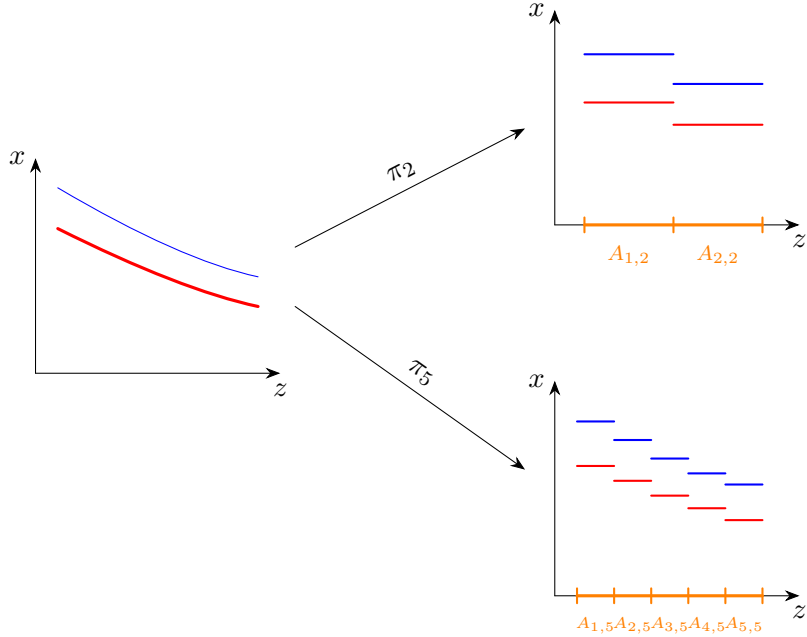
\begin{figure}[H]
\centering
\adjustbox{max width=\textwidth, max height=\dimexpr\pagegoal-\pagetotal-4\baselineskip\relax}{%
\begin{tikzpicture}[>={Stealth[length=2mm]}, line cap=round]

\begin{scope}[shift={(0,0)}]
  \draw[->] (0,-1.3) -- (0,1.6) node[left]{$x$};
  \draw[->] (0,-1.3) -- (3.3,-1.3) node[below]{$z$};
  \draw[blue] (0.3,1.2) .. controls (1.4,0.55) and (2.3,0.15) .. (3.0,0.0);
  \draw[red, very thick] (0.3,0.65) .. controls (1.4,0.1) and (2.3,-0.25) .. (3.0,-0.4);
\end{scope}

\begin{scope}[shift={(7,2)}]
  \draw[->] (0,-1.3) -- (0,1.6) node[left]{$x$};
  \draw[->] (0,-1.3) -- (3.3,-1.3) node[below]{$z$};
  \draw[blue, thick] (0.4,1.0)--(1.6,1.0) (1.6,0.6)--(2.8,0.6);
  \draw[red,  thick] (0.4,0.35)--(1.6,0.35) (1.6,0.05)--(2.8,0.05);
  \draw[orange, very thick] (0.4,-1.3)--(2.8,-1.3);
  \foreach \x in {0.4,1.6,2.8} \draw[orange, thick] (\x,-1.4)--(\x,-1.2);
  \node[below, orange, font=\scriptsize] at (1.0,-1.45){$A_{1,2}$};
  \node[below, orange, font=\scriptsize] at (2.2,-1.45){$A_{2,2}$};
\end{scope}

\begin{scope}[shift={(7,-3)}]
  \draw[->] (0,-1.3) -- (0,1.6) node[left]{$x$};
  \draw[->] (0,-1.3) -- (3.3,-1.3) node[below]{$z$};
  \draw[blue, thick]
    (0.3,1.05)--(0.8,1.05) (0.8,0.80)--(1.3,0.80) (1.3,0.55)--(1.8,0.55)
    (1.8,0.35)--(2.3,0.35) (2.3,0.20)--(2.8,0.20);
  \draw[red, thick]
    (0.3,0.45)--(0.8,0.45) (0.8,0.25)--(1.3,0.25) (1.3,0.05)--(1.8,0.05)
    (1.8,-0.12)--(2.3,-0.12) (2.3,-0.28)--(2.8,-0.28);
  \draw[orange, very thick] (0.3,-1.3)--(2.8,-1.3);
  \foreach \x in {0.3,0.8,1.3,1.8,2.3,2.8} \draw[orange, thick] (\x,-1.4)--(\x,-1.2);
  \node[below, orange, font=\tiny] at (0.55,-1.45){$A_{1,5}$};
  \node[below, orange, font=\tiny] at (1.05,-1.45){$A_{2,5}$};
  \node[below, orange, font=\tiny] at (1.55,-1.45){$A_{3,5}$};
  \node[below, orange, font=\tiny] at (2.05,-1.45){$A_{4,5}$};
  \node[below, orange, font=\tiny] at (2.55,-1.45){$A_{5,5}$};
\end{scope}

\draw[->] (3.5,0.4) -- (6.6,2.0) node[midway, above, sloped]{$\pi_2$};
\draw[->] (3.5,-0.4) -- (6.6,-2.6) node[midway, above, sloped]{$\pi_5$};

\end{tikzpicture}%
}
\caption{Projecting a path (left) onto the piecewise-constant paths of resolution $n=2$ (top right) and $n=5$ (bottom right).}
\label{fig:path-regularization}
\end{figure}

The same coarsening of the instrument affects the conditional distributions of the endogenous variable. If the instrument is observed only to lie in a cell $A_{k,n}$, the relevant endogenous distribution is the regularized marginal
\[P_{k,n}=\frac{1}{\lambda(A_{k,n})}\int_{A_{k,n}}P_z\,d\lambda(z),\qquad k\in[n].\]
The resulting marginal profile $(P_{k,n})_{k\in[n]}$ has joint law $P_n\in\P(\Z\times\X)$, obtained by averaging this profile over $\Z$ against $\lambda$, exactly as $P$ is obtained from $\mathbf{P}$. Figure~\ref{fig:regularizing-marginal} illustrates the effect of this regularization on a joint density of instrument and endogenous values.

\begin{figure}[H]
\centering
\begin{tikzpicture}
\begin{axis}[title={Original}, width=5cm, height=4.4cm, view={0}{90},
  xlabel={$z$}, ylabel={$x$}, xtick=\empty, ytick=\empty,
  enlargelimits=false, colormap/hot, shader=interp,
  domain=0:1, y domain=0:1]
\addplot3[surf, samples=60, samples y=45]
  {exp(-(y-(0.70-0.35*x))^2/(2*(0.12+0.03*x)^2))};
\end{axis}
\end{tikzpicture}\hfill
\begin{tikzpicture}
\begin{axis}[title={Regularized}, width=5cm, height=4.4cm, view={0}{90},
  xlabel={$z$}, ylabel={$x$}, xtick=\empty, ytick=\empty,
  enlargelimits=false, colormap/hot, shader=interp,
  domain=0:1, y domain=0:1]
\addplot3[surf, samples=60, samples y=45]
  {(x<0.5)*exp(-(y-0.70)^2/(2*0.12^2)) + (x>=0.5)*exp(-(y-0.35)^2/(2*0.15^2))};
\draw[orange, very thick] (axis cs:0.5,0) -- (axis cs:0.5,1);
\node[orange, font=\scriptsize, anchor=north] at (axis cs:0.25,0.02){$A_{1,2}$};
\node[orange, font=\scriptsize, anchor=north] at (axis cs:0.75,0.02){$A_{2,2}$};
\end{axis}
\end{tikzpicture}\hfill
\begin{tikzpicture}
\begin{axis}[width=5cm, height=4.4cm,
  xlabel={$\X$}, ylabel={Likelihood}, xtick=\empty, ytick=\empty,
  domain=0:1, enlargelimits=false]
\addplot[red,  thick, samples=100] {exp(-(x-0.70)^2/(2*0.12^2))};
\addplot[blue, thick, samples=100] {exp(-(x-0.35)^2/(2*0.15^2))};
\node[red]  at (axis cs:0.72,0.75){$P_{1,2}$};
\node[blue] at (axis cs:0.30,0.60){$P_{2,2}$};
\end{axis}
\end{tikzpicture}
\caption{Regularizing the marginal profile of conditional endogenous values: the original joint density of $(Z,X)$ (left), its cell-wise average (center), and the resulting regularized marginals $P_{1,2},P_{2,2}$ (right).}
\label{fig:regularizing-marginal}
\end{figure}
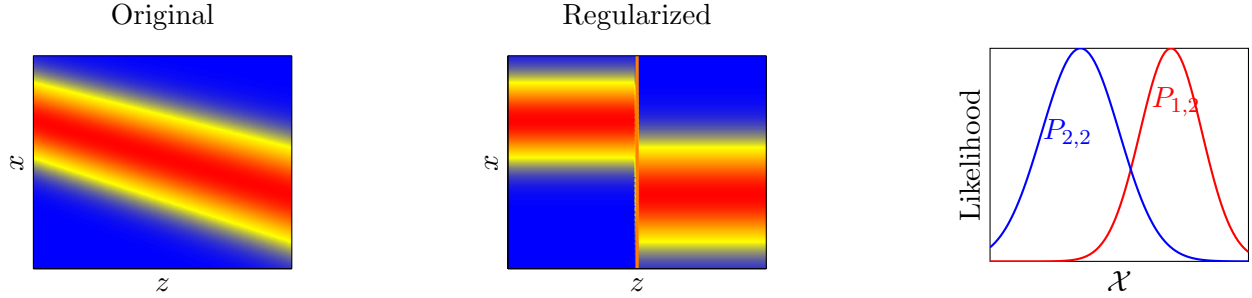

These two effects transfer directly to the partial identification and constrained transport problems of Section~\ref{sec:ot_connection}. Replacing the mechanism $g$ by the discretized mechanism $\pi_n\circ g$, the data-matching restriction becomes a restriction against the discretized joint law. In this case, a latent law is feasible only if the discretized mechanism pushes it forward to $P_n$. In terms of the marginal operator $T^\ast$, this is expressed compactly through its discretized counterpart $T_n^\ast=T^\ast\circ\pi_{n\#}$ (Appendix~\ref{subsec:marginal-operators}), giving the discretized identified set and lower bound on the partially identified effect:
\[\mathcal{Q}_n(\mathbf{P},\lambda)=\big\{Q\in\P(\Omega): T_n^\ast(g_\#Q)=P_n\big\},\qquad v_n(\mathbf{P},\lambda)=\inf_{Q\in\mathcal{Q}_n(\mathbf{P},\lambda)}E_Q[\Lambda].\]
The set $\mathcal{Q}_n(\mathbf{P},\lambda)$ represents the set of latent laws consistent with the discretized mechanism and data, while $v_n(\mathbf{P},\lambda)$ describes the lowest average effect of $\Lambda$ for consistent latent distributions.

The discretization acts analogously on the constrained optimal transport problem. Once the mechanism is discretized, two paths are distinguished only insofar as their projections differ. Following the worst-case logic of Section~\ref{sec:ot_connection}, the cost of a path $e\in\mathbb{H}$ is determined by the latent types whose discretized image coincides with the projected path $\pi_n(e)$:
\[\Psi_n(e) = \begin{cases}
    \inf\{\Lambda(\omega): \pi_n(g(\omega))=\pi_n(e)\}, & \text{if } \pi_n(e)\in \pi_n(\mathbb{K}),\\
    \sup_{\omega\in\Omega}\Lambda(\omega)+1, & \text{otherwise.}
\end{cases}\]
Correspondingly, the feasible path distributions must, once projected, both match the discretized marginal profile and remain supported on the projected achievable set,
\[\mathcal{M}_n(\mathbf{P},\lambda) = \big\{\mu\in\mathcal{P}(\mathbb{H}): T_n^\ast\mu=P_n \text{ and } \mathrm{supp}(\pi_n\#\mu)\subseteq \pi_n(\mathbb{K})\big\}.\]
These give rise to the discretized constrained optimal transport problem
\[u_n(\mathbf{P},\lambda) = \inf_{\mu\in\mathcal{M}_n(\mathbf{P},\lambda)} E_\mu[\Psi_n].\]

So far, nothing shows that these discretized problems bear any relation to the original ones and are informative quantities. Their relevance stems from the fact that, under a mild richness condition on the discretization, their values agree with each other and approximate the original bound. The condition that ensures this is:

\begin{assumption}\label{ass:constraints}
    For every feasible latent distribution $Q\in\mathcal{Q}(\mathbf{P},\lambda)$, there exists a sequence $(Q_n)_{n\in\N}$ of latent distributions such that
    \[Q_n\in\mathcal{Q}_n(\mathbf{P},\lambda)\quad\text{ and }\quad Q_n\Rightarrow Q.\]
\end{assumption}

Assumption~\ref{ass:constraints} requires that every latent law feasible for the original problem be approximable, weakly, by a sequence of latent laws feasible for the discretized problems. This is a mild richness condition and it is exactly what is needed to rule out the discretization from inadvertently excluding relevant latent laws in the limit.

As claimed above, under the above assumption, the discretization of the instrument leads to discretized partial identification and constrained optimal transport problems that retain the same value, and become asymptotically consistent with the limiting problem. These are collected in the following result. 

\begin{proposition}\label{prop:discretized-ot-pi}
Under Assumptions~\ref{ass:PI} and \ref{ass:constraints}, the discretized lower bound and constrained optimal transport problems share a common value,
\[v_n(\mathbf{P},\lambda)=u_n(\mathbf{P},\lambda),\]
and, as $n\uparrow+\infty$, these values converge to the original bounds:
\[v_n(\mathbf{P},\lambda)\rightarrow v(\mathbf{P},\lambda)\quad \text{ and } \quad u_n(\mathbf{P},\lambda)\rightarrow u(\mathbf{P},\lambda).\]
\end{proposition}

\begin{proof}
    See Appendix~\ref{app:proof_discretized_ot_pi}.
\end{proof}


With Proposition~\ref{prop:discretized-ot-pi} in hand, the resolution $n$ becomes a genuine control parameter, without loss of generality, the original infinite-dimensional problem is recovered exactly in the limit, while every finite $n$ delivers an essentially $n$ multi-marginal optimal transportation problem constrained by the mechanism. Working with an instrument taking $n$ values is therefore without loss of generality, and this is the representation the remainder of the paper builds on. The next section departs from this $n$-resolution representation to address the constraint discretization leaves untouched: the requirement that the transport plan be supported on the achievable set.

\subsection{Relaxing the support constraint}\label{sec:moment_penalization}

Throughout this section, we work with the instrument already discretized to at most $n$ values, whether because it is genuinely discrete, or because we have applied the projection $\pi_n$ of Section~\ref{sec:discretization}. Every optimal transport problem encountered so far, discretized or not, carries a support restriction alongside its marginal constraints. This constraint asserts that the transport plan must be supported on the set of achievable paths $\mathbb{K}$, or, in the discretized problem, on its projected image $\pi_n(\mathbb{K})$. This restriction is a genuine shape restriction on potential endogenous paths, not an incidental feature of the formulation, encoding real modeling content. A researcher confident that the endogenous variable responds monotonically to the instrument, for instance, is asserting that every achievable path is monotone in $z$; smoothness or continuity of the response is likewise a restriction on which paths belong to $\mathbb{K}$. Dropping the support constraint would mean discarding exactly this informed content, a feature that cannot simply be ignored. In the remainder of this section, we show how to work around the practical difficulties this restriction poses without giving it up.

Support and shape constraints of this kind are a longstanding source of difficulty in the optimal transport literature, and enforcing them exactly is typically impractical. Even the best-studied instance, the martingale constraint of model-free finance, is known to break the numerical machinery (e.g. linear-programming duality, network-flow algorithms, entropic regularization) built for the unconstrained problem, precisely because that machinery relies on the coupling being free to range over the entire product of marginals \citep{beiglbock2013model}. The standard resolution in that literature is not to enforce the constraint exactly, but to relax it into the objective through a penalization that vanishes as a control parameter shrinks \citep{guo2019computational}. We adopt the same strategy here.

Concretely, for a control parameter $\delta>0$, we replace the cost $\Psi_n$, which assigns an essentially infinite cost to every path outside the achievable set $\pi_n(\mathbb{K})$, with a penalized cost that instead charges a path in proportion to its distance from that set:
\[\Psi_{\delta,n}(e) = \inf_{\omega\in\Omega}\left\{\Lambda(\omega) + \delta^{-1}\big\|\pi_n(g(\omega))-\pi_n(e)\big\|_{\mathbb{H}}\right\}.\]
Rather than the smallest effect among the types exactly consistent with a path, cost function $\Psi_{\delta,n}$ takes the smallest sum of a type's effect and the cost of the shortfall between the path it generates and the path being priced. The parameter $\delta$ controls the trade-off between the two. Indeed, the smaller $\delta$, the more heavily an unachievable path is penalized, so that $\Psi_{\delta,n}$ approaches the original, essentially-infinite cost off $\pi_n(\mathbb{K})$. On achievable paths, $\Psi_{\delta,n}$ instead increases toward the unpenalized cost $\Psi_n$. Section~\ref{subsec:costs-correspondences} formalizes these properties (Lemma~\ref{lemma:regularized-cost}).

With the support constraint absorbed into the cost, the relevant set of path distributions is no longer $\mathcal{M}_n(\mathbf{P},\lambda)$, but the larger set of path distributions satisfying only the marginal constraint,
\[\mathcal{C}_n(\mathbf{P},\lambda) = \big\{\mu\in\mathcal{P}(\mathbb{H}): T_n^\ast\mu=P_n\big\}\]
(see Section~\ref{subsec:marginal-operators} for its properties). Paired with the penalized cost $\Psi_{\delta,n}$, this gives the \emph{unconstrained} optimal transport problem
\[u_{\delta,n}(\mathbf{P},\lambda) = \inf_{\mu\in\mathcal{C}_n(\mathbf{P},\lambda)} E_\mu[\Psi_{\delta,n}].\]
Since $\mathcal{C}_n(\mathbf{P},\lambda)\supseteq\mathcal{M}_n(\mathbf{P},\lambda)$ and $\Psi_{\delta,n}\leq\Psi_n$, this problem's value is always (weakly) below that of the constrained one it relaxes, but it is a substantially simpler object to work with in practice.

Although $u_{\delta,n}(\mathbf{P},\lambda)$ is not itself the value of a partial identification problem -- a path distribution feasible for $\mathcal{C}_n(\mathbf{P},\lambda)$ need not be the pushforward, by the mechanism, of any latent law -- it remains a useful and tractable proxy for $u_n(\mathbf{P},\lambda)$. Its optimal coupling may assign positive probability to unachievable paths, but the total such mass is controlled by $\delta$ and vanishes as $\delta\downarrow0$, at which point the optimal solutions of the relaxed problem converge to an optimal solution of the constrained one. Just as important for computation, $u_{\delta,n}(\mathbf{P},\lambda)$ is, once paths are identified with their $n$ cell values, nothing more than a classical multi-marginal optimal transport problem on $\X^n$.

To see this, note that $\Psi_{\delta,n}$ depends on a path only through its projection $\pi_n(e)$, and every projected path is uniquely identified, through the isometry $\tilde{\pi}_n$ (Proposition~\ref{prop:properties-projections} in Appendix~\ref{subsec:projections}), with a vector $(x_1,\ldots,x_n)\in\X^n$ of its $n$ cell values. This lets us descend $\Psi_{\delta,n}$ to a cost function directly on $\X^n$,
\[c_{\delta,n}(x_1,\ldots,x_n) = \Psi_{\delta,n}\big(\tilde{\pi}_n^{-1}(x_1,\ldots,x_n)\big),\]
which inherits boundedness and Lipschitz continuity from $\Psi_{\delta,n}$ (Lemma~\ref{lemma:c-is-cont}), exactly the regularity a multi-marginal Monge--Kantorovich problem needs. The same identification carries over to couplings. Because $\Psi_{\delta,n}$ and the constraint defining $\mathcal{C}_n(\mathbf{P},\lambda)$ both depend on a path only through $\pi_n$, the pushforwards $\tilde{\pi}_{n\#}\mathcal{C}_n(\mathbf{P},\lambda)$ are exactly the couplings on $\X^n$ with $k$-th marginal $P_{k,n}$,
\[\Pi_n(\mathbf{P},\lambda) = \big\{\nu\in\P(\X^n): \nu_k=P_{k,n},\ k\in[n]\big\}.\]
Together, $c_{\delta,n}$ and $\Pi_n(\mathbf{P},\lambda)$ define a classical multi-marginal Monge--Kantorovich problem on $\X^n$, whose value coincides with $u_{\delta,n}(\mathbf{P},\lambda)$, as the following proposition shows. Taking $\delta\downarrow0$ then recovers the discretized bound $u_n(\mathbf{P},\lambda)$, and taking both $\delta\downarrow0$ and $n\uparrow\infty$ jointly recovers the original bound $u(\mathbf{P},\lambda)$. Thus, through Theorem~\ref{thm:strong-duality} and Proposition~\ref{prop:discretized-ot-pi}, the joint limit also recovers the sharp lower bound $v(\mathbf{P},\lambda)$ itself.

\begin{proposition}\label{prop:support-penalization}
Under Assumptions~\ref{ass:PI} and \ref{ass:constraints}, the value of the penalized optimal transport problem $u_{\delta,n}(\mathbf{P},\lambda)$ coincides with the value of the multi-marginal Monge--Kantorovich problem with cost $c_{\delta,n}$ and marginals $(P_{k,n})_{k\in[n]}$,
\[u_{\delta,n}(\mathbf{P},\lambda) = \inf_{\nu\in\Pi_n(\mathbf{P},\lambda)} E_\nu[c_{\delta,n}].\]
Moreover, as $\delta\downarrow0$ with $n$ fixed, $u_{\delta,n}(\mathbf{P},\lambda)$ converges monotonically to $u_n(\mathbf{P},\lambda)$. Taking the joint limit $\delta\downarrow0$, $n\uparrow+\infty$, it converges to the original value $u(\mathbf{P},\lambda)$:
\[u_{\delta,n}(\mathbf{P},\lambda)\uparrow u_n(\mathbf{P},\lambda)\quad\text{ and }\quad u_{\delta,n}(\mathbf{P},\lambda)\rightarrow u(\mathbf{P},\lambda).\]
\end{proposition}

\begin{proof}
    See Appendix~\ref{app:proof_support_penalization}.
\end{proof}

Proposition~\ref{prop:support-penalization} is important for at least three reasons. First, it shows that the discretized partial identification problem, once its support constraint is relaxed, can be recast exactly as a classical multi-marginal Monge--Kantorovich problem on the finite-dimensional space $\X^n$, the simplification the next section builds its estimator around. Second, the characterization is consistent with every problem it approximates, as when $\delta\downarrow0$ it recovers the discretized bound $u_n(\mathbf{P},\lambda)$, and jointly with $n\uparrow\infty$ it recovers the original bound $u(\mathbf{P},\lambda)$. Third, and most consequential for what follows, casting the problem on $\X^n$ identifies a natural reference coupling on that space, the independent coupling of the marginals $(P_{k,n})_{k\in[n]}$. It is against this reference that the next section measures a relative-entropy penalty, turning the Monge--Kantorovich problem of Proposition~\ref{prop:support-penalization}, and not the original, support-constrained one, into the entropic optimal transport problem of Section~\ref{sec:entropic_perturbation}.

\begin{figure}[H]
\centering
\adjustbox{max width=0.8\textwidth}{
\begin{minipage}{\textwidth}
\centering
\begin{subfigure}[b]{0.48\textwidth}
\centering
\resizebox{\linewidth}{!}{%
\begin{tikzpicture}[>={Stealth[length=2mm]}, line cap=round]

\begin{scope}[shift={(0,0)}]
  \draw[thick] plot[smooth cycle, tension=0.9]
    coordinates {(-1.3,0.3)(-0.7,1.0)(0.4,1.1)(1.2,0.4)(0.9,-0.6)(-0.1,-1.0)(-1.0,-0.5)};
  \fill[red]  (-0.35,0.35) circle (1.6pt) node[above=1pt, black]{$\omega$};
  \fill[blue] ( 0.15,-0.30) circle (1.6pt) node[right=1pt, black]{$\omega'$};
  \node at (1.4,1.0) {$\Omega$};
\end{scope}

\begin{scope}[shift={(7,0)}]
  \draw[->] (0,-1.4) -- (0,1.5) node[left]{$x$};
  \draw[->] (0,-1.4) -- (3.4,-1.4) node[below]{$z$};
  \draw[blue, thick] (0.3,1.2) .. controls (1.5,0.55) and (2.3,0.2) .. (3.1,0.05);
  \draw[red,  thick] (0.3,0.6) .. controls (1.5,0.05) and (2.3,-0.25) .. (3.1,-0.4);
  \draw[orange, very thick] (0.4,-1.4) -- (2.8,-1.4);
  \foreach \x in {0.4,1.2,2.0,2.8} \draw[orange, thick] (\x,-1.5) -- (\x,-1.3);
  \node[below] at (0.8,-1.55){$A_{1,3}$};
  \node[below] at (1.6,-1.55){$A_{2,3}$};
  \node[below] at (2.4,-1.55){$A_{3,3}$};
  \node at (3.1,1.5){$\mathbb{H}$};
\end{scope}

\begin{scope}[shift={(7,-6)}]
  \draw[->] (0,-1.4) -- (0,1.5) node[left]{$x$};
  \draw[->] (0,-1.4) -- (3.4,-1.4) node[below]{$z$};
  \draw[blue, thick] (0.4,1.05)--(1.2,1.05) (1.2,0.65)--(2.0,0.65) (2.0,0.30)--(2.8,0.30);
  \draw[red,  thick] (0.4,0.15)--(1.2,0.15) (1.2,-0.15)--(2.0,-0.15) (2.0,-0.40)--(2.8,-0.40);
  \draw[orange, very thick] (0.4,-1.4) -- (2.8,-1.4);
  \foreach \x in {0.4,1.2,2.0,2.8} \draw[orange, thick] (\x,-1.5) -- (\x,-1.3);
  \node[below] at (0.8,-1.55){$A_{1,3}$};
  \node[below] at (1.6,-1.55){$A_{2,3}$};
  \node[below] at (2.4,-1.55){$A_{3,3}$};
  \node at (3.0,1.5){$\pi_n(\mathbb{H})$};
\end{scope}

\begin{scope}[shift={(0.3,-6)},
              x={(-0.42cm,-0.35cm)}, y={(1cm,0cm)}, z={(0cm,1cm)}]
  \draw[->] (0,0,0) -- (3,0,0) node[below left]{$x$};
  \draw[->] (0,0,0) -- (0,3,0) node[right]{$x$};
  \draw[->] (0,0,0) -- (0,0,3) node[left]{$x$};
  \fill[red] (1,1,1) circle (1.6pt);
  \draw[red, dashed, thin] (1,1,1) -- (1,1,0) -- (1,0,0);
  \draw[red, dashed, thin] (1,1,0) -- (0,1,0);
  \fill[blue] (2,2,2) circle (1.6pt);
  \draw[blue, dashed, thin] (2,2,2) -- (2,2,0) -- (2,0,0);
  \draw[blue, dashed, thin] (2,2,0) -- (0,2,0);
\end{scope}

\draw[->] (2.1,0.2) -- (6.6,0.2) node[midway, above]{$g$};
\draw[->] (8.5,-1.6) -- (8.5,-4.4) node[midway, right]{$\pi_n$};
\draw[->] (6.6,-4.7) -- (2.6,-4.7) node[midway, above]{$\tilde{\pi}_n$};
\draw[->, dashed] (-0.6,-0.9) -- (-0.6,-4.1) node[midway, left]{$\tilde{\pi}_n \circ g$};


\end{tikzpicture}%
}
\caption{From latent types to $\X^n$.}
\end{subfigure}
\hfill
\begin{subfigure}[b]{0.48\textwidth}
\centering
\resizebox{\linewidth}{!}{%
\begin{tikzpicture}[>={Stealth[length=2mm]}, line cap=round,
   nd/.style={draw, rounded corners, inner sep=4pt, font=\small}]
  \def\R{3.4}
  \node[nd] (n1) at (90:\R)   {$(\Lambda,\Omega)$};
  \node[nd] (n2) at (18:\R)   {$(\Psi,\mathbb{H})$};
  \node[nd] (n3) at (-54:\R)  {$(\Psi_n,\pi_n(\mathbb{H}))$};
  \node[nd] (n4) at (-126:\R) {$(\Psi_{\delta,n},\pi_n(\mathbb{H}))$};
  \node[nd] (n5) at (162:\R)  {$(c_{\delta,n},\X^n)$};

  \draw[->] (n1) -- (n2) node[midway, above right, align=center, font=\scriptsize]{$g$\\inf};
  \draw[->] (n2) -- (n3) node[midway, right, align=center, font=\scriptsize]{$\pi_n$\\inf};
  \draw[->] (n3) -- (n4) node[midway, below, align=center, font=\scriptsize]
       {Penalization\\$+\,\delta^{-1}\lVert\pi_n(g(\cdot))-\pi_n(e)\rVert_{\mathbb{H}}$};
  \draw[->] (n4) -- (n5) node[midway, left, font=\scriptsize]{$\tilde{\pi}_n$};
  \draw[->, dashed] (n1) -- (n5);
\end{tikzpicture}%
}
\caption{Redefining the cost.}
\end{subfigure}
\end{minipage}%
}
\caption{Two views of the same construction underlying Proposition~\ref{prop:support-penalization}. On the left, the chain of identifications $\Omega\to\mathbb{H}\to\pi_n(\mathbb{H})\to\X^n$ that recasts latent types as points of $\X^n$. On the right, the corresponding chain of costs each space carries -- $\Lambda$, $\Psi$, $\Psi_n$, $\Psi_{\delta,n}$, and finally $c_{\delta,n}$.}
\label{fig:morphisms-costs}
\end{figure}
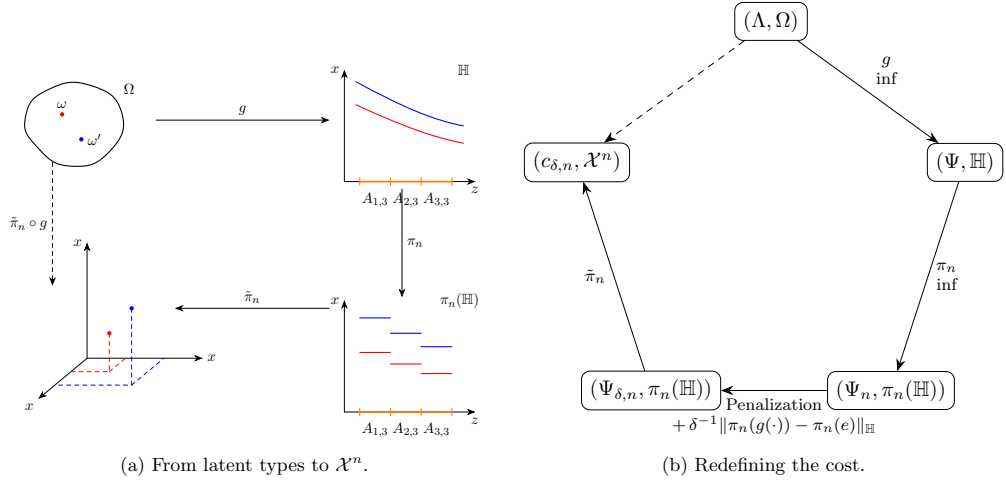

\subsection{Entropic regularization}\label{sec:entropic_perturbation}

Our investigation so far has reduced the sharp lower bound on the partially identified average of $\Lambda$ to a sequence of controlled approximations, each trading a manageable error for tractability. Theorem~\ref{thm:strong-duality} recast it as a constrained optimal transport problem on the path space $\mathbb{H}$; Proposition~\ref{prop:discretized-ot-pi} discretized this problem to a finite resolution $n$; and Proposition~\ref{prop:support-penalization} relaxed its support constraint into a penalized cost $c_{\delta,n}$, recasting it as the multi-marginal Monge--Kantorovich problem $u_{\delta,n}(\mathbf{P},\lambda)$ on $\X^n$. This last identification is already a considerable simplification, since rather than a distribution over an infinite-dimensional, mechanism-constrained path space, the object of interest is now a coupling of $n$ finite-dimensional marginals under a Lipschitz cost. Yet, as a linear program, $u_{\delta,n}(\mathbf{P},\lambda)$ still suffers from a familiar drawback. Its set of optimal couplings need not be a singleton, and is generically large -- a face, possibly of high dimension, of the polytope of couplings with the prescribed marginals. This multiplicity is more than a theoretical fact. In fact, without a way to select among optimal couplings, the numerical determination of the value $u_{\delta,n}(\mathbf{P},\lambda)$ becomes considerably more delicate in practice, particularly as $n$ grows.

To overcome the difficulty arising from this multiplicity of solutions, we perturb the objective of $u_{\delta,n}(\mathbf{P},\lambda)$ with a relative entropy term, controlled by a regularization parameter $\epsilon>0$. This single modification delivers a cluster of convenient properties at once. The resulting problem is a strictly convex minimization, and therefore admits a unique solution -- resolving the multiplicity of Proposition~\ref{prop:support-penalization} at the source. Strong duality holds, so the value of the entropic problem can equivalently be recovered from its dual, an unconstrained maximization over a finite vector of functions of the data. And because this dual problem is smooth, its first-order optimality condition is a fixed-point equation in the vector of maximizers, one that can be solved numerically, and efficiently, by the Sinkhorn algorithm \citep{cuturi2013sinkhorn, nutz2021entropic, peyre2019computational}. We develop each of these properties below.

Both the entropy penalty and the Sinkhorn algorithm require a reference measure against which a coupling's density is defined and its distance to independence is measured. Fortunately, the transport formulation of the previous two sections already singles out a natural candidate. Return to the endogenous-path description of Section~\ref{sec:discretization}, and consider the measure $R_n\in\P(\pi_n(\mathbb{H}))$ that assigns to every cylinder set of piecewise-constant paths the product probability of its $n$ cell values under the regularized marginals $(P_{k,n})_{k\in[n]}$,
\[R_n\big(\{e\in\pi_n(\mathbb{H}): \tilde{\pi}_{k,n}(e)\in B_k,\ k\in[n]\}\big)=\prod_{k\in[n]}P_{k,n}(B_k),\qquad B_1,\ldots,B_n\in\mathcal{B}(\X);\]
that is, under $R_n$, a piecewise-constant path has its value on each cell of the partition sampled independently from the corresponding regularized marginal. Every feasible coupling $\mu\in\mathcal{C}_n(\mathbf{P},\lambda)$ admits this measure as a natural dominating measure for its projection. Since $\pi_n\#\mu$ has $k$-th cell marginal $P_{k,n}$ by definition of $\mathcal{C}_n(\mathbf{P},\lambda)$, any cylinder set assigned zero probability by $R_n$ -- that is, with $P_{k,n}(B_k)=0$ for some $k$ -- is also assigned zero probability by $\pi_n\#\mu$, so $\pi_n\#\mu\ll R_n$. Thus $R_n$ offers a natural prior distribution against which every feasible coupling admits a density, and the task of finding an optimal coupling becomes that of determining the feasible density that minimizes the expected cost. To discourage densities that concentrate too heavily on regions $R_n$ regards as unlikely, we measure the deviation from $R_n$ by relative entropy. Then, a higher value signals a density concentrating on sets of small $R_n$-probability. This gives the entropic optimal transport problem
\[u_{\epsilon,\delta,n}(\mathbf{P},\lambda)=\inf_{\mu\in\mathcal{C}_n(\mathbf{P},\lambda)}E_\mu[\Psi_{\delta,n}]+\epsilon\,H(\pi_n\#\mu\,|\,R_n).\]
By the same argument as in Proposition~\ref{prop:support-penalization}, this path-space problem can be recast as a classical optimal transport problem on $\X^n$, and the same identification carries over directly to the reference measure. Indeed, since $\tilde{\pi}_n$ restricts to an isometric bijection between $\pi_n(\mathbb{H})$ and $\X^n$ (Proposition~\ref{prop:properties-projections}), it identifies $R_n$ with the reference measure
\[r_n:=\tilde{\pi}_{n\#}R_n=\bigotimes_{k\in[n]}P_{k,n}\in\P(\X^n),\]
obtained by sampling each of the $n$ coordinates independently from the corresponding regularized marginal profile. Because relative entropy is invariant under measurable bijections, $H(\pi_{n\#}\mu\,|\,R_n)=H(\tilde{\pi}_{n\#}\mu\,|\,r_n)$ for every $\mu\in\mathcal{C}_n(\mathbf{P},\lambda)$. By the same argument as above, every coupling $\nu\in\Pi_n(\mathbf{P},\lambda)$ is likewise absolutely continuous with respect to $r_n$. Using $r_n$ as the reference measure delivers the following equivalent, classical entropic optimal transport characterization:
\begin{equation}\label{eq:eot_value}
u_{\epsilon,\delta,n}(\mathbf{P},\lambda)=\inf_{\nu\in\Pi_n(\mathbf{P},\lambda)}E_\nu[c_{\delta,n}]+\epsilon\,H(\nu\,|\,r_n).
\end{equation}

Departing from \eqref{eq:eot_value}, we can obtain a dual characterization of its value. For this, we introduce the following auxiliary notation and objects. Similar to the path distributions of the previous sections, the set of couplings $\Pi_n(\mathbf{P},\lambda)$ can be expressed in terms of the adjoint of operator $\mathrm{T}_n:C_b(\X)^n\to C_b(\X^n)$ (see Appendix~\ref{subsec:marginal-operators}), mapping a vector of functions $\phi=(\phi_1,\ldots,\phi_n)$ on $\X$ to:
\[\mathrm{T}_n\phi(x)=\sum_{k\in[n]}\phi_k(x_k)\lambda_{k,n}.\]
Beyond describing the feasible set of couplings, this operator also appears in the dual of the entropic optimal transport above. Indeed, we refer to each vector $\phi\in C_b(\X)^n$ as a (Kantorovich) potential function, whose interpretation is:  for each component of the partition, the potential $\phi_k(x)$ is to be read as the conditional expected effect of $\Lambda$ for paths whose projection lands, in cell $k$, at the value $x\in\X$. Its objective is to choose the potentials that maximize the expected weighted value $\mathrm{T}_n\phi$ under the reference measure $r_n$, net of a soft (exponential) approximation to the indicator that this weighted value is dominated by the cost $c_{\delta,n}$ -- the same soft-max penalty that appears in the classical entropic optimal transport dual. 

Just as the support-relaxation parameter $\delta$, the entropic penalty vanishes as $\epsilon\downarrow0$, recovering the unregularized value. For instance, \citet{carlier2017convergence} shows that, under our assumptions, the entropic value is encapsulated within an explicit and vanishing distance of $u_{\delta,n}(\mathbf{P},\lambda)$,
\[u_{\delta,n}(\mathbf{P},\lambda)\ \leq\ u_{\epsilon,\delta,n}(\mathbf{P},\lambda)\ \leq\ u_{\delta,n}(\mathbf{P},\lambda)+d\,\epsilon\log\epsilon+\frac{\epsilon}{\delta},\]
where $d$ is bounded by the dimension of $\X$. This finer control is what the joint-limit statements below package into simple convergence. Together, these facts are collected in the following proposition.

\begin{proposition}\label{prop:duality}
Under Assumption~\ref{ass:PI} and \ref{ass:constraints}, the value of the entropic optimal transport problem \eqref{eq:eot_value} equals the value of its dual,
\[u_{\epsilon,\delta,n}(\mathbf{P},\lambda)=\sup_{\phi\in C_b(\X)^n}E_{r_n}\!\left[\mathrm{T}_n\phi-\epsilon\left(e^{\frac{\mathrm{T}_n\phi-c_{\delta,n}}{\epsilon}}-1\right)\right].\]
Moreover, as $\epsilon\downarrow0$ with $\delta,n$ fixed, $u_{\epsilon,\delta,n}(\mathbf{P},\lambda)\downarrow u_{\delta,n}(\mathbf{P},\lambda)$, jointly sending $\delta\downarrow0$ and $\epsilon/\delta\rightarrow0$ with $n$ fixed, $u_{\epsilon,\delta,n}(\mathbf{P},\lambda)\rightarrow u_n(\mathbf{P},\lambda)$, and, further letting $n\uparrow+\infty$,
\[u_{\epsilon,\delta,n}(\mathbf{P},\lambda)\rightarrow u(\mathbf{P},\lambda).\]
\end{proposition}

\begin{proof}
    See Appendix~\ref{app:proof_duality}.
\end{proof}

The relevance of this dual characterization goes beyond its value. The entropic problem \eqref{eq:eot_value} is attainable, in the sense that its dual admits a maximizer, and this maximizer can be recovered from a first-order optimality condition. From a simple application of variational analysis and differentiating the dual objective of Proposition~\ref{prop:duality} with respect to each potential $\phi_k$ shows that, at an optimum, the system of $n$ equations
\begin{equation}\label{eq:foc}
\phi_k(x_k) = -\frac{\epsilon}{\lambda_{k,n}}\,\log\!\left(E_{r_{-k,n}}\!\left[\exp\!\left(\frac{\mathrm{T}_{k,n}\phi_{-k}-c_{\delta,n}(x_k,\cdot)}{\epsilon}\right)\right]\right),\qquad P_{k,n}\text{-a.s.},
\end{equation}
must hold for every $k\in[n]$, where $\mathrm{T}_{k,n}:C_b(\X)^{n-1}\to C_b(\X^{n-1})$ is the operator of Appendix~\ref{subsec:marginal-operators} obtained from $\mathrm{T}_n$ by dropping the $k$-th coordinate, and $r_{-k,n}=\bigotimes_{j\neq k}P_{j,n}$ is the corresponding marginal of the reference measure $r_n$. Equation~\eqref{eq:foc} expresses each potential as a fixed point of the remaining $n-1$. Such system can be solved iteratively, updating one potential at a time while holding the others fixed -- an algorithm known as Sinkhorn's, which converges to the optimal vector of potentials under mild conditions \citep{cuturi2013sinkhorn, peyre2019computational}. Attainment, the fixed-point characterization, and the resulting expression for the optimal coupling are collected in the following proposition.

\begin{proposition}\label{prop:entropic-ot-dual}
Under Assumption~\ref{ass:PI}, the dual problem of Proposition~\ref{prop:duality} admits a maximizer $\phi=(\phi_1,\ldots,\phi_n)\in C_b(\X)^n$, each component $\delta^{-1}$-Lipschitz continuous and bounded by $\frac{1+2\delta^{-1}\mathrm{diam}(\X)}{\lambda_{k,n}}$, satisfying the first-order condition \eqref{eq:foc}. The value of the entropic problem is recovered from these potentials by
\[u_{\epsilon,\delta,n}(\mathbf{P},\lambda)=E_{r_n}[\mathrm{T}_n\phi]=\sum_{k\in[n]}\lambda_{k,n}\,E_{P_{k,n}}[\phi_k]=\sum_{k\in[n]}\int_{A_{k,n}\times\X}\phi_{k}(x)dP(z,x).\]
Moreover, the optimal coupling $\nu^\ast_{\epsilon,\delta,n}\in\Pi_n(\mathbf{P},\lambda)$ solving \eqref{eq:eot_value} exists, is unique, and its log-density relative to $r_n$ is given by
\[\log\!\left(\frac{d\nu^\ast_{\epsilon,\delta,n}}{dr_n}\right)(x_1,\ldots,x_n)=\frac{\mathrm{T}_n\phi(x_1,\ldots,x_n)-c_{\delta,n}(x_1,\ldots,x_n)}{\epsilon}.\]
\end{proposition}

\begin{proof}
    See Appendix~\ref{app:proof_entropic_ot_dual}.
\end{proof}

The log-density formula of Proposition~\ref{prop:entropic-ot-dual} admits a transparent reading. Each potential $\phi_k(x_k)$ can be interpreted as the average effect of $\Lambda$ under the optimal latent distribution, conditional on the exogenous variable falling in cell $A_{k,n}$ and the endogenous realization equal to $x_k$. Averaged against the weights $(\lambda_{k,n})_{k\in[n]}$, the quantity $\mathrm{T}_n\phi(x)$ records the average effect of $\Lambda$ under the optimal latent distribution, conditional on the discretized path landing at the point $x\in\X^n$. The cost $c_{\delta,n}(x)$, on the other hand, records the minimum value of $\Lambda$ attainable by any latent type whose induced path lies close to $x$. 

The optimal coupling $\nu^\ast_{\epsilon,\delta,n}$ balances these two objects. Its likelihood at $x$ is, by Proposition~\ref{prop:entropic-ot-dual}, proportional to the exponentiated difference between the average effect $\mathrm{T}_n\phi(x)$ and the minimal cost $c_{\delta,n}(x)$, so that the optimal coupling concentrates on points where the average effect is close to the smallest value the cost can attain. Since the entropic problem \eqref{eq:eot_value} minimizes expected cost subject to the marginal constraint, this is exactly the behavior one would expect of an optimal coupling. It internalizes the trade-off between conforming to the observed marginals and favoring points where the (unobserved) worst-case effect is small. Undoing the identification of $\X^n$ with $\pi_n(\mathbb{H})$, the same Gibbs structure describes the optimal path distribution $\mu^\ast_{\epsilon,\delta,n}=\tilde{\pi}_n^{-1}\#\nu^\ast_{\epsilon,\delta,n}$ relative to the reference measure $R_n$, closing the loop back to the endogenous-path formulation of Section~\ref{sec:discretization}.

Together with Propositions~\ref{prop:discretized-ot-pi} and~\ref{prop:support-penalization}, Proposition~\ref{prop:duality} shows that, for any desired accuracy, one can choose $n$ large and $\delta,\epsilon$ small enough that $u_{\epsilon,\delta,n}(\mathbf{P},\lambda)$ lies within any prescribed tolerance of the sharp bound $v(\mathbf{P},\lambda)$, while remaining, for every fixed $(\epsilon,\delta,n)$, an unconstrained, strictly convex problem solvable by the Sinkhorn iteration of Proposition~\ref{prop:entropic-ot-dual}. Section~\ref{sec:x_discretization} closes the remaining gap before estimation, discretizing the endogenous space itself so that the problem becomes a genuinely finite, computable object when solved with data.

\subsection{Discretizing the endogenous space}\label{sec:x_discretization}

The entropic problem of Section~\ref{sec:entropic_perturbation} resolves the multiplicity of optimal couplings, but a difficulty remains. Fixing $(\epsilon,\delta,n)$ does not, by itself, turn \eqref{eq:eot_value} into a finite-dimensional problem. The marginals $(P_{k,n})_{k\in[n]}$ entering its dual are, in general, continuous distributions with possibly uncountable support on $\X$, and so, correspondingly, are the dual potentials $\phi=(\phi_1,\ldots,\phi_n)$ of Proposition~\ref{prop:entropic-ot-dual}. This happens as the first-order condition \eqref{eq:foc} pins down the value of each $\phi_k$ at every point of the support of $P_{k,n}$, so that solving for $\phi$ means solving for $n$ functions on the support of the respective marginal. Even after four steps of regularization, the entropic problem therefore remains, formally, an infinite-dimensional root-finding problem.

A second, related difficulty appears once the marginals $(P_{k,n})_{k\in[n]}$ are themselves unknown and must be estimated. The plug-in estimator of the value in Section~\ref{sec:entropic_perturbation} replaces each $P_{k,n}$ by the empirical distribution of the endogenous observations falling in cell $A_{k,n}$. Since no information is aggregated along the endogenous dimension, the resulting first-order condition can carry as many unknowns as the sample has observations. Fitting $N$ parameters from a sample of size $N$ is a textbook recipe for overfitting, the estimated potentials would interpolate the sample exactly, at the cost of a noisy, high-variance approximation to the population potentials they target. Discretizing the endogenous space, and grouping both probability mass and observations into the same, finitely many, cells, addresses both difficulties at once. At one hand, it turns \eqref{eq:eot_value} into a genuinely finite-dimensional problem. At another, it aggregates observations before they ever enter the plug-in first-order condition, yielding a more efficient, tractable, and less noisy estimator. This is the purpose of the present section.

To accomplish this, we proceed as Section~\ref{sec:discretization} did for the exogenous variable. For each $m\in\N$, partition $\X$ into finitely many cells, $\mathcal{B}_m=(B_{j,m})_{j\in[m]}$, and select a representative point $x_{j,m}\in B_{j,m}$ within each. We take the partitions $(\mathcal{B}_m)_{m\in\N}$ to be increasing -- so that every cell of $\mathcal{B}_{m+1}$ is a union of cells of $\mathcal{B}_m$ -- and require their mesh size, $\rho_m:=\max_{j\in[m]}\mathrm{diam}(B_{j,m})$, to shrink as the number of cells grows, vanishing in the limit. Write $\X_m:=\{x_{1,m},\ldots,x_{m,m}\}$ for the resulting grid of representative points.

For each $k\in[n]$, define the discretized marginal $P_{k,n,m}\in\P(\X)$ by collapsing the mass $P_{k,n}$ assigns to each cell onto its representative point,
\[P_{k,n,m} := \sum_{j\in[m]} P_{k,n}(B_{j,m})\,\delta_{x_{j,m}}.\]
Each $P_{k,n,m}$ is finitely supported on $\X_m$, on at most $m$ points, and approximates $P_{k,n}$ more closely as the number of cells grows, converging weakly to it once the mesh $\rho_m$ vanishes. 

We collect these discretized marginals into the finitely supported reference measure $r_{n,m}:=\bigotimes_{k\in[n]}P_{k,n,m}\in\P(\X^n)$. The feasible set of couplings with marginals $(P_{k,n,m})_{k\in[n]}$ is
\[\Pi_{n,m}(\mathbf{P},\lambda) := \{\nu\in\P(\X^n): \nu_k=P_{k,n,m},\ k\in[n]\},\]
and it can be seen as the discretized counterpart of $\Pi_n(\mathbf{P},\lambda)$. Every coupling $\nu\in\Pi_{n,m}(\mathbf{P},\lambda)$ is supported on the finite grid $\X_m^n\subseteq\X^n$, and can therefore be represented exactly as an $n$-way tensor of size $m\times\cdots\times m$. Figure~\ref{fig:x-discretization} illustrates the resulting construction. 
\begin{figure}[H]
\centering
\begin{subfigure}[b]{0.315\textwidth}
\centering
\begin{tikzpicture}[x=2.5cm,y=2.5cm]
  \draw[thick] (0,0) rectangle (1,1);
  \draw[thick] (0.22,0) -- (0.22,1);
  \draw[thick] (0.48,0) -- (0.48,1);
  \draw[thick] (0.75,0) -- (0.75,1);
  \draw[thick] (0,0.32) -- (1,0.32);
  \draw[thick] (0,0.68) -- (1,0.68);
  \foreach \x in {0.10,0.36,0.60,0.90}{
    \draw[dotted] (\x,0) -- (\x,1);
    \fill[black] (\x,0) circle (1.1pt);
  }
  \foreach \y in {0.14,0.50,0.84}{
    \draw[dotted] (0,\y) -- (1,\y);
    \fill[black] (0,\y) circle (1.1pt);
  }
  \foreach \x in {0.10,0.36,0.60,0.90}{
    \foreach \y in {0.14,0.50,0.84}{
      \fill[red] (\x,\y) circle (1.1pt);
    }
  }
  \node[below] at (0.5,-0.09) {\scriptsize $\Z$};
  \node[left] at (-0.10,0.5) {\scriptsize $\X$};
  \node[below] at (0.90,-0.03) {\scriptsize $z_{4,n}$};
  \node[left] at (-0.02,0.84) {\scriptsize $x_{3,m}$};
\end{tikzpicture}
\caption{The grid $\X_m^n$}
\end{subfigure}
\hfill
\begin{subfigure}[b]{0.315\textwidth}
\centering
\begin{tikzpicture}[x=2.5cm,y=2.5cm]
  \draw[thick] (0,0) rectangle (1,1);
  \draw[blue,thick] (0,0.20) -- (0.05,0.10) -- (0.10,0.25) -- (0.15,0.18) -- (0.20,0.30)
                     -- (0.24,0.55) -- (0.29,0.75) -- (0.34,0.90) -- (0.39,0.80) -- (0.44,0.85)
                     -- (0.50,0.60) -- (0.57,0.40) -- (0.63,0.55) -- (0.69,0.45) -- (0.74,0.50)
                     -- (0.78,0.65) -- (0.84,0.80) -- (0.90,0.90) -- (0.95,0.78) -- (1.00,0.85);
  \node[below] at (0.5,-0.09) {\scriptsize $\Z$};
  \node[left] at (-0.10,0.5) {\scriptsize $\X$};
\end{tikzpicture}
\caption{An endogenous path $e$}
\end{subfigure}
\hfill
\begin{subfigure}[b]{0.315\textwidth}
\centering
\begin{tikzpicture}[x=2.5cm,y=2.5cm]
  \draw[thick] (0,0) rectangle (1,1);
  \draw[thick] (0.22,0) -- (0.22,1);
  \draw[thick] (0.48,0) -- (0.48,1);
  \draw[thick] (0.75,0) -- (0.75,1);
  \draw[thick] (0,0.32) -- (1,0.32);
  \draw[thick] (0,0.68) -- (1,0.68);
  \draw[blue!35,thick] (0,0.20) -- (0.05,0.10) -- (0.10,0.25) -- (0.15,0.18) -- (0.20,0.30)
                     -- (0.24,0.55) -- (0.29,0.75) -- (0.34,0.90) -- (0.39,0.80) -- (0.44,0.85)
                     -- (0.50,0.60) -- (0.57,0.40) -- (0.63,0.55) -- (0.69,0.45) -- (0.74,0.50)
                     -- (0.78,0.65) -- (0.84,0.80) -- (0.90,0.90) -- (0.95,0.78) -- (1.00,0.85);
  \foreach \x/\y in {0.10/0.14,0.36/0.84,0.60/0.50,0.90/0.84}{
    \fill[blue] (\x,\y) circle (1.4pt);
  }
  \node[below] at (0.5,-0.09) {\scriptsize $\Z$};
  \node[left] at (-0.10,0.5) {\scriptsize $\X$};
\end{tikzpicture}
\caption{$e$ mapped onto $\X_m^n$}
\end{subfigure}
\caption{Discretizing the endogenous space, jointly with the exogenous partition of Section~\ref{sec:discretization} (both shown here as scalar, for clarity). Panel (a): the square $\Z\times\X$ is partitioned into $n=4$ exogenous and $m=3$ endogenous cells, with representative points $z_{k,n}$ and $x_{j,m}$ (black) on each axis and the induced grid of $n\cdot m$ representative pairs (red). Panel (b): a generic endogenous path $e$, unconstrained by the grid. Panel (c): the same path, once information is represented by the single grid points of $\X_m^n$ closest to its behavior there (blue), turning $e$ into one of the $m^n$ elements of $\X_m^n$.}
\label{fig:x-discretization}
\end{figure}
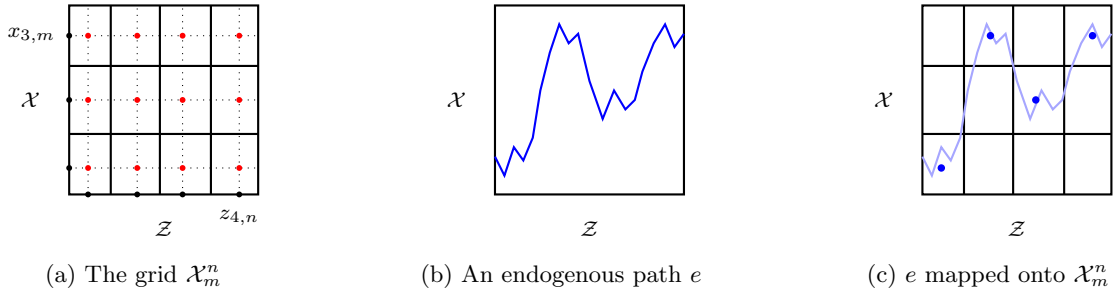

Pairing this feasible set with the cost $c_{\delta,n}$ of Section~\ref{sec:moment_penalization}, the resulting entropic optimal transport problem is
\begin{equation}\label{eq:eot_value_discretized}
u_{\epsilon,\delta,n,m}(\mathbf{P},\lambda) := \inf_{\nu\in\Pi_{n,m}(\mathbf{P},\lambda)} E_\nu[c_{\delta,n}] + \epsilon\,H(\nu\,|\,r_{n,m}).
\end{equation}

As the number of endogenous cells grows, we expect the value of \eqref{eq:eot_value_discretized} to converge to the full-support value $u_{\epsilon,\delta,n}(\mathbf{P},\lambda)$. Because $c_{\delta,n}$ is Lipschitz, this convergence can in fact be quantified, at a rate controlled by the mesh $\rho_m$ and the Lipschitz constant $\delta^{-1}$ of the cost function -- a result that follows from the same lines of \citet{eckstein2023quantitative}. Moreover, since $r_{n,m}$ is supported on the finite grid $\X_m^n$, the first-order condition of the dual of \eqref{eq:eot_value_discretized} reduces to a system of at most $n\cdot m$ equations in as many unknowns, one for the value of each of the $n$ dual potentials at each of the $m$ points of $\X_m$. In this case, the vector of dual potentials solves the (non-linear) system
\begin{equation}\label{eq:eot_fullydiscretized_foc}
\phi_k(x_{\ell,m})=-\frac{\epsilon}{\lambda_{k,n}}\log\!\left(E_{r_{-k,n,m}}\!\left[\exp\!\left(\frac{\mathrm{T}_{k,n}\phi_{-k}-c_{\delta,n}(x_{\ell,m},\cdot)}{\epsilon}\right)\right]\right),\qquad\text{ if }P(A_{k,n}\times B_{\ell,m})>0.
\end{equation}
When $P(A_{k,n}\times B_{\ell,m})=0$, the value of the dual potential $\phi_k(x_{\ell,m})$ is undetermined, as it is a point outside the support of the marginal, and irrelevant to the dual objective. Thus, it can be set to any convenient value, for instance zero. The value of \eqref{eq:eot_value_discretized} and the log-density of the optimal coupling $\nu^\ast_{\epsilon,\delta,n,m}\in\Pi_{n,m}(\mathbf{P},\lambda)$ are then recovered from the solution of this system by the discretized analogs of the formulas in Proposition~\ref{prop:entropic-ot-dual}:
\begin{equation}\label{eq:eot_fullydiscretized_value}
u_{\epsilon,\delta,n,m}(\mathbf{P},\lambda) =E_{r_{n,m}}[\mathrm{T}_{n}\phi] = \sum_{k\in[n]} \sum_{\ell\in[m]} \phi_k(x_{\ell,m})\, P(A_{k,n}\times B_{\ell,m}),
\end{equation}
\begin{equation}\label{eq:eot_fullydiscretized_coupling}
  \log\!\left(\frac{\nu^\ast_{\epsilon,\delta,n,m}}{r_{n,m}}\right)\!(x_{\ell_1,m},\ldots,x_{\ell_n,m}) = \frac{\sum_{k\in[n]}\phi_k(x_{\ell_k,m})\lambda_{k,n}-c_{\delta,n}(x_{\ell_1,m},\ldots,x_{\ell_n,m})}{\epsilon}, \text{ for all } \ell_1,\ldots,\ell_n\in[m].
\end{equation}

These results, together with the rate at which \eqref{eq:eot_value_discretized} approximates \eqref{eq:eot_value}, are collected in the following proposition.

\begin{proposition}\label{prop:fully-discretized-ot}
Under Assumption~\ref{ass:PI} and \ref{ass:constraints}, the entropic optimal transport problem \eqref{eq:eot_value_discretized} is well-posed and strong duality holds, that is, its value equals that of the dual
\[u_{\epsilon,\delta,n,m}(\mathbf{P},\lambda)=\sup_{\phi\in C_b(\X)^n}E_{r_{n,m}}\!\left[\mathrm{T}_n\phi-\epsilon\left(e^{\frac{\mathrm{T}_n\phi-c_{\delta,n}}{\epsilon}}-1\right)\right],\]
attained by a maximizer $\phi=(\phi_1,\ldots,\phi_n)$ satisfying the first-order condition \eqref{eq:eot_fullydiscretized_foc}. The value of the program and the log-density of the optimal coupling $\nu^\ast_{\epsilon,\delta,n,m}\in\Pi_{n,m}(\mathbf{P},\lambda)$ relative to $r_{n,m}$ are recovered, respectively, from \eqref{eq:eot_fullydiscretized_value} and \eqref{eq:eot_fullydiscretized_coupling}. Moreover, for every fixed $(\epsilon,\delta,n)$,
\[|u_{\epsilon,\delta,n,m}(\mathbf{P},\lambda)-u_{\epsilon,\delta,n}(\mathbf{P},\lambda)|\leq \frac{\rho_m}{\delta},\]
so that $u_{\epsilon,\delta,n,m}(\mathbf{P},\lambda)\to u_{\epsilon,\delta,n}(\mathbf{P},\lambda)$ as the mesh $\rho_m$ of the partition $\mathcal{B}_m$ vanishes. If, in addition, $\frac{\rho_m}{\delta}\to0$, $\frac{\epsilon}{\delta}\to0$, and $n\to\infty$ jointly, then
\[u_{\epsilon,\delta,n,m}(\mathbf{P},\lambda)\to v(\mathbf{P},\lambda).\]
\end{proposition}

\begin{proof}
    See Appendix~\ref{app:proof_fully_discretized_ot}.
\end{proof}

Proposition~\ref{prop:fully-discretized-ot} reduces the discretized entropic optimal transport problem to a genuinely \textbf{finite task}: solving the system \eqref{eq:eot_fullydiscretized_foc} of at most $n\cdot m$ equations in as many unknowns for the potentials $\phi$, from which both the value $u_{\epsilon,\delta,n,m}(\mathbf{P},\lambda)$ and the optimal coupling $\nu^\ast_{\epsilon,\delta,n,m}$ follow in closed form via \eqref{eq:eot_fullydiscretized_value} and \eqref{eq:eot_fullydiscretized_coupling}. This finite-dimensional character is more than an incidental simplification, but a genuine feature of the discretized problem, and the remainder of the paper -- in particular, the estimator developed in Section~\ref{sec:estimation} -- is built directly on top of it. We conclude this section making this structure explicit, recasting problem \eqref{eq:eot_value_discretized} as a maximization over a matrix of dual potentials, relating its first-order optimality condition to a system of non-linear equations and showing how to approximately solve it using the Sinkhorn algorithm.

First, note that the discretized marginal profile $(P_{k,n,m})_{k\in[n]}$ and the cell weights $(\lambda_{k,n})_{k\in[n]}$ are together described by a pair $(p,\lambda)$, where $p$ is an $n\times m$ matrix in $(\Delta_m)^n$, i.e., each row lies in the $m$-simplex $\Delta_m$, and a vector of weights $\lambda\in\Delta_n$. Each entry $p_{k,\ell}$ collects the conditional probability that the endogenous variable falls in the cell $B_{\ell,m}$ given that the instrument lies in $A_{k,n}$, $p_{k,\ell} = P_{k,n,m}(x_{\ell,m}) = P_{k,n}(B_{\ell,m})$,
while each component of the vector $\lambda_{k,n}=\lambda(A_{k,n})$ collects the probability that the instrument itself lies in $A_{k,n}$.

Second, since expectation in \eqref{eq:eot_value_discretized} involves $\phi_k$ only at the $m$ points of $\X_m$, every dual potential can likewise be identified with a matrix $\phi_{k,\ell}:=\phi_k(x_{\ell,m})$ specifying the value of each potential at each grid point. In this case, the vector of continuous functions $\phi=(\phi_1,\ldots,\phi_n)$, for the purpose of this program, is equivalent to a matrix $\phi\in\R^{n\times m}$.

Finally, the cost function $c_{\delta,n}$ enters problem \eqref{eq:eot_value_discretized} only through its values on the $m^n$ points of the grid $\X_m^n$. This means that the exponentiated-cost can be described as a tensor $C\in\R_+^{[m]^n}$, with entry
\[C(\ell)=\exp\!\left(-\frac{c_{\delta,n}(x_{\ell_1,m},\ldots,x_{\ell_n,m})}{\epsilon}\right)\]
for every array of indices $\ell=(\ell_1,\ldots,\ell_n)\in[m]^n$.

With this notation, maximizing the dual objective of Proposition~\ref{prop:fully-discretized-ot} over continuous functions is the same as maximizing over $n\times m$ real-valued matrices the function $\Phi:\R^{n\times m}\times(\Delta_m)^n\times\Delta_n\to\R$,
\[\Phi(\phi,p,\lambda)=\sum_{\ell\in[m]^n}\left(\lambda^T\phi_{\cdot,\ell}-\epsilon\left(C(\ell)\,e^{\frac{\lambda^T\phi_{\cdot,\ell}}{\epsilon}}-1\right)\right)\prod_{k\in[n]}p_{k,\ell_k},\]
where $\phi_{\cdot,\ell}=(\phi_{1,\ell_1},\ldots,\phi_{n,\ell_n})^T$ contains the value of every potential at the coordinates of $\ell$ (i.e. $u_{\epsilon,\delta,n,m}(\mathbf{P},\lambda)=\max_{\phi\in\R^{n\times m}}\Phi(\phi,p,\lambda))$.\footnote{We refer the reader to Appendix~\ref{subsec:eot_properties} for a detailed study of $\Phi$ and the value of this program as $p,\lambda$ varies. Here we only record the facts needed for computation.}

Beyond transforming the program into a finite-dimensional optimization, this description also allows us to pose the first-order optimality condition \eqref{eq:eot_fullydiscretized_foc} in matrix notation. It is this what makes the first-order system amenable to the Sinkhorn Algorithm and solvable numerically. Indeed, note that the first-order condition \eqref{eq:eot_fullydiscretized_foc} reads, for every \emph{active} cell $(k,\ell)$, that is,  the cells whose probability is $p_{k,\ell}>0$,
\begin{equation}\label{eq:foc-matrix}
\phi_{k,\ell}=\alpha_{k,\ell}(\phi_{-k},p,\lambda):=-\frac{\epsilon}{\lambda_{k,n}}\log\left(\sum_{\ell_{-k}\in[m]^{n-1}}\,C(\ell_k,\ell_{-k})\,e^{\frac{\lambda_{-k}^T\phi_{-k,\ell_{-k}}}{\epsilon}}\Big(\prod_{j\neq k}p_{j,\ell_j}\Big)\right),
\end{equation}
where $\ell_{-k}$ and $\phi_{-k,\ell_{-k}}$ are used to refer to the indices and potential values of the coordinates other than $k$. Equation~\eqref{eq:foc-matrix} can be understood as a search for a matrix that is a fixed point of the mapping $\alpha$. As for \emph{inactive} cells, $p_{k,\ell}=0$, the entry $\phi_{k,\ell}$ enters neither \eqref{eq:foc-matrix} nor the value, and is left undetermined. Thus, we can arbitrarily adopt the harmless convention of setting it to zero.

This convention for inactive entries is not, by itself, enough to pin down a unique solution of \eqref{eq:foc-matrix}. This happens as, similarly to any Kantorovich dual, potentials are generally identified only up to additive constants that cancel out when averaged against the $\lambda$ weight, $\lambda^T\phi_{\cdot,\ell}$. Following the multi-marginal entropic optimal transport literature, we remove this indeterminacy by normalizing each of the first $n-1$ rows of $\phi$. Under this normalization, each of these rows should have zero mean under its own conditional distribution, $p_{k,\cdot}^T\phi_{k,\cdot}=0$ for $k\in[n-1]$. The $n$-th row is left free, absorbing the one remaining degrees of freedom. Together with the convention on inactive cells, this singles out the subspace
\[\mathcal{K}_p=\left\{\phi\in\R^{n\times m}: p_{k,\cdot}^T\phi_{k,\cdot}=0 \text{ for all }k\in[n-1], \text{ and }\phi_{k,\ell}=0 \text{ whenever }p_{k,\ell}=0\right\}.\]

Subspace $\mathcal{K}_p$ is more than a convenient choice of normalization, but the natural domain of the mapping $\alpha$ in \eqref{eq:foc-matrix} for which we can combine the geometric normalization and iterations of the map $\alpha$ to approach the unique solution of the program lying in this set. This idea is the foundation of the \textbf{Sinkhorn algorithm}, where block coordinate descent is combined with normalization. Concretely, starting from an initial $t=0$ point in $\mathcal{K}_p$, say $\phi^0=0\in \mathcal{K}_p$, the $t+1$-th iteration updates, for all active pair $(k,\ell)$,
\[\phi^{t+1}_{k,\ell}=\begin{cases}
\alpha_{k,\ell}(\phi^{t+1}_{1},\dots,\phi^{t+1}_{k-1},\phi^{t}_{k},\dots,\phi^{t}_n,p,\lambda)-p_{k,\cdot}^T\alpha_{k,\cdot}(\phi^{t+1}_{1},\dots,\phi^{t+1}_{k-1},\phi^{t}_{k},\dots,\phi^{t}_n,p,\lambda), & k<n,\\[4pt]
\alpha_{k,\ell}(\phi^{t+1}_{1},\dots,\phi^{t+1}_{n-1},\phi^{t}_{n},p,\lambda), & k=n,
\end{cases}\]
leaving every inactive entry at zero, where the right-hand side uses the already-updated rows $1,\ldots,k-1$ together with the previous iteration's rows $k,\ldots,n$ (a block-coordinate, or Gauss--Seidel, update). By construction $\phi^t\in \mathcal{K}_p$ for every $t$, as the recentering subtracted from the first $n-1$ rows enforces the normalization exactly, and the last row is left unnormalized, matching the definition of $\mathcal{K}_p$.

Each Sinkhorn iterate maximizes $\Phi(\cdot,p,\lambda)$ along one row of $\phi$, holding the rest fixed, so the sequence $\Phi(\phi^t,p,\lambda)$ increases monotonically toward the value of the program, $u_{\epsilon,\delta,n,m}(\mathbf{P},\lambda)$ -- and, in fact, does so at a geometric rate. If we cap the number of iterations at some value $I$, we obtain an approximate solution $\phi^{I}\in\mathcal{K}_p$ that is within a certain tolerance of the true solution. Associated to it, we can define a proxy for the value of the program, denoted by
\[u_{\epsilon,\delta,n,m,I}(\mathbf{P},\lambda)=\Phi(\phi^{I},p,\lambda),\]
evaluating the dual objective at the $I$-th iterate.
Translating the linear-convergence result of \citet{carlier2022linear} for the multi-marginal Sinkhorn algorithm into this problem's notation gives the following proposition.

\begin{proposition}\label{prop:sinkhorn-convergence}
Under Assumption~\ref{ass:PI} and \ref{ass:constraints}, let $(\phi^t)_{t\in[I]}$ be the sequence of Sinkhorn iterates defined above, started at $\phi^0=0\in \mathcal{K}_p$. If $\epsilon<\max\{\mathrm{diam}(\X),1\}$ and $I>0$, then,
\[|u_{\epsilon,\delta,n,m}(\mathbf{P},\lambda)-u_{\epsilon,\delta,n,m,I}(\mathbf{P},\lambda)| \ \leq\ \exp\left(-Ie^{-\frac{16n}{\delta \epsilon}\mathrm{diam}(\X)}\right).\]
\end{proposition}
\begin{proof}
    See Appendix~\ref{app:proof_sinkhorn_convergence}.
\end{proof}

Proposition~\ref{prop:sinkhorn-convergence} guarantees that a number of iterations $I$, chosen large enough that the right-hand side above falls below any desired tolerance, delivers an iterate $\phi^I$ whose objective value $u_{\epsilon,\delta,n,m,I}(\mathbf{P},\lambda)$ is within that tolerance of $u_{\epsilon,\delta,n,m}(\mathbf{P},\lambda)$, without ever solving the non-linear system \eqref{eq:foc-matrix} exactly. This is the departing point for the estimator developed in Section~\ref{sec:estimation}, since we can replace the population pair $(p,\lambda)$ by its sample analog $(\hat{p}^N,\hat{\lambda}^N)$ and running the same iteration on the empirical objective $\Phi(\cdot,\hat{p}^N,\hat{\lambda}^N)$. The resulting estimator can then be shown to be consistent, by suitably controlling the error introduced by each approximation, as well as asymptotically normal, when these parameters are held fixed.

\section{Estimation and inference}\label{sec:estimation}
So far, our treatment of the lower bound $v(\mathbf{P},\lambda)$ has been posed at the population level, where the marginal profile $\mathbf{P}=(P_z)_{z\in\Z}$ and the instrument distribution $\lambda$ are taken as known. In practice, a researcher observes only a sample of $N$ draws from these and must construct empirical analogs from it. This section lifts the population-level results of Section~\ref{sec:EOT} to this empirical setting. We show how the lower bound on the partially identified average $v(\mathbf{P},\lambda)$ can be estimated consistently from the plug-in estimator of the entropic optimal transport bound $u_{\epsilon,\delta,n,m,I}(\mathbf{P},\lambda)$, provided the regularization parameters $(\epsilon,\delta,n,m,I)$ are suitably selected. When these parameters are held fixed, the resulting plug-in estimator obeys a standard $\sqrt{N}$ central limit theorem, with identifiable asymptotic variance.

\subsection{Data and Empirical Marginals}\label{sec:dgp}

We observe an i.i.d.\ sample $\{(Z_i,X_i)\}_{i\in[N]}$ of $N$ joint realizations of the instrument and the endogenous variable, valued in $\Z\times\X$, with joint law $P\in\P(\Z\times\X)$ as in Assumption~\ref{ass:PI} and instrument marginal $\lambda\in\P(\Z)$. Our goal in this section is to use this sample -- without imposing further structure, such as a parametric form, on either distribution -- to construct empirical estimates of the discretized conditional marginals $P_{k,n,m}$ and the cell weights $\lambda_{k,n}$ of Section~\ref{sec:x_discretization}. The resolution and regularization parameters $\beta=(\epsilon,\delta,n,m,I)$ are all held fixed throughout this section. Only in Section~\ref{sec:consistency} do we let them vary with the sample size $N$, in order to obtain a consistent estimator of the lower bound.

To estimate the instrument marginal, note that each cell weight $\lambda_{k,n}$ can be consistently estimated, without additional structure, by the share of the sample landing in the corresponding cell. For each $k\in[n]$, let $N_k$ denote the number of instrument realizations falling in $A_{k,n}$, $N_k = \sum_{i\in[N]} \mathbf{1}_{[Z_i\in A_{k,n}]}$,
and collect the resulting empirical cell probabilities in the vector $\hat{\lambda}^N\in\Delta_n$, with $\hat{\lambda}^N_k:=\frac{N_k}{N}$.

To estimate the discretized endogenous marginals, we need the conditional probability that the endogenous variable falls in cell $B_{\ell,m}$ given that the instrument falls in cell $A_{k,n}$, for every combination of cells. Again absent further structure on the conditional law, the natural estimate is the corresponding empirical frequency. Let $N_{k,\ell}$ denote the number of sample observations falling in the joint cell $A_{k,n}\times B_{\ell,m}$, $N_{k,\ell} = \sum_{i\in[N]} \mathbf{1}_{[(Z_i,X_i)\in A_{k,n}\times B_{\ell,m}]}$. The empirical conditional probability of $B_{\ell,m}$ given $A_{k,n}$ is the share of the observations in $A_{k,n}$ that also fall in $B_{\ell,m}$. We collect these estimates in the random matrix $\hat{p}^N\in(\Delta_m)^n$, with $\hat{p}^N_{k,\ell}:=\frac{N_{k,\ell}}{N_k}$.

Figure~\ref{fig:sample-cells} illustrates the resulting construction, and Appendix~\ref{subsec:sample_properties} collects the asymptotic properties of these estimates -- consistency, joint asymptotic normality at rate $\sqrt{N}$, and an $L^1$-error bound -- used in the estimation results below.

\begin{figure}[H]
\centering
\begin{subfigure}[b]{0.47\textwidth}
\centering
\begin{tikzpicture}[x=3.0cm,y=3.0cm]
  \draw[thick] (0,0) rectangle (1,1);
  \draw[thick] (0.22,0) -- (0.22,1);
  \draw[thick] (0.48,0) -- (0.48,1);
  \draw[thick] (0.75,0) -- (0.75,1);
  \draw[thick] (0,0.32) -- (1,0.32);
  \draw[thick] (0,0.68) -- (1,0.68);
  \foreach \p in {(0.04,0.38),(0.10,0.45),(0.06,0.55),(0.15,0.42),(0.09,0.62)}
    \fill[blue] \p circle (1.1pt);
  \foreach \p in {(0.05,0.72),(0.12,0.80),(0.07,0.90),(0.16,0.95)}
    \fill[blue] \p circle (1.1pt);
  \foreach \p in {(0.27,0.05),(0.33,0.12),(0.29,0.20),(0.40,0.08),(0.36,0.25)}
    \fill[blue] \p circle (1.1pt);
  \foreach \p in {(0.28,0.38),(0.35,0.48),(0.30,0.58),(0.42,0.44)}
    \fill[blue] \p circle (1.1pt);
  \foreach \p in {(0.27,0.72),(0.34,0.80),(0.30,0.90),(0.40,0.85),(0.37,0.95)}
    \fill[blue] \p circle (1.1pt);
  \foreach \p in {(0.53,0.06),(0.60,0.14),(0.56,0.22),(0.68,0.10)}
    \fill[blue] \p circle (1.1pt);
  \foreach \p in {(0.54,0.38),(0.62,0.48),(0.57,0.58),(0.70,0.44),(0.65,0.62)}
    \fill[blue] \p circle (1.1pt);
  \foreach \p in {(0.55,0.75),(0.63,0.85),(0.58,0.93)}
    \fill[blue] \p circle (1.1pt);
  \foreach \p in {(0.80,0.06),(0.87,0.14),(0.83,0.22),(0.94,0.10)}
    \fill[blue] \p circle (1.1pt);
  \foreach \p in {(0.80,0.40),(0.88,0.50),(0.83,0.60),(0.95,0.45),(0.90,0.62)}
    \fill[blue] \p circle (1.1pt);
  \node[below] at (0.11,-0.02) {\scriptsize $A_{1,4}$};
  \node[below] at (0.35,-0.02) {\scriptsize $A_{2,4}$};
  \node[below] at (0.615,-0.02) {\scriptsize $A_{3,4}$};
  \node[below] at (0.875,-0.02) {\scriptsize $A_{4,4}$};
  \node[left] at (-0.02,0.16) {\scriptsize $B_{1,3}$};
  \node[left] at (-0.02,0.50) {\scriptsize $B_{2,3}$};
  \node[left] at (-0.02,0.84) {\scriptsize $B_{3,3}$};
\end{tikzpicture}
\caption{Sample realizations across the joint grid $\mathcal{A}_n\times\mathcal{B}_m$}
\end{subfigure}
\hfill
\begin{subfigure}[b]{0.47\textwidth}
\centering
\renewcommand{\arraystretch}{1.3}
\begin{tabular}{c|ccc|c}
 & $B_{1,3}$ & $B_{2,3}$ & $B_{3,3}$ & $\hat\lambda^N_k$\\\hline
$A_{1,4}$ & $0.00$ & $0.56$ & $0.44$ & $0.20$\\
$A_{2,4}$ & $0.35$ & $0.29$ & $0.36$ & $0.32$\\
$A_{3,4}$ & $0.33$ & $0.42$ & $0.25$ & $0.27$\\
$A_{4,4}$ & $0.44$ & $0.56$ & $0.00$ & $0.21$\\
\end{tabular}
\caption{The matrix $\hat p^N$ (rows sum to $1$), with $\hat\lambda^N$ in the right-most column}
\end{subfigure}
\caption{From sample realizations to the pair $(\hat p^N,\hat\lambda^N)$. Panel (a): the sample (blue) scattered across the cells of the grid of Figure~\ref{fig:x-discretization}(a), with $n=4$ instrument cells and $m=3$ endogenous cells. Cells $(A_{1,4},B_{1,3})$ and $(A_{4,4},B_{3,3})$ receive no observations, so $\hat p^N_{1,1}=\hat p^N_{4,3}=0$ -- inactive for the Sinkhorn algorithm. Panel (b): the resulting empirical conditional probabilities $\hat p^N_{k,\ell}$ and cell weights $\hat\lambda^N_k$.}
\label{fig:sample-cells}
\end{figure}

\subsection{Lower Bound Estimation and Sinkhorn Algorithm}\label{sec:plug_in}

We are now in a position to turn the sample into an estimate of the partially identified average of $\Lambda$. Fix a vector of regularization parameters $\beta=(\epsilon,\delta,n,m,I)$, held fixed throughout this section, and to be chosen as a function of the sample size $N$ only once we turn to consistency, in Section~\ref{sec:consistency}. Plugging the empirical estimates $(\hat{p}^N,\hat{\lambda}^N)$ of Section~\ref{sec:dgp} into the first-order condition \eqref{eq:foc-matrix} in place of the population pair $(p,\lambda)$, and running $I$ iterations of the Sinkhorn algorithm of Section~\ref{sec:x_discretization} on the resulting empirical objective $\Phi(\cdot,\hat p^N,\hat\lambda^N)$, delivers an estimate of the dual potential matrix, $\hat\phi^N\in\mathcal{K}_{\hat p^N}$. Evaluating the finite-dimensional objective $\Phi$ at these empirical quantities gives the 	\textbf{lower bound estimator}
\begin{equation}\label{eq:estimator}
\hat{v}_N(\mathbf{P},\lambda):=\Phi(\hat\phi^N,\hat{p}^N,\hat{\lambda}^N).
\end{equation}
As before, both $\alpha$ and $\Phi$ also depend on the regularization vector $\beta$, a dependence we leave implicit to reduce notational burden.

This estimator is computationally feasible precisely because of the Sinkhorn algorithm of Section~\ref{sec:x_discretization}. Rather than solving the $n\times m$-equation system \eqref{eq:foc-matrix} exactly, $I$ rounds of the block-coordinate update deliver an approximate solution whose objective value is, by Proposition~\ref{prop:sinkhorn-convergence}, within an explicit and controllable distance of the exact optimum. Importantly, the algorithm is guaranteed to converge regardless of the sample realization, thereby ensuring strong consistency. Algorithm~\ref{alg:sinkhorn} summarizes the resulting procedure, from the raw sample to the estimator $\hat{v}_N(\mathbf{P},\lambda)$.

\begin{algorithm}[H]
\caption{Sinkhorn algorithm for the lower bound estimator $\hat{v}_N(\mathbf{P},\lambda)$}
\label{alg:sinkhorn}
\begin{algorithmic}[1]
\Require Sample $\mathbb{X}=\{(Z_i,X_i)\}_{i\in[N]}$
\Statex 	\textbf{Initialization step:}
\State $N\gets$ sample size
\State $\beta\gets\beta(N)$ \Comment{select the regularization parameters}
\State Construct $C,\ \hat{p}^N,\ \hat{\lambda}^N$ \Comment{cost tensor, empirical marginal matrix, cell weights}
\State Define $\Phi,\ \alpha$ \Comment{objective and first-order map, given $(C,\hat p^N,\hat\lambda^N)$}
\Statex 	\textbf{Sinkhorn algorithm:}
\State $\phi\gets 0\in\R^{n\times m}$
\For{$t=1,\ldots,I$}
    \For{$k=1,\ldots,n$}
        \If{$k<n$}
            \For{$\ell=1,\ldots,m$}
                \If{$\hat{p}^N_{k,\ell}=0$}
                    \State move to $\ell+1$
                \Else
                    \State $\phi_{k,\ell}\gets\alpha_{k,\ell}(\phi,\hat{p}^N,\hat{\lambda}^N)-(\hat{p}^N_{k,\cdot})^T\alpha_{k,\cdot}(\phi,\hat{p}^N,\hat{\lambda}^N)$
                \EndIf
            \EndFor
        \Else
            \For{$\ell=1,\ldots,m$}
                \If{$\hat{p}^N_{k,\ell}=0$}
                    \State move to $\ell+1$
                \Else
                    \State $\phi_{k,\ell}\gets\alpha_{k,\ell}(\phi,\hat{p}^N,\hat{\lambda}^N)$
                \EndIf
            \EndFor
        \EndIf
    \EndFor
\EndFor
\Statex 	\textbf{Compute the resulting value:}
\State $\hat{v}_N(\mathbf{P},\lambda)\gets\Phi(\phi,\hat{p}^N,\hat{\lambda}^N)$
\Statex 	\textbf{Output:} $\hat{v}_N(\mathbf{P},\lambda)$
\end{algorithmic}
\end{algorithm}

With $\hat{v}_N(\mathbf{P},\lambda)$ in hand, the question left is whether it approximates the targeted value. The next section shows that, for a suitable choice of the regularization parameters $\beta(N)$, the Algorithm~\ref{alg:sinkhorn} yields a consistent estimator of the lower bound $v(\mathbf{P},\lambda)$.

\subsection{Consistency}\label{sec:consistency}

The route taken so far is long, built over a sequence of regularizations grounded in the optimal transport reformulation of Section~\ref{sec:EOT}. Tracking the effect that each of these regularizations has on the resulting finite-dimensional program let us construct, in Section~\ref{sec:plug_in}, an estimator of the partially identified bound $v(\mathbf{P},\lambda)$ that is numerically feasible. What remains is to show that these effects can be controlled jointly. That is, we show that a suitable choice of the regularization vector $\beta=(\epsilon,\delta,n,m,I)$, as a function of the sample size, drives $\hat{v}_N(\mathbf{P},\lambda)$ to the population value it targets. This is the content of the following result.

\begin{theorem}\label{thm:consistency}
    Under Assumptions~\ref{ass:PI} and \ref{ass:constraints}, suppose the vector of regularization parameters $\beta=(\epsilon,\delta,n,m,I)$ satisfies:
    \[\frac{e^{\frac{16n}{\delta \epsilon}\mathrm{diam}(\X)}}{I},\frac{\epsilon}{\delta},\frac{\rho_m}{\delta}=o(1),\quad\text{ and }\quad \frac{\overline{\lambda}_n^{\frac{1}{2}}nm}{\delta}=o\left(N^{\frac{1}{2}}\right).\]
    Then, the estimator $\hat{v}_N$ defined in equation~\eqref{eq:estimator} is consistent in $L^1$, that is:
    \[\hat{v}_N(\mathbf{P},\lambda)\xrightarrow{L^1}v(\mathbf{P},\lambda).\]
\end{theorem}

\begin{proof}
    See Appendix~\ref{app:proof_consistency}.
\end{proof}

Theorem~\ref{thm:consistency} is, to the best of our knowledge, the first statistical consistency result of its kind for a partially identified bound built from a regularized, multi-marginal entropic optimal transport problem. Its rate conditions make transparent how each stage of the construction in Sections~\ref{sec:x_discretization}--\ref{sec:plug_in} feeds a single asymptotic requirement The Sinkhorn error must vanish relative to the entropic and cell-width regularizations ($\frac{\exp\left(\frac{16n}{\delta\epsilon}\mathrm{diam}(\X)\right)}{I}=o(1)$), the entropic penalty must vanish relative to the localization parameter ($\frac{\epsilon}{\delta}=o(1)$), the endogenous discretization mesh must vanish relative to $\delta$ ($\frac{\rho_m}{\delta}=o(1)$), and, jointly, the exogenous and endogenous grids $n,m$ cannot refine faster than the sample size accumulates information, relative to $\delta$. Rather than requiring each regularization parameter to vanish on its own, the theorem shows that they must be selected \emph{at compatible relative rates}, letting a single tuning sequence $\beta(N)$ drive every source of approximation error (discretization, entropic smoothing, localization, and finite-sample Sinkhorn truncation) to zero together. Section~\ref{sec:asymptotics} builds on this same logic to characterize the estimator's limiting distribution, under an analogous but sharper choice of $\beta(N)$.

\subsection{Asymptotic distribution}\label{sec:asymptotics}

The previous section shows that letting the regularization vector $\beta=(\epsilon,\delta,n,m,I)$ vary suitably with the sample size drives the estimator $\hat{v}_N(\mathbf{P},\lambda)$ to the population lower bound $v(\mathbf{P},\lambda)$. The natural next step, and our task in this section, is to characterize the asymptotic distribution of the estimator.

Differently from Section~\ref{sec:consistency}, however, we now hold $\beta=(\epsilon,\delta,n,m)$ fixed and let only the number of Sinkhorn iterations $I$ grow with the sample size $N$. There are at least three prominent reasons for this choice. First, with $(\epsilon,\delta,n,m)$ fixed, the estimator value is simply an ``almost'' maximum of a strongly concave, coercive objective over a finite-dimensional subspace, that is, an $M$-estimator. For this reason, standard $M$-estimation arguments and a direct application of the delta method deliver a normal limit. Moreover, a closed-form asymptotic variance can be identified by the weighted population variance of the entropic optimal transport potentials. Second, in a substantial part of the applied literature the joint distribution of the instrument and the endogenous variable is already discrete and finitely supported, so that $n,m$ are \emph{de facto} fixed and holding them fixed here imposes no additional restriction. Third, when $(\epsilon,\delta,n,m)$ are instead allowed to vanish (respectively, diverge) with $N$, as in Section~\ref{sec:consistency}, the limiting distribution of the estimator depends on the particular rate at which they do so, and is not, to our knowledge, yet fully understood in the multi-marginal optimal transport literature. We leave this more general case for future research and refer the reader to our discussion in Section~\ref{sec:discussion}.

Throughout this section, $\beta=(\epsilon,\delta,n,m)$ is fixed. Write $\hat{\phi}^N_\beta\in\mathcal{K}_{\hat{p}^N}$ for the dual potential matrix obtained after $I$ Sinkhorn iterations on the empirical marginals, and $\phi^\ast_\beta\in\mathcal{K}_p$ for the (unique) dual potential matrix solving the first-order condition \eqref{eq:foc-matrix} at the population pair $(p,\lambda)$. Each row of the matrices $\hat{\phi}^N_\beta$ and $\phi^\ast_\beta$ records the value of a potential function at the grid points of $\X_m$, points sampled according to $\hat{p}^N_{k,\cdot}$ and $p_{k,\cdot}$, respectively. This lets us associate to each row a variance, and hence to each potential matrix an instrument $\lambda$-weighted aggregate variance,
\begin{align*}
    \textnormal{Var}_{(p,\lambda)}[\phi^\ast_\beta]&:=\sum_{k\in[n]}\textnormal{Var}_{p_k}[\phi^\ast_{\beta,k}]\lambda_{k,n}=\sum_{k\in[n]} \left((\phi^\ast_{\beta,k,\cdot})^{2} \cdot p_{k,\cdot} - (\phi^\ast_{\beta,k,\cdot}\cdot p_{k,\cdot})^2\right)\lambda_{k,n},\\
    \textnormal{Var}_{(\hat{p}^N,\hat{\lambda}^N)}[\hat{\phi}^N_\beta]&:=\sum_{k\in[n]}\textnormal{Var}_{\hat{p}^N_k}[\hat{\phi}^N_{\beta,k}]\hat{\lambda}^N_{k,n}=\sum_{k\in[n]} \left((\hat{\phi}^N_{\beta,k,\cdot})^{2} \cdot \hat{p}^N_{k,\cdot} - (\hat{\phi}^N_{\beta,k,\cdot}\cdot \hat{p}^N_{k,\cdot})^2\right)\hat{\lambda}^N_{k,n},
\end{align*}
where we adopt the convention of interpreting $a_{k,\cdot}^2=(a_{k,1}^2,\ldots,a_{k,m}^2)$ for any matrix $a\in\R^{n\times m}$. Each aggregate variance is the sum, across rows, of the variance of each $k$-th potential function under its own conditional law, weighted by the cell weight $\lambda_{k,n}$. As the next result shows, it is exactly this quantity that controls the asymptotic variance of the estimator $\hat{v}_N$.

\begin{theorem}\label{thm:asymptotic_distribution}
    Under Assumptions~\ref{ass:PI} and \ref{ass:constraints}, and for a fixed vector of regularization parameters $\beta=(\epsilon,\delta,n,m)$, if the Sinkhorn iterations $I$ grow faster than $\ln N$, so that $\frac{\ln N}{I} = o(1)$, then the estimator $\hat{v}_N$ defined in equation~\eqref{eq:estimator} is asymptotically normal:
    \[\sqrt{N}\left(\hat{v}_N(\mathbf{P},\lambda)-u_{\epsilon,\delta,n,m}(\mathbf{P},\lambda)\right)\xrightarrow{d}\mathcal{N}\left(0,\textnormal{Var}_{(p,\lambda)}[\phi^\ast_\beta]\right).\]
    Moreover, the variance $\textnormal{Var}_{(p,\lambda)}[\phi^\ast_\beta]$ can be consistently estimated by $\textnormal{Var}_{(\hat{p}^N,\hat{\lambda}^N)}[\hat{\phi}^N_\beta]$ so that:
    \[\textnormal{Var}_{(\hat{p}^N,\hat{\lambda}^N)}[\hat{\phi}^N_\beta]\rightarrow\textnormal{Var}_{(p,\lambda)}[\phi^\ast_\beta]\textnormal{ a.s.}\]
\end{theorem}

\begin{proof}
    See Appendix~\ref{app:proof_asymptotic_distribution}.
\end{proof}

Theorem~\ref{thm:asymptotic_distribution} delivers on the promise made at the outset of this section (Section~\ref{sec:estimation}): a standard $\sqrt{N}$ central limit theorem, with a plug-in-estimable asymptotic variance. This means that, for a fixed choice of the regularization parameters $\beta=(\epsilon,\delta,n,m)$, the limiting distribution of our estimator $\hat{v}_N(\mathbf{P},\lambda)$ is known and can be used to construct estimates of the regularized value $u_{\epsilon,\delta,n,m}(\mathbf{P},\lambda)$, confidence intervals, and hypothesis tests. 

The asymptotic variance, consistently estimated by $\textnormal{Var}_{(\hat{p}^N,\hat{\lambda}^N)}[\hat{\phi}^N_\beta]$, can be interpreted as measuring the following variation. Fixed a cell $A_{k,n}$ of the instrument grid, the $k$-th row of the dual potential matrix $\phi^\ast_\beta$ is a function defined on the endogenous grid $\X_m$. This function collects, at each grid point $x_\ell \in\X_m$, the approximated expected cost of the effect function $\Lambda$ conditioned on paths falling in $B_{\ell,m}$ for instrument realizations in $A_{k,n}$ under the optimal distribution path distribution. Its variance then measures how much the conditional cost varies across the support of the endogenous variable. The weighted sum of these variances, across all instrument cells, gives the overall asymptotic variance of the estimator $\hat{v}_N(\mathbf{P},\lambda)$. Next section, we show how to, with additional assumptions, employ this result in the construction of confidence intervals and hypothesis tests for the partially identified bound $v(\mathbf{P},\lambda)$ itself.





\section{Applications, Limitations, and Future Directions}\label{sec:discussion}

This section discusses uses of the estimation results of Section~\ref{sec:estimation} for empirical practice, the main limitations they inherit from the underlying framework, and the directions for future research it suggests.

\noindent\textbf{Confidence intervals for the partially identified bound.} The asymptotic distribution result of Theorem~\ref{thm:asymptotic_distribution} is stated for the regularized value $u_{\epsilon,\delta,n,m}(\mathbf{P},\lambda)$, not the population bound $v(\mathbf{P},\lambda)$ itself, since $\beta$ is held fixed throughout Section~\ref{sec:asymptotics}. Suppose, however, that a rate function $\kappa(\beta)$ is known for the (deterministic) population bias this regularization introduces,
\[|u_{\epsilon,\delta,n,m}(\mathbf{P},\lambda)-v(\mathbf{P},\lambda)|\leq \kappa(\beta).\]
Fix $\beta$ and, to lighten notation, write $\widehat{\mathrm{se}}_N(\beta):=\sqrt{\frac{\textnormal{Var}_{(\hat{p}^N,\hat{\lambda}^N)}[\hat{\phi}^N_\beta]}{N}}$ for the estimated standard error of $\hat{v}_N$ from Theorem~\ref{thm:asymptotic_distribution}. Then, for any level $\tau\in(0,1)$, an asymptotically valid confidence interval at this level for $v(\mathbf{P},\lambda)$, rather than merely for $u_{\epsilon,\delta,n,m}(\mathbf{P},\lambda)$, is given by
\[\left[\hat{v}_N-\kappa(\beta)- z_{1-\frac{\tau}{2}}\,\widehat{\mathrm{se}}_N(\beta),\ \hat{v}_N +\kappa(\beta) + z_{1-\frac{\tau}{2}}\,\widehat{\mathrm{se}}_N(\beta)\right],\]
where $z_{1-\frac{\tau}{2}}$ is the $(1-\frac{\tau}{2})$-quantile of the standard normal distribution. This is simply the usual Wald interval, widened on both sides by the known bias bound $\kappa(\beta)$. The rate function $\kappa(\beta)$ need not be optimal for this to be valid. In fact, a non-optimal rate can already be obtained from purely qualitative properties of the model, such as continuity of the dual solutions to the population problem, itself inherited from continuity of the effect function $\Lambda$, the mechanism $g$, and the regularity of the underlying spaces. Since these properties are specific to each application, generic formulas for the rate are beyond the scope of this work. However, it is worth mentioning that the sharper the rate is known to be, the tighter the resulting confidence interval.

A second, related determinant of the width of this interval is the asymptotic variance. Our results suggest that this quantity is sensitive to the non-negligibility of the regularization parameters $\epsilon$ and $\delta$. Indeed, the smaller these are, the larger the asymptotic variance tends to be. Intuitively, this happens as the parameters $\epsilon$ and $\delta$ control the smoothness of the entropic dual objective, and hence the stability of its maximizer, the dual potential $\phi^\ast_\beta$. Then, the less smooth the objective, the more sensitive $\phi^\ast_\beta$ is to small perturbations of the data, and the larger its sampling variance. Together with $\kappa(\beta)$'s own dependence on $\beta$, this points to a bias-variance trade-off in the choice of the regularization vector, where $\beta$ should be selected to balance the deterministic bias $\kappa(\beta)$ against its effect on the estimator's sampling variance, rather than to minimize both simultaneously.

\vspace{0.1cm}

\noindent\textbf{Hypothesis testing.} The same rate function $\kappa(\beta)$ also lets us test whether the partially identified bound exceeds a value of interest, $\underline{v}\in\R$, fixed for instance by economic theory or a policy threshold. Consider the null hypothesis $H_0: v(\mathbf{P},\lambda)>\underline{v}$ against the alternative $H_1: v(\mathbf{P},\lambda)\leq\underline{v}$, at significance level $\tau\in(0,1)$. Fix $\beta$ and let
\[\underline{v}_N := \underline{v}-\kappa(\beta)+ z_{\tau}\,\widehat{\mathrm{se}}_N(\beta).\]
The test $\hat{t}_N$ that rejects $H_0$ whenever $\hat{v}_N(\mathbf{P},\lambda)\leq \underline{v}_N$ is asymptotically valid at level $\tau$. Indeed, under $H_0$, its type-I error is controlled by
\[\mathbb{P}\left[\hat{t}_N \textnormal{ rejects } H_0 \mid H_0\right]\leq \mathbb{P}\left[\frac{\hat{v}_N(\mathbf{P},\lambda)-u_{\epsilon,\delta,n,m}(\mathbf{P},\lambda)}{\widehat{\mathrm{se}}_N(\beta)}\leq z_\tau \ \middle|\ H_0\right]\longrightarrow \tau,\]
where the convergence follows directly from Theorem~\ref{thm:asymptotic_distribution}.

\vspace{0.1cm}

\noindent\textbf{The curse of dimensionality.} These two applications rely on the estimator $\hat v_N(\mathbf{P},\lambda)$ being computable in the first place. When the joint support of the exogenous and endogenous variables is uncountable, Section~\ref{sec:ot_connection} rewrites the partially identified bound as an optimal transport problem over the infinite-dimensional path space $\mathbb{H}$. In this case, complete discretization of the support and regularization is what makes this problem computationally tractable. Nevertheless, it comes with the cost of introducing both a deterministic bias, of the kind $\kappa(\beta)$ above must account for, and a system of $n\cdot m$ non-linear equations, one for each cell of the joint grid, to be solved for the dual potentials. This system is defined through a cost tensor $C$ over discretized paths and can only ever be solved approximately, after $I$ iterations of the Sinkhorn algorithm. Although finite-dimensional, the space of discretized paths $\X_m^n$ can itself contain up to $m^n$ configurations, so that the cost tensor, and the computational burden of forming and iterating over it, can grow exponentially in the number $n$ of instrument cells $Z$.

A first, partial remedy comes from the model's own structure. The paths that matter for the transport problem are only those achievable through the mechanism function. This achievability restriction typically encodes a large number of \emph{shape restrictions}, such as Lipschitz continuity or monotonicity, as we introduce in the application and simulation sections, on the ``active'' domain of the cost tensor (the set of discretized paths whose cost is not prohibitively high, since every unachievable path is penalized at order $\delta^{-1}$). As the cost of paths enter in the first-order optimality system exponentiated (proportional to $e^{-\frac{c}{\epsilon}}$), higher cost leads to less relevance of the path in the first-order system. As a consequence, one can try to neglect the paths with high costs in this system, effectively decreasing the number of active equations to be solved. Similarly, since the likelihood of a given discretized path in $\X_m^n$ is a product of cell-wise likelihoods, if a particular path's likelihood is known to be negligible, we can then further shrink the set of paths that need to be tracked in practice by eliminating unlikely paths. In either case, the model's own structure can be used to reduce the effective dimension of the problem, and hence the size of the cost tensor, without changing the underlying discretization.

A second trade-off is more fundamental. The bias introduced by the endogenous discretization is controlled by its mesh $\rho_m$ (Proposition~\ref{prop:fully-discretized-ot}), which shrinks as the number of cells $m$ grows. However, a finer partition also enlarges the grid $\X_m^n$ over which the cost tensor is defined, so refining the discretization to reduce bias directly increases the computational burden it must overcome. In this case, better knowledge of the bias rate can be of paramount importance in the selection of the discretization parameters $(n,m)$. This is because the knowledge of the precise bias allows the researcher to choose a coarser discretization that keeps the cost tensor manageable while still controlling the introduced error.

Thirdly, and independently of the size of the cost tensor, the number of Sinkhorn iterations $I$ required to make the resulting truncation error small is itself sensitive to dimension. As Proposition~\ref{prop:sinkhorn-convergence} shows, it must grow at rate $\exp\left(\frac{n}{\delta\epsilon}\right)$ for the error bound to vanish. Importantly, this is not an incidental feature of our methodology, but an unavoidable feature of the problem. In fact, linear programs in general suffer from the difficulty of determining its proximal solutions via finite iterative methods. The entropic regularization improves this, by permitting the identification of a suitable unique solution via a system of non-linear equations. Nevertheless, the multi-marginal Sinkhorn algorithm used to ``solve'' it also suffers from this kind of curse of dimensionality in the number of marginals, as is well documented in the computational optimal transport literature. We refer the reader to \citet{peyre2019computational}, whose complexity bound for the algorithm scales exponentially in the number of marginals; \citet{lin2022complexity}, who derive matching bounds of order $m^n$ for approximating the multi-marginal transport problem itself; \citet{altschuler2022wasserstein}, who show that closely related multi-marginal problems are, in the worst case, NP-hard; and the references therein.

A further, promising direction is to exploit the model's structure to break the problem into smaller subproblems before applying the discretization and entropic perturbation proposed here. \citet{tan2026partial} do so in the Roy model, where the structure of the problem splits it into smaller, parallelizable optimal transports. More generally, this is possible when the mechanism $g$ and the effect $\Lambda$ are additively separable across blocks of the instrument $Z$, in which case the marginal constraints decouple and each block can be solved independently. Even absent exact separability, when the cost tensor admits a low-rank or otherwise structured approximation, the Sinkhorn iterations can be carried out without ever forming the full $m^n$ tensor \citep{strossner2023lowrank}, sharply reducing the cost per iteration.

Confronted by these unpleasant manifestations of the curse of dimensionality of the problem, we advocate for two practical responses. One is to fix $(\epsilon,\delta,n,m)$ and let only $I$ grow with the sample size, as we do in Section~\ref{sec:asymptotics}. As shown, this sidesteps the exponential rate at the cost of leaving the discretization bias uncontrolled. The other is to let $n$ grow, but slowly enough -- for instance, at the iterated-logarithm rate $n=O(\ln\ln N)$ -- to keep the resulting number of iterations manageable. This, however, demands a correspondingly large number of observations per instrument cell, which can undermine its practical viability for moderate sample sizes.

\vspace{0.1cm}

\noindent\textbf{Future directions.} Overcoming these limitations will require leaning on the shape restrictions already embedded in a given application (through the mechanism and the achievable path set it generates), which are the most promising lever for taming the curse of dimensionality. These shape restrictions shrink the active domain of the cost tensor rather than fight the exponential rate directly. We see as the most immediate directions for future work: (i) formalizing how much a given shape restriction reduces the effective dimension of the problem; (ii) identifying which structural features of the model, such as separability or low tensor rank, permit parallelized or low-rank computation; (iii) deriving sharp bias rates $\kappa(\beta)$ for the regularization, and how shape restrictions improve them; and (iv) extending the consistency and asymptotic distribution results of Section~\ref{sec:estimation} to the case where all of $\beta=(\epsilon,\delta,n,m,I)$ vary jointly with the sample size.

\section{Simulations}\label{sec:simulations}

This simulation section is intended to empirically validate and assess the theoretical statements from the previous sections. We focus on approximation accuracy, inference, and computational cost. We employ the two-stage selection model of Section~\ref{sec:examples}, varying its complexity. As the effect function, we pick the treatment effect, $\Lambda(\w)=\w_2(1)-\w_2(0)$. 

We start with the simple binary setting from \citet{balke1997bounds} phrased in our approach. Then we focus on a setting with continuous treatment-outcomes and binary instrument, where we take advantage of its lower computational burden to assess its statistical properties. Finally, we focus on a setting with continuous treatment-outcomes and a multi-valued instrument. 

\vspace{0.2cm}

\noindent\textbf{The \citet{balke1997bounds} benchmark.}

We first focus on the simple setup with binary instrument, treatment and outcome, giving $16$ unobserved latent types. As Table~\ref{tab:short-binary} depicts, at $(\delta,\epsilon)=(0.25,0.0002)$, the maximum absolute difference from the sharp linear program (LP) is $1.48\times10^{-4}$. 
\begin{table}[!ht]\centering
\caption{Binary population benchmark}\label{tab:short-binary}
{\small\begin{tabular}{lrr}
\toprule Design & Sharp LP interval & Regularized endpoints\\\midrule
Strong & $[0.405000, 0.705000]$ & $[0.405022, 0.704985]$\\
Weak & $[-0.120000, 0.780000]$ & $[-0.119852, 0.779897]$\\
Opposite effects & $[-0.090000, 0.270000]$ & $[-0.089883, 0.269875]$\\
\bottomrule\end{tabular}}
\par\medskip\begin{minipage}{.97\linewidth}\footnotesize Population marginals; all sixteen binary response types are admitted. Regularized entries use $\delta=0.25$ and $\epsilon=0.0002$.
\end{minipage}\end{table}

Figure~\ref{fig:binary-accuracy} varies both tuning parameters. It shows that reducing entropy produces approximately proportional error reductions at sufficiently strong support penalties, while a substantial error still persists at $\delta=0.75$. Removing entropy at $\delta=0.25$ gives the sharp bounds to full numerical precision.
\begin{figure}[!ht]\centering
\includegraphics[width=0.7\linewidth]{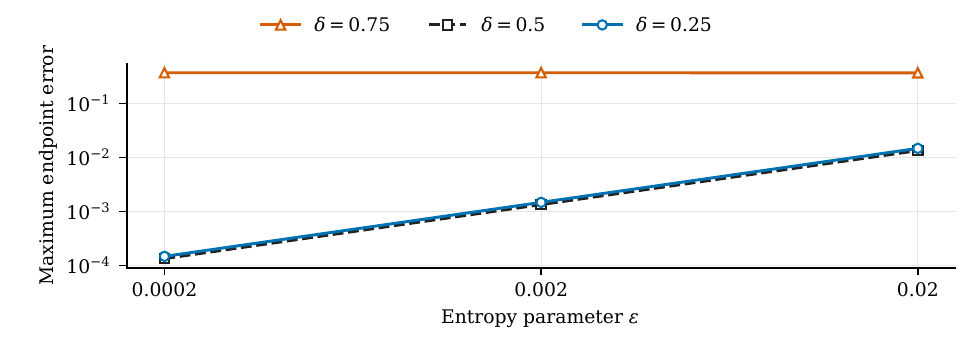}
\caption{Binary approximation error across tuning parameters}\label{fig:binary-accuracy}
\begin{minipage}{.97\linewidth}\footnotesize Each point is the maximum absolute difference from the sharp LP across all three designs and both endpoints. At $\epsilon=0$, the maximum errors for $\delta=0.75,0.5,0.25$ are $0.3733333$, $1.11\times10^{-16}$, and zero, respectively; zero entropy is excluded from the logarithmic axis. As proved in the theoretical section, entropy can partly offset the support-relaxation error, which explains the slightly rising error as $\epsilon$ decreases at $\delta=0.75$.
\end{minipage}\end{figure}

\vspace{0.3cm}

\noindent\textbf{Continuous outcome and treatment with binary instrument for statistical analysis.}

We now analyze the statistical properties of the estimator in the computationally easier setting where the instrument is binary, $\Z=\{0,1\}$, while both treatment and outcome take values in $\mathcal{D}=\mathcal{R}=[0,1]$. This is computationally easier because we only have two marginals and the cost function is just in a matrix form, not a tensor.

For the (simulated) treatment and outcome distributions, we let $U,V$ be independent uniforms on $[-1,1]$, independent of instrument, and generate the counterfactual outcomes as
\[
Y(d)=0.5+0.1U+0.1V+(-0.4+0.2U)(d-0.5)+0.03\sin(2\pi d).
\]
The treatment values are generated as
\[
D(0)=0.5+0.3U,\quad\text{ and }\quad D(1)=0.6-0.3U,
\]
and the instrument is a Bernoulli random variable with success probability $0.5$. Composition of the treatment with the potential outcome gives us the conditional distribution profile $\mathbf{P}$. Note that the potential treatment and outcome functions are $1$-Lipschitz in their respective inputs. Moreover, the resulting treatment effect is $-0.4+0.2U$, and its average equals $-0.4$.

As for the latent space of potential treatment and outcomes, we allow for arbitrary dependence between them, and we assume only $1$-Lipschitz first- and second-stage responses (i.e. $\Omega_1=\mathrm{Lip}_1(\Z,\mathcal{D})$ and $\W_2=\mathrm{Lip}_1(\mathcal{D},\mathcal{R})$). From the outset, we do not assume monotonicity. 

We use eight bins per coordinate represented at cell centers, $\delta=0.025$, and $\epsilon\in\{0.02,0.005\}$. For each $N\in\{500,2000,10000\}$, we simulate 500 datasets and construct the individual 95\% Wald intervals using the fixed-regularization central limit theorem and dual-potential standard errors derived before. As Table~\ref{tab:short-coverage} shows, the marginal coverage ranges from 94.4\% to 95.8\%. The average estimated standard errors divided by empirical standard deviations range from $0.949$ to $1.031$.

\begin{table}[H]\centering
\caption{Finite-sample inference for regularized population endpoints}\label{tab:short-coverage}
{\small\textit{A. Lower endpoint}}\par\smallskip
{\small\begin{tabular}{rrrrrr}
\toprule $N$ & $\epsilon$ & Bias $\times10^3$ & RMSE $\times10^3$ & SE/SD & Coverage\\\midrule
500 & 0.020 & 0.610 & 6.264 & 1.019 & 0.952\\
500 & 0.005 & 0.721 & 6.259 & 1.031 & 0.956\\
2,000 & 0.020 & 0.201 & 3.158 & 1.006 & 0.948\\
2,000 & 0.005 & 0.140 & 3.150 & 1.015 & 0.952\\
10,000 & 0.020 & -0.033 & 1.435 & 0.986 & 0.952\\
10,000 & 0.005 & -0.070 & 1.451 & 0.985 & 0.948\\
\bottomrule\end{tabular}}

\medskip
{\small\textit{B. Upper endpoint}}\par\smallskip
{\small\begin{tabular}{rrrrrr}
\toprule $N$ & $\epsilon$ & Bias $\times10^3$ & RMSE $\times10^3$ & SE/SD & Coverage\\\midrule
500 & 0.020 & 0.109 & 16.586 & 1.013 & 0.948\\
500 & 0.005 & 0.674 & 17.787 & 1.012 & 0.954\\
2,000 & 0.020 & -0.396 & 8.938 & 0.949 & 0.944\\
2,000 & 0.005 & -0.233 & 9.498 & 0.956 & 0.946\\
10,000 & 0.020 & 0.042 & 3.820 & 0.995 & 0.958\\
10,000 & 0.005 & 0.069 & 4.083 & 0.996 & 0.952\\
\bottomrule\end{tabular}}
\par\medskip\begin{minipage}{.97\linewidth}\footnotesize 500 replications per sample size; eight bins per coordinate; $\delta=0.025$. Bias and RMSE are multiplied by $10^3$. SE/SD is the average reported standard error divided by the Monte Carlo sample standard deviation. Coverage is for individual $95\%$ Wald intervals $\widehat v_\beta\pm1.96\,\widehat{\mathrm{se}}(\widehat v_\beta)$ based on the fixed-regularization central limit theorem, with standard errors estimated from the dual potentials. Coverage Monte Carlo standard errors are approximately $0.01$.
\end{minipage}\end{table}

\vspace{0.3cm}

\noindent\textbf{Continuous outcome and treatment with $5$ realizations for the instrument.}

Retain the treatment and outcome spaces of the preceding example, but consider an instrument with more than just two realizations (e.g. $\Z=\{0,\frac{1}{4},\frac{1}{2},\frac{3}{4},1\}$). Since our cost function in general does not factor nicely, this means that we need to consider the entire cost tensor of order $5$. This is computationally significantly more costly than binary instruments. 

The (simulated) outcome follows the exact same structure as in the previous example, while the treatment is generated as
\[D(z)=0.25+0.15U+(0.4-0.35U)z,\]
where $Z$ is uniformly distributed on $\Z$.

We consider two possible cases for the latent space, depending on the shape restriction imposed on the potential outcome functions. First, we assume that both stage responses are $1$-Lipschitz, with a nondecreasing first stage  potential treatment function -- a shape restriction satisfied by the true treatment response above. In our second specification, we restrict the former latent space by considering only the subset of latent profiles $(\w_1,\w_2)$ whose second stage responses are nonincreasing.

As Table~\ref{tab:short-five} shows, using observation rectangles and $(\delta,\epsilon)=(0.025,0.001)$, the outcome monotonicity reduces the twelve-bin outer interval width by 44.7\%, moving its upper endpoint from $0.5045$ to $-0.1377$. The strict separation from zero uses the data as well as the shape restriction. Figure~\ref{fig:five-benchmark} compares these coarsened-population outer intervals with the analytical sharp continuous-population benchmarks, which makes the remaining approximation error clearly visible. 

\begin{table}[H]\centering
\caption{Five-value instrument: ATE bounds}\label{tab:short-five}
{\small\begin{tabular}{rlr}
\toprule Bins & Outcome restriction & Outer interval\\\midrule
8 & Lipschitz & $[-0.9657, 0.5739]$\\
8 & + Nonincreasing & $[-0.9657, -0.1006]$\\
12 & Lipschitz & $[-0.9324, 0.5045]$\\
12 & + Nonincreasing & $[-0.9323, -0.1377]$\\
\bottomrule\end{tabular}}
\par\medskip\begin{minipage}{.97\linewidth}\footnotesize Both specifications have a nondecreasing, $1$-Lipschitz first stage. Outcome responses are $[0,1]$-valued and $1$-Lipschitz. The entries use $\delta=0.025$, $\epsilon=0.001$.
\end{minipage}\end{table}
\begin{figure}[H]\centering
\includegraphics[width=0.75\linewidth]{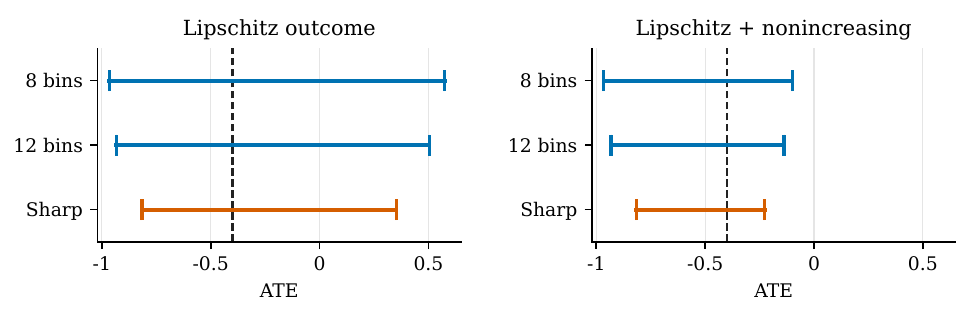}
\caption{Blue intervals are entropy-adjusted outer bounds for the coarsened quadrature population at $\delta = 0.025$, $\epsilon = 0.001$. Orange intervals are analytical sharp bounds for the original continuous population: $[-0.8158, 0.3516]$
and $[-0.8158, -0.2275]$. Dashed lines mark the generating ATE, $-0.4$. The analytical formulas are evaluated numerically. The eight- and twelve-bin partitions are not nested.}\label{fig:five-benchmark}
\end{figure}

The complete five-value comparison, including costs, two resolutions, both outcome specifications, and three entropy levels, took approximately three minutes on a MacBook Pro with a 2.3 GHz eight-core Intel Core i9-9880H processor and 16 GB of RAM. It used the full cost tensors of up to $227,136$ entries per sign.

\section{Application to Demand Estimation}\label{sec:application}
As an application, we revisit the food and leisure expenditure setting of \citet{gunsilius2019path}, using an analysis extract of the $1995/1996$ UK Family Expenditure Survey \citep{fes1995}. In line with previous literature, we concentrate on the fraction of households allocating at most a given budget share to each good, while focusing on the entire distribution of expenditure shares. Prior work has focused on Engel curves with endogenous expenditure \citep{blundell2007semi}, identification in nonseparable triangular models \citep{imbens2009identification}, and demand bounds under revealed preference \citep{blundell2008best}.
\begin{figure}[!ht]\centering
\includegraphics[width=0.75\linewidth]{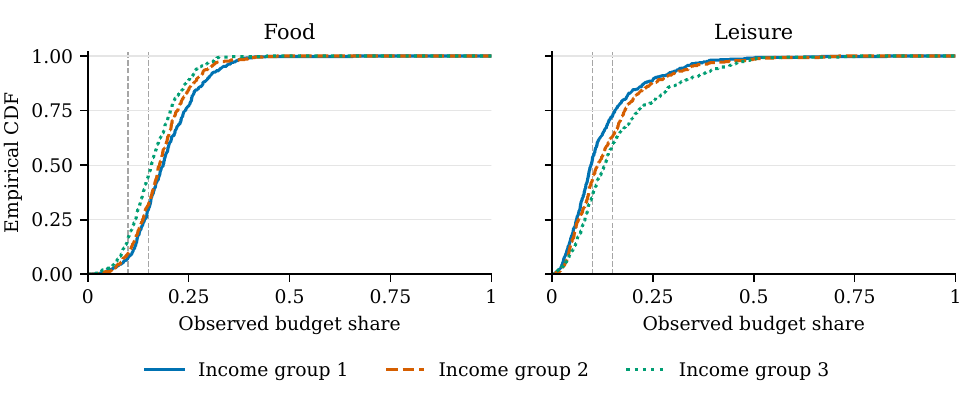}
\caption{Observed budget-share distributions across income groups}\label{fig:app-ecdf}
\begin{minipage}{.97\linewidth}\footnotesize Observed empirical conditional CDFs within the three income intervals in Table~\ref{tab:app-sample}, using 488, 434, and 398 estimation observations, respectively. Vertical lines mark the primary 10\% and 15\% thresholds.
\end{minipage}\end{figure}

The dataset we use consists of $1,650$ households with two adults, at most two children, and a self-employed or full-time-employed head aged 20--54. Using $330$ households as a calibration sample to determine the cutoffs for the instrument $Z$ (to have enough data in each bracket), we are left with $1,320$ households for our estimation. The outcome $Y$ is the food or leisure expenditure share, the treatment $D$ is normalized log total expenditure, and the instrument $Z$ is the normalized log-income of the head of the household. We decompose the latter into $3$ terciles, as this level of coarseness provides enough data in each section. Empirical CDFs are depicted in Figure \ref{fig:app-ecdf}. Following the previous literature \citep{blundell2007semi, imbens2009identification}, we assume the instrument is valid. Table~\ref{tab:app-sample} depicts the summary statistics of the sample. 
\begin{table}[!ht]\centering
\caption{Estimation sample by income interval}\label{tab:app-sample}
{\small\begin{tabular}{llrrrr}
\toprule Group & Income interval & $n$ & Mean $X$ & Food share & Leisure share\\\midrule
1 & $[0.0000,0.7763)$ & 488 & 0.458 & 0.195 & 0.126\\
2 & $[0.7763,0.8114)$ & 434 & 0.506 & 0.183 & 0.142\\
3 & $[0.8114,1.0000]$ & 398 & 0.578 & 0.164 & 0.166\\
All & $[0.0000,1.0000]$ & 1,320 & 0.510 & 0.182 & 0.143\\
\bottomrule\end{tabular}}\par\medskip
\begin{minipage}{.97\linewidth}\footnotesize The 330-household calibration sample sets the two internal income cutoffs for $Z$; the remaining $1,320$ are the data used. Cutoffs are rounded for display; group membership uses full precision. The last two columns are mean observed budget shares, not counterfactual quantities.
\end{minipage}\end{table}

Our target functional is the distributional contrast. Holding a budget-share threshold $q$ fixed, the targeted parameter is constructed from the effect function:
\[\Lambda(\w)=1_{\{\w_2(0.75)\leq q\}}-1_{\{\w_2(0.25)\leq q\}}.\]
Fixed a good and for each household's expenditure share function $\w_2$, the effect function evaluates whether the expenditure share crosses the threshold $q$ when moving from the lower to the higher expenditure level, $0.25$ to $0.75$. When averaged over the population of households, we obtain the distributional contrast
\[
\theta(q):=E_Q[\Lambda]=Q\{\w_2(0.75)\le q\}-Q\{\w_2(0.25)\le q\},\qquad q\in\{0.10,0.15\}.
\]
The reason we focus on this is analogous to the reasoning in \citet{gunsilius2019path}. If $\theta(q)$ is positive, it means that the higher expenditure intervention increases the probability of a share at or below $q$. A negative value, on the other hand, means the reverse. These correspond to the definitions of a necessity and luxury good, respectively. Focusing on food and leisure (a necessity and leisure good, respectively) we can then gauge if the corresponding bounds are informative in the sense that they validate that food is a necessity and leisure is luxury. This way, we can check with real data if the bounds give reasonable estimates. 

We assume that both the first and second stage latent variables are $1$-Lipschitz continuous. This baseline specification means that we do not make any monotonicity restrictions on either stage. As a secondary specification, we assume that each household's food share is nonincreasing, and the leisure share nondecreasing, in expenditure. These are monotonicity restrictions on the second stage, but we leave the first-stage direction unrestricted. They are put on the shares, so nonincreasing food shares do not require food spending itself to fall.

To compute the confidence intervals, we hold the income partition, observation grid, structural restrictions, distance penalty $\delta=0.025$, and positive entropy parameter fixed. 

Table~\ref{tab:app-asymptotic-sensitivity} reports the results and the strength of the respective monotonicity assumptions. In particular, without the monotonicity restriction of the second stage, the bounds contain $0$ and are largely uninformative about whether food is a necessity or luxury good. Once we impose the monotonicity on the second stage the bounds become strongly informative and do suggest, in line with economic intuition, that food is a necessity good and leisure is a luxury good. In particular, this is even true when considering the $95\%$ confidence intervals around each bound separately, as in Table~\ref{tab:app-asymptotic-sensitivity}.
\begin{table}[!ht]\centering
\caption{Fixed-entropy sensitivity: endpointwise 95\% Gaussian intervals ($\epsilon=0.001$)}\label{tab:app-asymptotic-sensitivity}
{\footnotesize\begin{tabular}{rllrrrr}
\toprule $q$ & Outcome & Model & $\widehat L_\epsilon$ & 95\% CI & $\widehat U_\epsilon$ & 95\% CI\\\midrule
0.10 & Food & L & -0.9529 & $[-0.9673, -0.9385]$ & 0.9958 & $[0.9911, 1.0005]$\\
0.10 & Food & L+$\downarrow$ & 0.0007 & $[0.0007, 0.0008]$ & 0.9954 & $[0.9907, 1.0001]$\\
0.10 & Leisure & L & -0.9834 & $[-0.9915, -0.9754]$ & 0.9113 & $[0.8851, 0.9375]$\\
0.10 & Leisure & L+$\uparrow$ & -0.9767 & $[-0.9853, -0.9681]$ & -0.0014 & $[-0.0015, -0.0014]$\\
0.15 & Food & L & -0.9756 & $[-0.9857, -0.9655]$ & 0.9987 & $[0.9973, 1.0000]$\\
0.15 & Food & L+$\downarrow$ & 0.0008 & $[0.0008, 0.0009]$ & 0.9983 & $[0.9970, 0.9996]$\\
0.15 & Leisure & L & -0.9888 & $[-0.9949, -0.9827]$ & 0.9313 & $[0.9087, 0.9539]$\\
0.15 & Leisure & L+$\uparrow$ & -0.9682 & $[-0.9803, -0.9560]$ & -0.0015 & $[-0.0015, -0.0014]$\\
\bottomrule\end{tabular}}\par\medskip
\begin{minipage}{.97\linewidth}\footnotesize L denotes Lipschitz assumptions; arrows denote outcome monotonicity. Each confidence interval concerns its own fixed-entropy endpoint and uses the normal critical value $1.96$.
\end{minipage}\end{table}

\section{Conclusion}

  This paper develops a general approach to computing and estimating sharp bounds on partially identified averages. The framework accommodates unobserved heterogeneity represented by general structural response functions, together with distributional restrictions indexed by a potentially continuous range of observed conditions. Under our regularity conditions, we show that the identification problem is equivalent to a constrained optimal transport problem on the space of admissible response paths. This representation preserves the identifying information in the observable distributions while making the role of structural assumptions explicit.

  We make this formulation computable through deterministic discretization, relaxation of structural support restrictions, and relative entropy regularization. The resulting problems can be solved by generalized Sinkhorn iterations, an area that is currently very active in the trajectory inference community \citep[e.g.][]{lavenant2024mathematical, schiebinger2019optimal}. Our statistical analysis establishes consistency for the original sharp bounds when the approximation parameters are suitably chosen. For fixed regularization and discretization, we also establish a central limit theorem with an estimable asymptotic variance. The simulations illustrate the computational implementation and the contribution of shape restrictions, while the household expenditure application demonstrates how the method can provide informative bounds.

  Several directions remain open. First, computation becomes increasingly demanding as the number of discretizations in $\Z$ grows in practice, which motivates settings that exploit additional structure in the transport cost and admissible response paths such as the contribution in \citet{tan2026partial}. An alternative is to use low-rank tensor approximations for the general cost tensor or exploiting Markov structures. Another extension, which is a well-known and longstanding open problem in the literature on statistical optimal transport, is to obtain inference to allow discretization and regularization to vary with sample size.

\newpage 

\appendix

\section{Auxiliary Mathematical Results}\label{app:aux}

This appendix collects the auxiliary mathematical objects and supporting results used throughout the paper: the discretization projections (Section~\ref{subsec:projections}), the cost functions and associated correspondences (Section~\ref{subsec:costs-correspondences}), and the marginal operators together with the resulting constraint sets (Section~\ref{subsec:marginal-operators}). Throughout, we follow the notation and terminology introduced in the main text.

Let $\mathcal{H}$ be a separable Hilbert space with norm $|\cdot|_\mathcal{H}$, and let $\X\subseteq\mathcal{H}$ be a convex, compact subset, endowed with the subspace topology and norm $|\cdot|_\X$; under this norm $\X$ is a compact Polish space. Fixing the cell probabilities $(\lambda_{k,n})_{k\in[n]}$ of the partition $\mathcal{A}_n$ of Section~\ref{sec:discretization}, we metrize the product $\X^n$ by
\[|x-y|_{\X^n}=\sum_{k\in[n]}|x_k-y_k|_{\X}\,\lambda_{k,n}\,,\]
a distance under which $\X^n$ is likewise a compact Polish space; this is the metric with respect to which $\pi_n$ will be shown to act as an isometry below.

Write $\mathbb{L}^1(\Z,\lambda;\mathcal{H})$ for the Bochner space of $\mathcal{H}$-valued functions on the probability space $(\Z,\sigma(\Z),\lambda)$, equipped with the norm
\[\|e\|_{\mathbb{L}^1(\Z,\lambda;\mathcal{H})}=\int_\Z |e_z|_{\mathcal{H}}\,d\lambda(z),\qquad e\in\mathbb{L}^1(\Z,\lambda;\mathcal{H}),\]
where $\sigma(\Z)$ denotes its Borel $\sigma$-algebra; we refer the reader to \cite{Hytonen2016} for a comprehensive treatment of Bochner spaces and their properties.

The subset of functions taking values in $\X$, $\lambda$-almost surely, is denoted $\mathbb{H}$; endowed with the restriction of $\|\cdot\|_{\mathbb{L}^1(\Z,\lambda;\mathcal{H})}$, which, with some abuse of notation, we continue to write $\|\cdot\|_{\mathbb{H}}$, $\mathbb{H}$ is itself a Polish space -- the path space $L^1(\Z,\lambda;\X)$ of Section~\ref{sec:setup_objects}. The remaining objects used in this appendix are introduced as needed, in the three subsections below.

\subsection{Projections and Discretizations}\label{subsec:projections}

The projection $\pi_n$ of Section~\ref{sec:discretization} is the single device that reduces the endogenous path, the exogenous variable, and the marginal constraints simultaneously to a common finite resolution $n$; this subsection collects its properties, each doing specific work later in the paper. Idempotence and the resulting compactness of the range $\pi_n(\mathbb{H})$ make the discretized transport problem of Section~\ref{sec:discretization} well posed; the isometry between $\pi_n(\mathbb{H})$ and $\X^n$ is what turns that problem into a finite-dimensional program; compatibility with refinement, $\pi_n\circ\pi_{n+1}=\pi_n$; and convergence to the identity is what lets the discretization be undone in the limit.

Throughout this section, we fix a refining sequence of partitions $(\mathcal{A}_n)_{n\in\N}$ of $\Z$ as in Section~\ref{sec:discretization}: each $\mathcal{A}_n=\{A_{1,n},\ldots,A_{n,n}\}$ has $n$ cells of strictly positive measure, $\lambda_{k,n}:=\lambda(A_{k,n})>0$ for every $k\in[n]$. The partitions are ordered and nested, so that each cell $A_{k,n}$ of $\mathcal{A}_n$ is a union of cells of the finer partition $\mathcal{A}_{n+1}$. Their mesh vanishes,$\rho_n:=\max_{k\in[n]}\operatorname{diam}_\Z(A_{k,n})\ \rightarrow\ 0$ as $n\to\infty$. This rules out a partition that keeps refining in one region while leaving a cell of fixed size untouched forever, and is what makes the piecewise-constant approximation below track a path's local behavior, not only its cell count.

Define the projections $\pi_n:\mathbb{L}^1(\Z,\lambda;\mathcal{H})\rightarrow \mathbb{L}^1(\Z,\lambda;\mathcal{H})$ and $\tilde{\pi}_n:\mathbb{L}^1(\Z,\lambda;\mathcal{H})\rightarrow \mathcal{H}^n$, together with their components $\pi_{k,n}:\mathbb{L}^1(\Z,\lambda;\mathcal{H})\rightarrow \mathbb{L}^1(\Z,\lambda;\mathcal{H})$ and $\tilde{\pi}_{k,n}:\mathbb{L}^1(\Z,\lambda;\mathcal{H})\rightarrow \mathcal{H}$, by
\begin{align*}
    \tilde{\pi}_{k,n}(e) &= \frac{1}{\lambda_{k,n}}\int_{A_{k,n}} e_z\,d\lambda(z), & \tilde{\pi}_n(e) &= (\tilde{\pi}_{k,n}(e))_{k\in[n]},\\
    \pi_{k,n}(e)(z) &= \tilde{\pi}_{k,n}(e)\,1_{A_{k,n}}(z), & \pi_n(e)(z) &= \sum_{k\in[n]} \pi_{k,n}(e)(z).
\end{align*}

In words, $\tilde\pi_{k,n}(e)$ is the $\lambda$-conditional average of the path $e$ over $A_{k,n}$, and $\tilde\pi_n(e)$ collects these $n$ averages into a single point of $\mathcal{H}^n$; correspondingly, $\pi_n(e)$ patches these averages together into the piecewise-constant path equal to $\tilde\pi_{k,n}(e)$ on each cell $A_{k,n}$. Restricted to $\mathbb{H}\subseteq\mathbb{L}^1(\Z,\lambda;\mathcal{H})$, the range of $\pi_n$ is denoted $\pi_n(\mathbb{H})$: the set of all piecewise-constant functions on $\mathcal{A}_n$ taking values in $\X$.

\begin{proposition}\label{prop:properties-projections}
    Under Assumption~\ref{ass:PI}, the projections defined above satisfy the following properties.
    \begin{enumerate}
        \item The operators $\pi_n$, $\pi_{k,n}$, $\tilde{\pi}_n$, and $\tilde{\pi}_{k,n}$ are linear and continuous, with operator norm at most $1$.
        \item Restricted to $\mathbb{H}$, $\pi_n$ is idempotent, $\pi_n\circ\pi_n=\pi_n$, and its range $\pi_n(\mathbb{H})$ is a convex, compact subset of $\mathbb{H}$.
        \item The restriction of $\tilde{\pi}_n$ to $\pi_n(\mathbb{H})$ is an isometric bijection onto $\X^n$.
        \item The projections are compatible with refining the partition:
        \[\tilde{\pi}_n\circ\pi_{n+1}=\tilde{\pi}_n \qquad\text{and}\qquad \pi_n\circ\pi_{n+1}=\pi_n.\]
        \item For any path $e\in\mathbb{H}$, $\|\pi_n(e)-e\|_{\mathbb{H}}\to 0$ as $n\to\infty$, and the convergence is uniform over the achievable path set $\mathbb{K}\subseteq\mathbb{H}$, that is,
        \[\kappa_n:=\sup_{e\in\mathbb{K}}\|\pi_n(e)-e\|_{\mathbb{H}}\rightarrow 0.\]
    \end{enumerate}
\end{proposition}

\begin{proof}
    See Appendix~\ref{app:proof_properties_projections}.
\end{proof}

\subsection{Cost Functions and Correspondences}\label{subsec:costs-correspondences}

The correspondences $G,G_n$ collect the latent types consistent with a given path, and are behind the co-mechanism of Proposition~\ref{prop:co-mechanism}. The cost functions $\Psi,\Psi_n$, and their moment-penalized counterpart $\Psi_{\delta,n}$ of Section~\ref{sec:moment_penalization}, are the objectives of the transport problem \eqref{eq:OT_value} and its discretized and regularized counterparts. This section establishes the regularity each of these objects needs to play its role.

Let $G,G_n:\mathbb{H}\rightrightarrows \W$ be the correspondences collecting the latent types mapped to a given path, exactly or after discretization:
\[G(e)=\{\w\in\W: g(\w)=e\}\quad\text{ and } \quad G_n(e)=\{\w\in\W: \pi_n(g(\w))=\pi_n(e)\},\qquad e\in\mathbb{H}.\]
The fiber $G(e)$ is exactly the set over which the cost $\Psi(e)$ of \eqref{eq:Psi} minimizes $\Lambda$; $G_n(e)$ plays the same role for its discretized counterpart $\Psi_n$.

\begin{lemma}\label{lemma:correspondences-objectives}
    Under Assumption~\ref{ass:PI}, the correspondences $G,G_n:\mathbb{H}\rightrightarrows \W$ are compact-valued and upper hemicontinuous; the cost functions $\Psi,\Psi_n:\mathbb{H}\rightarrow \R$ are bounded and lower semicontinuous; $\Psi_n$ monotonically converges to $\Psi$ for paths $e\in\mathbb{K}$, i.e. $\Psi_n(e)\uparrow \Psi(e)$.
\end{lemma}

\begin{proof}
    See Appendix~\ref{app:proof_correspondences_objectives}.
\end{proof}

The moment-penalized cost $\Psi_{\delta,n}$ of Section~\ref{sec:moment_penalization} is likewise attained on a correspondence, selecting the latent types achieving its minimum:
\[G^\ast_{\delta,n}(e)=\{\w\in\W: \Psi_{\delta,n}(e)=\Lambda(\w)+\delta^{-1}\|\pi_n(g(\w)-e)\|_{\mathbb{H}}\},\qquad e\in\mathbb{H}.\]

\begin{lemma}\label{lemma:regularized-cost}
    Under Assumption~\ref{ass:PI}, the cost function $\Psi_{\delta,n}:\mathbb{H}\rightarrow\R$ is bounded and $\delta^{-1}$-Lipschitz continuous, and the correspondence $G^\ast_{\delta,n}$ is non-empty, compact-valued, and upper hemicontinuous. Moreover, for any path with $\pi_{n}(e)\in\pi_n(\mathbb{K})$, $\Psi_{\delta,n}$ increases monotonically to $\Psi_n$ as $\delta\downarrow0$, and every $\w\in G^\ast_{\delta,n}(e)$ satisfies
    \[\|\pi_n(g(\w)-e)\|_{\mathbb{H}}\leq \delta.\]
    Finally, for every path $e\in\mathbb{H}$, the perturbed cost satisfies
    \[\Psi_{\delta,n}(e)\geq \frac{1}{\delta}\inf_{\w\in\W}\|\pi_n(g(\w)-e)\|_{\mathbb{H}}.\]
\end{lemma}
\begin{proof}
    See Appendix~\ref{app:proof_regularized_cost}.
\end{proof}

An immediate consequence of Lemma~\ref{lemma:regularized-cost} is that the regularized cost $\Psi_{\delta,n}$ is well-behaved for paths whose projection is achievable, while it grows without bound for those whose distance to the projected achievable set is positive.

Once paths are identified with their images in $\X^n$ via the isometry of Proposition~\ref{prop:properties-projections}, $\Psi_{\delta,n}$ descends to a cost function $c_{\delta,n}:\X^n\rightarrow \R$ on the finite-dimensional space. We show it inherits well-definedness, boundedness, and $\delta^{-1}$-Lipschitz continuity from $\Psi_{\delta,n}$.

\begin{lemma}\label{lemma:c-is-cont}
Under Assumption~\ref{ass:PI}, the function $c_{\delta,n}:\X^n\rightarrow \R$ is well-defined, bounded by $1+2\delta^{-1}\mathrm{diam}(\X)$, and $\delta^{-1}$-Lipschitz continuous.
\end{lemma}
\begin{proof}
    See Appendix~\ref{app:proof_c_is_cont}.
\end{proof}

We close this subsection by providing a variant of the cost function $\Psi_{\delta,n}$ (and $c_{\delta,n}$) that is less computationally demanding in its construction. This is because $\Psi_{\delta,n}$ involves taking a minimum over the entire latent space $\W$, which is often high-dimensional and numerically expensive. The variant replaces this minimization with a log-sum-exp operation, an entropic smoothing of the minimum operator. It is this regularized version that is used in practice by the numerical implementation of the generalized Sinkhorn algorithm.

For this, let $d$ denote the metric of $\W$ and let $Q_0\in\P(\W)$ be a reference distribution with full support on $\W$. For each fixed parameter $\zeta>0$, we define the entropy-regularized cost function $\Psi_{\delta,n,\zeta}:\mathbb{H}\rightarrow\R$ as
\[\Psi_{\delta,n,\zeta}(e) = -\zeta\log\left( E_{Q_0}\left[ \exp(-\zeta^{-1} (\Lambda(\w)+\delta^{-1} \|\pi_n(g(\w)-e)\|_{\mathbb{H}}))\right]\right),\]
and define the corresponding cost function $c_{\delta,n,\zeta}:\X^n\rightarrow\R$ on the finite-dimensional space analogously.

A classical argument in the spirit of the Laplace principle in large deviation theory states that if the reference distribution $Q_0$ puts enough mass on small balls, then the regularized cost function $\Psi_{\delta,n,\zeta}$ converges to the original cost function $\Psi_{\delta,n}$ as $\zeta \to 0$. This ensures that the regularized version is a good approximation of the original one. Indeed, for each $\w^\ast\in\W$ and $t>0$, let the small-ball rate function be $s(t,\w^\ast)=-\log Q_0(d(\w,\w^\ast)<t)$, and let its Legendre-type transform be
\[s^\ast(\tau,\w^\ast) = \inf_{t>0}\left\{s(t,\w^\ast)+\tau t\right\}, \text{ for all }\tau>0\text{ and }\w^\ast\in\W.\]
Then, we have:

\begin{lemma}\label{lemma:log-sum-exp-cost}
Suppose the reference distribution $Q_0$ has a small-ball rate function $s$ that is continuous in $(0,+\infty)\times \W$ and satisfies
    \[\sup_{\w^\ast\in\W}\frac{s^\ast(\tau,\w^\ast)}{\tau} \to 0 \quad \text{as } \tau \to +\infty.\]
If the effect function $\Lambda$ and the mechanism function $g$ are $M$-Lipschitz continuous for some $M>0$, then for any fixed $\delta>0$, $n\in\N$, and $\zeta>0$, we have:
\[\|\Psi_{\delta,n,\zeta}-\Psi_{\delta,n}\|_\infty\leq \sup_{\w^\ast\in\W} \left(\zeta s^\ast\left(\frac{M(1+\delta^{-1})}{\zeta},\w^\ast\right)\right)\to 0 \quad \text{as } \zeta \to 0.\]
\end{lemma}

\begin{proof}
    See Appendix~\ref{app:proof_log_sum_exp_cost}.
\end{proof}

Under Assumption~\ref{ass:PI}, the latent space $\W$ is compact. Since $Q_0$ has full support and $s$ is continuous, $\sup_{\w^\ast\in\W}s(t,\w^\ast)<\infty$ for every $t>0$. Then, for every $t>0$,
\[\sup_{\w^\ast\in\W}\frac{s^\ast(\tau,\w^\ast)}{\tau}\leq t+\frac{\sup_{\w^\ast\in\W}s(t,\w^\ast)}{\tau}\longrightarrow t\quad\text{as }\tau\to+\infty,\]
and letting $t\downarrow0$ shows that the hypothesis of Lemma~\ref{lemma:log-sum-exp-cost} is satisfied automatically in our setting. We nevertheless state it explicitly because it makes transparent the rate at which $\Psi_{\delta,n,\zeta}$ approaches $\Psi_{\delta,n}$. For instance, if $s(t,\w^\ast)\leq -d_0\log t+C$ uniformly in $\w^\ast$, then $s^\ast(\tau,\w^\ast)=O(d_0\log\tau)$, and the bound of Lemma~\ref{lemma:log-sum-exp-cost} is of order $-\zeta\log\zeta$.

\subsection{Marginal Operators and Constraint Sets}\label{subsec:marginal-operators}

This section defines the marginal operators $T^\ast$ and $T$ used to express the data-matching constraint $T^\ast\mu=P$ of Section~\ref{sec:ot_connection} compactly, together with their discretized counterparts $T_n^\ast$, $\mathrm T_n^\ast$, $\mathrm T_{k,n}^\ast$ that appear once the transport problem is discretized (Section~\ref{sec:discretization}) or estimated (Section~\ref{sec:estimation}). We first collect their common properties -- linearity, continuity, positivity, and mapping the constant function $1$ to $1$ -- and show these suffice for the adjoint operators to map probability measures to probability measures. Equipped with this, we show that the resulting sets of feasible measures, $\mathcal{Q}(\mathbf{P},\lambda)$, $\mathcal{M}(\mathbf{P},\lambda)$, and their discretized counterparts, are convex and compact, which is what guarantees the transport problems attain their infimum. We close by showing that the discretized constraint sets converge to the original feasible set.

Let us start with the operator $T$, whose adjoint $T^\ast$ is used to define the constraint set $\mathcal{Q}(\mathbf{P},\lambda)$. As described in the main text, $T$ maps functions on the observable space $\Z\times \X$ into functions on the path space $\mathbb{H}$, and is formally defined as the map $T:C_b(\Z\times \X)\rightarrow C_b(\mathbb{H})$ such that, for any $\phi\in C_b(\Z\times \X)$ and $e\in\mathbb{H}$:
\[T\phi(e)=\int_{\Z} \phi(z,e_z) \, d\lambda(z).\]
Its discretized version $T_n$ is defined analogously, precomposing with the projection $\pi_n$: the map $T_n:C_b(\Z\times \X)\rightarrow C_b(\mathbb{H})$ acts as
\[T_n\phi(e):=T\phi(\pi_n(e))=\sum_{k\in[n]}\int_{A_{k,n}}\phi(z,\pi_{k,n}(e))d\lambda(z).\]
Once restricted to $\X^n$, the data-matching constraint is expressed instead through the adjoint of $\mathrm{T}_n$, the operator mapping an $n$-tuple of functions on $\X$ into a function on $\X^n$: $\mathrm{T}_n:C_b(\X)^n\rightarrow C_b(\X^n)$ is defined by
\[\mathrm{T}_n(\phi_1,\ldots,\phi_n)(x_1,\ldots,x_n) = \sum_{k\in[n]} \phi_k(x_k)\lambda_{k,n}.\]
We also define, for later use in the dual iteration of Section~\ref{sec:entropic_perturbation}, the operator $\mathrm{T}_{k,n}:C_b(\X)^{n-1}\rightarrow C_b(\X^{n-1})$ obtained by dropping the $k$-th coordinate,
\[\mathrm{T}_{k,n}(\phi_1,\ldots,\phi_{k-1},\phi_{k+1},\ldots,\phi_n) = \sum_{j\in[n]\setminus\{k\}} \phi_j(x_j)\lambda_{j,n}.\]

Our first result aggregates the properties of the marginal operators $T$, $T_n$, $\mathrm{T}_n$ needed throughout the rest of the paper. We identify the dual of $C_b(\X)^n$ with the set of $n$-tuples of regular signed measures on $\X$, so that the restriction of $\mathrm{T}_n^\ast$ to $\P(\X^n)$ is understood as the map from couplings on $\X^n$ to their $n$-tuple of marginals, $\mathrm{T}_n^\ast: \P(\X^n)\rightarrow \P(\X)^n$.
\begin{lemma}\label{lemma:marginal-operators}
    The operators $T$, $T_n$, and $\mathrm{T}_n$ are linear, continuous with operator norm at most $1$, positive, and map the constant element $1$ of their domains to the constant element $1$ of their ranges.
\end{lemma}
\begin{proof}
    See Appendix~\ref{app:proof_marginal_operators}.
\end{proof}

The above properties suffice to show that the adjoint operators $T^\ast$, $T_n^\ast$, and $\mathrm{T}_n^\ast$ are well-defined and map probability distributions to probability distributions, by the following classical result:
\begin{lemma}\label{lemma:prob-to-prob}
    Let $K$ be a compact Polish space, and $H$ be a Polish space. If $S:C_b(K)\rightarrow C_b(H)$ is a linear, continuous, positive operator that maps the constant function $1$ to $1$, then its adjoint restricted to probability measures, $S^\ast:\P(H)\rightarrow \P(K)$, is well-defined, linear on convex combinations, continuous, and takes values in the set of probability measures.
\end{lemma}
\begin{proof}
    See Appendix~\ref{app:proof_prob_to_prob}.
\end{proof}

Equipped with this result, we show that the set of feasible measures associated to each operator is convex and compact. Denote by $\mathcal{Q}(\mathbf{P},\lambda)$, $\mathcal{Q}_n(\mathbf{P},\lambda)$, $\mathcal{M}(\mathbf{P},\lambda)$, $\mathcal{M}_n(\mathbf{P},\lambda)$, $\mathcal{C}_n(\mathbf{P},\lambda)$, and $\Pi_n(\mathbf{P},\lambda)$ the sets defined by:
\begin{align*}
    \mathcal{Q}(\mathbf{P},\lambda) &= \{Q\in\P(\W): T^\ast(g_\#Q)=P\},
    \qquad \mathcal{Q}_n(\mathbf{P},\lambda) = \{Q\in\P(\W): T_n^\ast(g_\#Q)=P_n\},\\
    \mathcal{M}(\mathbf{P},\lambda) &= \{\mu\in\P(\mathbb{H}): T^\ast\mu=P, \mathrm{supp}(\mu)\subseteq \mathbb{K}\},\\
    \mathcal{M}_n(\mathbf{P},\lambda) &= \{\mu\in\P(\mathbb{H}): T_n^\ast\mu=P_n, \mathrm{supp}((\pi_n)_\#\mu)\subseteq \pi_n(\mathbb{K})\},\\
    \mathcal{C}_n(\mathbf{P},\lambda) &= \{\mu\in\P(\mathbb{H}): T_n^\ast\mu=P_n\},\\
    \Pi_n(\mathbf{P},\lambda) &= \{\nu\in\P(\X^n): \mathrm{T}_n^\ast\nu=(\mathbf{P}_{k,n})_{k\in[n]}\},
\end{align*}
where $P_n$ is the discretized joint law of Section~\ref{sec:discretization}: the sets $\mathcal{Q}_n$, $\mathcal{M}_n$, $\mathcal{C}_n$ are matched against $P_n$, not $P$, since they encode only the resolution-$n$ constraint, exactly as in the discretized transport problem of Section~\ref{sec:discretization}. The following proposition shows that these sets are convex and compact.

\begin{proposition}\label{prop:feasible-sets}
    The sets of feasible measures $\mathcal{Q}(\mathbf{P},\lambda)$, $\mathcal{Q}_n(\mathbf{P},\lambda)$, $\mathcal{M}(\mathbf{P},\lambda)$, ${\pi_n}_\# \left(\mathcal{C}_n(\mathbf{P},\lambda)\right)$, and ${\pi_n}_\# \left(\mathcal{M}_n(\mathbf{P},\lambda)\right)$ are convex and compact. Moreover, $\tilde{\pi}_{n\#} \left(\mathcal{C}_n(\mathbf{P},\lambda)\right)=\Pi_n(\mathbf{P},\lambda)$.
\end{proposition}

\begin{proof}
    See Appendix~\ref{app:proof_feasible_sets}.
\end{proof}

To conclude this section, we demonstrate that the constraint sets are nested and, with Assumption~\ref{ass:constraints}, we can show that the limiting set of discretizations is indeed the feasible constraint set of the original problem. This follows directly from the following lemma:

\begin{lemma}\label{lemma:marginals-of-T}
    The $\Z$-marginals of the images of the operators $T^\ast$ and $T^\ast_n$ are $\lambda$, that is,
    \[T^\ast\mu(A\times \X)=T^\ast_n\mu(A\times \X)=\lambda(A),\quad\text{ for any measurable }A\subseteq\Z.\]
    Moreover, for any measurable sets $A\subseteq \Z$ and $B\subseteq \X$, we have that:
    \[T^\ast_n\mu (A\times B)= \sum_{k\in[n]}{\tilde{\pi}_{k,n\#}}\mu(B)\lambda(A\cap A_{k,n}).\]
\end{lemma}
\begin{proof}
    See Appendix~\ref{app:proof_marginals_of_T}.
\end{proof}

These are summarized in the following proposition.

\begin{proposition}\label{prop:nested-constraints}
    Under Assumptions~\ref{ass:PI} and \ref{ass:constraints}, the feasible sets satisfy
    \[\mathcal{Q}(\mathbf{P},\lambda) = \bigcap_{n\in\N}\overline{\bigcup_{m\geq n}\mathcal{Q}_{m}(\mathbf{P},\lambda)} \qquad\text{and}\qquad \mathcal{M}(\mathbf{P},\lambda) = \bigcap_{n\in\N}\overline{\bigcup_{m\geq n}\mathcal{M}_{m}(\mathbf{P},\lambda)},\]
    that is, $\mathcal{Q}(\mathbf{P},\lambda)$ (resp. $\mathcal{M}(\mathbf{P},\lambda)$) is exactly the set of weak limit points of sequences of the discretized feasible sets.
\end{proposition}

\begin{proof}
    See Appendix~\ref{app:proof_nested_constraints}.
\end{proof}

\subsection{Properties of the Entropic Optimal Transport Problem}\label{subsec:eot_properties}

This section collects two properties of the multi-marginal entropic optimal transport problem needed in Section~\ref{sec:estimation}. First, we extend the quantitative stability inequality of Theorem 3.7(ii) in \citet{eckstein2023quantitative} to our setting, adapted to the weighted metric $|\cdot|_{\X^n}$. Second, restricting to marginals with finite, common support, we show that the value of the entropic optimal transport problem is differentiable with respect to the marginal probabilities and that its dual admits a unique optimal solution within a suitably normalized compact set.

Throughout, we consider the multi-marginal entropic optimal transport problem with $n$ generic marginals $\rho_k \in\P(\X)$, $k\in[n]$, and an $L$-Lipschitz continuous cost function $c:\X^n\rightarrow \R$, where $\X^n$ is endowed with the metric $|\cdot|_{\X^n}$. The parameter $\epsilon>0$ is fixed, and the reference measure is the product of the marginals, $R:=\bigotimes_{k\in[n]}\rho_k\in \P(\X^n)$. Write $\mathbf{\rho}=(\rho_k)_{k\in[n]}$ for the profile of marginals. The entropic optimal transport problem is defined via its dual formulation as
\[S(\mathbf{\rho})=\sup_{\phi\in C_b(\X)^n} E_{R}\left[\mathrm{T}_n\phi - \epsilon\left(e^{\frac{\mathrm{T}_n\phi -c}{\epsilon}}-1\right)\right].\]

Recall that the $W_1$-distance between two probability measures $\rho,\tilde{\rho}\in\P(\X)$ is defined as
\[W_1(\rho,\tilde{\rho})=\inf_{\mu\in\Pi(\rho,\tilde{\rho})} E_{\mu}[|X-Y|_\X],\]
where $\Pi(\rho,\tilde{\rho})$ is the set of all couplings between $\rho$ and $\tilde{\rho}$, and $X,Y$ are random variables for the first and second marginals of $\mu$, respectively. This distance can be equivalently defined via the Kantorovich-Rubinstein duality as
\[W_1(\rho,\tilde{\rho})=\sup_{\phi\in \mathrm{Lip}_1(\X)} E_{\rho}[\phi]-E_{\tilde{\rho}}[\phi],\]
where $\mathrm{Lip}_1(\X)$ is the set of all $1$-Lipschitz continuous functions on $\X$ (we employ this characterization below). As for measures in $\P(\X^n)$, the $W_1$-distance is defined with respect to the metric $|\cdot|_{\X^n}$, and we write $W_1(\nu,\tilde{\nu})$ for $\nu,\tilde{\nu}\in\P(\X^n)$, defined analogously as
\[W_1(\nu,\tilde{\nu})=\sup_{\phi\in \mathrm{Lip}_1(\X^n)} E_{\nu}[\phi]-E_{\tilde{\nu}}[\phi].\]
Since the norm is itself $1$-Lipschitz with respect to $|\cdot|_{\X^n}$, maximizing the resulting separable objective coordinatewise shows that this distance satisfies the inequality
\begin{equation}\label{eq:w1-marginal-lower-bound}
W_1(\nu,\tilde{\nu})\geq \sum_{k\in[n]} W_1(\nu_k,\tilde{\nu}_k)\lambda_{k,n}.
\end{equation}

As the following lemma shows, a direct extension of Lemma 3.2 in \citet{eckstein2023quantitative}, for any two profiles of marginals $\mathbf{\rho},\tilde{\mathbf{\rho}}\in\P(\X)^n$ and a coupling $\mu$ of $\mathbf{\rho}$, we can find a further coupling $\tilde{\mu}$ of $\tilde{\mathbf{\rho}}$ whose distance to $\mu$ is controlled by the averaged distances between the marginals, and whose relative entropy with respect to the product of its own marginals is no larger than that of the original coupling.

\begin{lemma}\label{lemma:coupling-approximation}
    Let $\mathbf{\rho},\tilde{\mathbf{\rho}}\in\P(\X)^n$ be two profiles of marginals, and let $\mu\in \Pi_n(\mathbf{\rho})$ be a coupling of $\mathbf{\rho}$. Then there exists a coupling $\tilde{\mu}\in \Pi_n(\tilde{\mathbf{\rho}})$ such that
    \[W_1(\mu,\tilde{\mu})= \sum_{k\in[n]} W_1(\rho_k,\tilde{\rho}_k)\lambda_{k,n},\]
    and
    \[H\Big(\tilde{\mu}\,\Big|\bigotimes_{k\in[n]} \tilde{\rho}_k\Big)\leq H\Big(\mu\,\Big|\bigotimes_{k\in[n]} \rho_k\Big).\]
\end{lemma}

\begin{proof}
    See Appendix~\ref{app:proof_coupling_approximation}.
\end{proof}

The above can be used to obtain a quantitative stability estimate of the entropic optimal transport problem with respect to the marginal profile. This is summarized in the following result:

\begin{lemma}\label{lemma:stability-eot}
    Let $\mathbf{\rho},\tilde{\mathbf{\rho}}\in\P(\X)^n$ be two profiles of marginals, and let $c:\X^n\rightarrow \R$ be an $L$-Lipschitz continuous cost function. Then the value of the entropic optimal transport problem satisfies
    \[|S(\mathbf{\rho})-S(\tilde{\mathbf{\rho}})|\leq L\sum_{k\in[n]} W_1(\rho_k,\tilde{\rho}_k)\lambda_{k,n}.\]
\end{lemma}

\begin{proof}
    See Appendix~\ref{app:proof_stability_eot}.
\end{proof}

The above result establishes that, if the marginals are close in the $L^1$ sense -- as measured by the $W_1$-distance -- and the cost functions are Lipschitz, then the value of the entropic optimal transport problem is also close. This is used in Section~\ref{sec:estimation} to show that the value of the entropic optimal transport problem with estimated marginals converges to the value of the entropic optimal transport problem with true marginals.

Next, we concentrate on the case when all marginals $\rho_k$ have finite support in a common set $\X_m=\{x_1,\ldots,x_m\}\subseteq \X$ of $m$ points. In this case, the value of the entropic optimal transport problem is differentiable with respect to the marginal probabilities, and its dual admits a unique optimal solution in a suitably normalized compact set. To state this precisely, let $\Phi:\R^{n\times m}\times (\Delta_m)^n\times \Delta_n\to \R$ be the function
\[\Phi(\phi,p,\lambda)=\sum_{\ell\in [m]^n}\left( \lambda^T\phi_{\cdot,\ell} - \epsilon\left(C(\ell)\,e^{\lambda^T\phi_{\cdot,\ell}/\epsilon}-1\right)\right)\prod_{k\in[n]}p_{k,\ell_k},\]
where $C(\ell)=\exp\!\left(-\frac{c(x_{\ell_1},\ldots,x_{\ell_n})}{\epsilon}\right)$ is the exponentiated cost of the path passing through $x_{\ell}$, and $\phi_{\cdot,\ell}=(\phi_{1,\ell_1},\ldots,\phi_{n,\ell_n})^T$ is the column vector of the potential $\phi_k$ evaluated at $x_{\ell_k}$. The function $\Phi$ is the finite-dimensional version of the dual objective of the entropic optimal transport problem, where the marginals are represented by their probability vectors $p_k=(p_{k,1},\ldots,p_{k,m})^T\in \Delta_m$, and $\lambda\in \Delta_n$ collects the weights $\lambda_{k,n}$ associated with each marginal.

When the profile of marginals $\mathbf{\rho}$ has finite support in $\X_m$, we can identify the entropic optimal transport $S(\mathbf{\rho})$ with the maximization of $\Phi$ holding $(p,\lambda)$ fixed at: $p_{k,\ell_k}=\rho_k(\{x_{\ell_k}\})$ and $\lambda_k=\lambda_{k,n}$, for $k\in[n]$ and $\ell\in [m]^n$. If we denote by $\mathcal{K}_p\subseteq \R^{n\times m}$ the set of matrices satisfying the normalizations:
\[p_k^T\phi_{k,\cdot}=0 \text{ for all }k\in[n-1]\text{ and } \phi_{k,\ell}=0\text{ if }p_{k,\ell}=0,\]
with supremum norm bounded by $3(\underline{\lambda}_n^{-1}-1)\|c\|_\infty$, where $\underline{\lambda}_n=\min_{k\in[n]}\lambda_{k,n}$, we can show that there exists a unique solution $\phi\in \mathcal{K}_p$. The following lemma shows that the value of the entropic optimal transport problem is  then differentiable with respect to the probability of each point assigned by the marginals, that there exists a unique optimal solution to the dual problem in a special compact set, and that the set of solutions is stable to perturbations:

\begin{lemma}\label{lemma:finite-support-eot-properties}
    Let $\mathbf{\rho}\in\P(\X)^n$ be a profile of marginals with finite support in $\X_m\subseteq \X$. Then, the value of the entropic optimal transport problem $S(\mathbf{\rho})$ equals:
    \[S(\mathbf{\rho})=\max_{\phi\in\R^{n\times m}}\Phi(\phi,p,\lambda):= V(p,\lambda),\]
    where $p_{k,\ell_k}=\rho_k(\{x_{\ell_k}\})$ and $\lambda_k=\lambda_{k,n}$ for $k\in[n]$ and $\ell\in [m]^n$. There exists a unique optimal solution $\phi^\ast\in \R^{n\times m}$ to the dual problem in the compact set $\mathcal{K}_p$; the function $V$ is differentiable with respect to the probability of each point assigned by the marginals, $(p,\lambda)$, if $\lambda_{k,n}>0$ for all $k\in[n]$, and satisfies:
    \begin{align*}
        D_{p,\lambda}V(p,\lambda)(q,\mu)&=D_{p,\lambda}\Phi (\phi^\ast,p,\lambda)(q,\mu)= \nabla_p V(p,\lambda)^T q\\
        &=\sum_{k\in[n]}\sum_{\ell_k\in[m]}q_{k,\ell_k}\sum_{\ell_{-k}\in[m]^{n-1}}\left(\prod_{j\in[n]\setminus\{k\}}p_{j,\ell_j}\right)\left(\lambda^T\phi^\ast_{\cdot,\ell}\right),
    \end{align*}
    for any $\mu\in \R^n$ and $q\in \R^{n\times m}$ such that $\sum_{k\in[n]}\mu_k=\sum_{\ell\in[m]}q_{k,\ell}=0$ for all $k\in[n]$, and $q_{k,\ell}=0$ if $p_{k,\ell}=0$.
    Finally, for any sequence of marginals $(p^N,\lambda^N)$ converging to $(p,\lambda)$, and $\phi^N\in \mathcal{K}_{p^N}$ almost optimal, i.e. $\Phi(\phi^N,p^N,\lambda^N)+ o(1)= V(p^N,\lambda^N)$, we have that $\phi^N\to \phi^\ast\in \mathcal{K}_p$.
\end{lemma}

\begin{proof}
    See Appendix~\ref{app:proof_finite_support_eot_properties}.
\end{proof}

\subsection{Sample Properties of the Discretized Marginals}\label{subsec:sample_properties}

This section collects the sampling properties of the empirical estimates of the discretized marginals $(P_{k,n,m})_{k\in[n]}$ of Section~\ref{sec:x_discretization} and the cell weights $(\lambda_{k,n})_{k\in[n]}$, used in the estimation results of Section~\ref{sec:estimation}.

We start by introducing the i.i.d.\ sequence of indicator matrices $(Y_i)_{i\in[N]}$ with values in $\{0,1\}^{n\times m}$, $Y_{i,(k,\ell)}=\mathbf{1}_{[Z_i\in A_{k,n},\,X_i\in B_{\ell,m}]}$. For these matrices, exactly one entry of $Y_i$ equals $1$, marking the joint cell $A_{k,n}\times B_{\ell,m}$ visited by $(Z_i,X_i)$. Their sample average is the random matrix $\bar{Y}^N=N^{-1}\sum_{i=1}^N Y_i$, whose expectation is the matrix of joint cell probabilities, $(P(A_{k,n}\times B_{\ell,m}))_{(k,\ell)\in[n]\times [m]}$.

To estimate each marginal distribution $P_{k,n,m}(x_{\ell})$ and the vector of probability weights $\lambda_{k,n}$, let us introduce the following objects. The number of realizations of the exogenous value in each cell $A_{k,n}$, $N_k$, is the sum of the $k$-th row of the matrix $\bar{Y}^N$, that is,
\[N_k=\sum_{\ell\in[m]}\sum_{i\in [N]}Y_{i,(k,\ell)}.\]
Conditioned on the event $N_k>0$ for all $k\in[n]$, set the estimates $(\hat{p}^N,\hat{\lambda}^N)\in (\Delta_m)^n\times \Delta_n^\circ$:
\[\hat{p}^N_{k,\ell}=\frac{\bar{Y}^N_{k,\ell}}{\sum_{\ell'\in[m]}\bar{Y}^N_{k,\ell'}},\quad\text{ and }\quad \hat{\lambda}^N_{k}=\sum_{\ell\in[m]}\bar{Y}^N_{k,\ell}.\]
These estimates are well-defined on this event, and may be taken arbitrarily otherwise. Since $P[N_k=0\text{ for some }k]\leq ne^{-N\underline{\lambda}_n}$, this arbitrary choice does not affect the asymptotic results below. The random matrix $\hat{p}^N$ estimates the matrix $p=(p_{k,\ell})_{(k,\ell)\in[n]\times [m]}$ of probabilities $p_{k,\ell}=P_{k,n,m}(x_\ell)$, and $\hat{\lambda}^N$ estimates the vector $\lambda=(\lambda_{k,n})_{k\in[n]}$ of cell weights. As the following lemma shows, these estimates are almost surely consistent, asymptotically normal at rate $\sqrt{N}$, asymptotically independent of one another, and have small $L^1$-error.

\begin{lemma}\label{lemma:sample-properties}
    Let $(\hat{p}^N,\hat{\lambda}^N)$ be the estimates of the marginal probabilities and weights defined above. Then:
    \begin{enumerate}
        \item[(i)] The estimates are consistent: $\hat{p}^N\to p$ and $\hat{\lambda}^N\to \lambda$ almost surely as $N\to\infty$.
        \item[(ii)] The estimates are jointly asymptotically normal at rate $\sqrt{N}$,
        \[\sqrt{N}\left(\hat{p}^N-p,\hat{\lambda}^N-\lambda\right)\Rightarrow \mathcal{N}(0,\Sigma),\]
        where $\Sigma$ is a block-diagonal covariance matrix with blocks $\Sigma_k=\frac{1}{\lambda_{k,n}}\left(\mathrm{diag}(p_{k,\cdot})-p_{k,\cdot}p_{k,\cdot}^T\right)$ for $k\in[n]$, and $\Sigma_\lambda=\mathrm{diag}(\lambda)-\lambda\lambda^T$. In particular, $\hat{p}_1^N,\ldots,\hat{p}_n^N$, and $\hat{\lambda}^N$, are pairwise asymptotically independent.
        \item[(iii)]  If $N\underline{\lambda}_n=N\min_{k\in[n]}\lambda_{k,n}>1$, then the estimates have small $L^1$-error,
        \[E[|\hat{p}_{k,\ell}^N-p_{k,\ell}|]\leq \frac{2}{\sqrt{N}}\left(\frac{1+\sqrt{p_{k,\ell}(1-p_{k,\ell})}}{\sqrt{\underline{\lambda}_n}}\right).\]
    \end{enumerate}
\end{lemma}

\begin{proof}
    See Appendix~\ref{app:proof_sample_properties}.
\end{proof}

\section{Proofs}\label{app:proofs_all}

 This subappendix will collect the proofs of the results stated in the main text, as well as the proofs of the auxiliary results in Appendix~\ref{app:aux}.

 \subsection{Proofs of Section~\ref{sec:ot_connection}}

\begin{proof}[Proof of Proposition~\ref{prop:co-mechanism}]\label{app:proof_co_mechanism}
    As shown in Lemma~\ref{lemma:correspondences-objectives}, the correspondence $G:\mathbb{K}\rightrightarrows \Omega$ defined by $G(e)=\{\w\in\Omega: g(\w)=e\}$ is non-empty, compact-valued, and upper hemicontinuous -- in particular, weakly measurable. Since $\Lambda$ is continuous (Assumption~\ref{ass:PI}), the function $(e,\w)\mapsto\Lambda(\w)$ is a Carathéodory function on $\mathrm{Graph}(G)$: constant, hence measurable, in $e$, and continuous in $\w$. By the Measurable Maximum Theorem (Theorem 18.19 in \cite{aliprantis2006infinite}), the correspondence
    \[G^\ast(e)=\{\w\in G(e): \Lambda(\w)=\Psi(e)\}=\operatorname*{arg\,min}_{\w\in G(e)}\Lambda(\w)\]
    is non-empty, compact-valued, weakly measurable, and admits a measurable selection $i:\mathbb{K}\to \Omega$. Such a function satisfies the desired properties, since $i(e)\in G^\ast(e)$ for every $e\in\mathbb{K}$ implies $i(e)\in G(e)$, so that $g(i(e))=e$, and optimality of elements of $G^\ast(e)$ implies $\Lambda(i(e))=\Psi(e)$.

    To show that the pushforward map $i_\#$ maps feasible path distributions into feasible latent laws, let $\mu\in\mathcal{M}(\mathbf{P},\lambda)$ be a feasible path distribution. Because $i$ is measurable, the pushforward $i_\#\mu$ is a well-defined Borel probability measure on $\Omega$. Moreover, for any test function $\phi\in C_b(\Omega)$,
    \[E_{g_\#\left(i_\#\mu\right)}[\phi] = E_{\mu}[\phi(g(i)) 1_{\mathbb{K}}]=E_{\mu}[\phi],\]
    since $i$ is a right-inverse and $\mu$ is supported in $\mathbb{K}$. Therefore, $g_\#(i_\#\mu)\in\P(\mathbb{K})$ can be identified with the restriction of $\mu$ to its support $\mathbb{K}$, and their unique extension to $\mathbb{H}$ must coincide. In particular, $g_\#(i_\#\mu)=\mu$, and we must have $T^\ast (g_\#(i_\#\mu))=T^\ast\mu=P$, so that $i_\#\mu\in\mathcal{Q}(\mathbf{P},\lambda)$. Moreover, the fact that $g_\#(i_\#\mu)=\mu$ and $i_\#(\mathcal{M}(\mathbf{P},\lambda))\subseteq \mathcal{Q}(\mathbf{P},\lambda)$ implies that $g_\#$ is surjective onto $\mathcal{M}(\mathbf{P},\lambda)$, concluding the proof.
\end{proof}

\begin{proof}[Proof of Theorem~\ref{thm:strong-duality}]\label{app:proof_strong_duality}
    First, note that both problems are feasible and have well-defined finite values, since the non-emptiness of $\mathcal{Q}(\mathbf{P},\lambda)$ implies the non-emptiness of $\mathcal{M}(\mathbf{P},\lambda)$, and the cost functions are bounded. Moreover, because both cost functions are also lower semicontinuous and the feasible sets are compact -- see Lemma~\ref{lemma:correspondences-objectives} -- the infima are attained, so that optimal solutions exist for both problems.

    Second, Proposition~\ref{prop:co-mechanism} shows that:
    \[\Lambda(i(e))=\Psi(e),\quad \forall e\in\mathbb{K},\quad \text{ and } \quad i_\#(\mathcal{M}(\mathbf{P},\lambda))\subseteq \mathcal{Q}(\mathbf{P},\lambda).\]
    Hence, for any optimal path distribution $\mu^\ast\in\mathcal{M}(\mathbf{P},\lambda)$, the pushforward $i_\#\mu^\ast$ is a feasible latent law in $\mathcal{Q}(\mathbf{P},\lambda)$, and the expected costs coincide:
    \[u(\mathbf{P},\lambda) =E_{\mu^\ast}[\Psi]= E_{\mu^\ast}[\Lambda\circ i] = E_{i_\#\mu^\ast}[\Lambda]  \geq v(\mathbf{P},\lambda).\]

    As for the converse inequality, observe that by definition of $\Psi$ and Proposition~\ref{prop:co-mechanism} again:
    \[\Lambda(\w)\geq \Psi(g(\w)),\quad \forall \w\in\Omega,\quad\text{ and }\quad g_\#(\mathcal{Q}(\mathbf{P},\lambda))=\mathcal{M}(\mathbf{P},\lambda).\]
    Thus, for any optimal latent law $Q^\ast\in\mathcal{Q}(\mathbf{P},\lambda)$, the pushforward $g_\#Q^\ast$ is a feasible path distribution in $\mathcal{M}(\mathbf{P},\lambda)$, and the expected costs satisfy:
    \[v(\mathbf{P},\lambda)=E_{Q^\ast}[\Lambda]\geq E_{Q^\ast}[\Psi\circ g] = E_{g_\#Q^\ast}[\Psi]\geq u(\mathbf{P},\lambda).\]

    From these inequalities, we conclude that the values of both programs must coincide. Moreover, since the pushforward maps $g_\#$ and $i_\#$ preserve feasibility, and the expected costs coincide, they must also map optimal solutions into optimal solutions,
    \[g_\#(\mathcal{Q}^\ast(\mathbf{P},\lambda))\subseteq \mathcal{M}^\ast(\mathbf{P},\lambda)\quad \text{ and }\quad i_\#(\mathcal{M}^\ast(\mathbf{P},\lambda))\subseteq\mathcal{Q}^\ast(\mathbf{P},\lambda).\]
    To show that these inclusions are in fact equalities, note first that $g_\#i_\#\mu^\ast=\mu^\ast$ for any $\mu^\ast\in\mathcal{M}^\ast(\mathbf{P},\lambda)$. Now, if $Q^\ast\in\mathcal{Q}^\ast(\mathbf{P},\lambda)$, then $g_\#Q^\ast$ must also be optimal. Moreover, by definition of the co-mechanism and the cost function $\Psi$, we have:
    \[\Lambda\circ i (g(\w))=\Psi(g(\w))\leq \Lambda(\w),\quad \forall \w\in\Omega.\]
    Note that $i_\#g_\#Q^\ast$ is also a feasible latent law. Therefore, the expected costs of $Q^\ast$ and $i_\#g_\#Q^\ast$ must coincide: both are feasible, $Q^\ast$ is optimal, and
    \[E_{i_\#g_\#Q^\ast}[\Lambda]\leq E_{Q^\ast}[\Lambda].\]
    This equality holds if and only if
    \[\Lambda\circ i\circ g = \Lambda, \quad Q^\ast\text{-a.s.}\]
    Therefore, $Q^\ast(i\circ g = \mathrm{id}_\Omega)=1$, and we conclude that $i_\#g_\#Q^\ast=Q^\ast$. This shows that the pushforward maps $g_\#$ and $i_\#$ are surjective onto the respective sets of optimal solutions, and concludes the proof.
\end{proof}

\subsection{Proofs of Section~\ref{sec:discretization}}

\begin{proof}[Proof of Proposition~\ref{prop:discretized-ot-pi}]\label{app:proof_discretized_ot_pi}
     Assumptions~\ref{ass:PI} and \ref{ass:constraints} guarantee that both constraint sets are non-empty. Lemma~\ref{lemma:correspondences-objectives} ensures the objective function is always lower semicontinuous and bounded. Hence, the values are well-defined and finite. Note also that the discretized constrained optimal transport depends only on the pushforwards of feasible path distribution by the projection $\pi_n$. Indeed, the cost function is invariant to $\pi_n$, $\Psi_n\circ \pi_n=\Psi_n$, and every feasible distribution remains feasible when pushed forward by $\pi_n$. Thus,
     \[u_n(	\mathbf{P},\lambda)=\inf_{\mu\in\pi_{n\#}\mathcal{M}_n(\mathbf{P},\lambda)} E_\mu[\Psi_n].\]
    By Proposition~\ref{prop:nested-constraints}, the set $\pi_{n\#}\mathcal{M}_n(\mathbf{P},\lambda)$ is compact and convex. Since the cost function is lower semicontinuous, the minimum is attainable in $\pi_{n\#}\mathcal{M}_n(\mathbf{P},\lambda)$. By definition of this set, the infimum must also be attainable in the bigger set $\mathcal{M}_n(\mathbf{P},\lambda)$.

    To show that both values equal, we proceed as in the proof of Proposition~\ref{prop:co-mechanism}. First, observe that by Lemma~\ref{lemma:correspondences-objectives} the correspondence $G_n:\pi_n(\mathbb{K})\rightrightarrows \W$ is upper hemicontinuous, non-empty and compact valued. Therefore, by the Measurable Maximum Theorem, there exists a measurable selection $i_n:\pi_n(\mathbb{K})\rightarrow \W$ of $G_n$ so that $\Psi_n(e)=\Lambda(i_n(e))$ for every path $e\in \pi_n(\mathbb{K})$. By definition of $G_n$, such selection must be a right-inverse for the discretized mechanism $\pi_n\circ g$, as for any path $e\in\pi_n(\mathbb{K})$: $\pi_n\circ g(i_n(e))=\pi_n(e)=e$.

    Now, let $\mu\in \mathcal{M}_n(\mathbf{P},\lambda)$ be an optimal path distribution. Since the support of $\pi_{n\#}\mu$ has to lie in $\pi_n(\mathbb{K})$, the pushforward $i_{n\#}\pi_{n\#}\mu$ is a well-defined latent distribution. Moreover, the pushforward of this latent distribution by the discretized mechanism must equal $\pi_{n\#}\mu$, since $i_n$ is a right-inverse. Therefore, as $\pi_{n\#}\mu$ satisfies the discretized marginal constraint, $i_{n\#}\pi_n\mu\in \mathcal{Q}_n(\mathbf{P},\lambda)$. Also, the fact that $i_n$ selects the latent variables attaining the minimum effect, $\Psi_n=\Lambda\circ i_n$, leads to:
    \[v_{n}(\mathbf{P},\lambda)\leq E_{i_{n\#}\pi_{n\#}\mu}[\Lambda]= E_{\pi_{n\#}\mu}[\Psi_n]=u_n(	\mathbf{P},\lambda).\]
    For the converse inequality, let $Q\in\mathcal{Q}_n(\mathbf{P},\lambda)$ be an optimal latent distribution -- which exists, since this set is convex and compact (Proposition~\ref{prop:feasible-sets}). Moreover, the path distribution $\pi_{n\#}g_\#Q$ must be feasible. As $\Lambda\geq \Psi_n(\pi_n\circ g)$ for any latent variable, then:
    \[v_{n}(\mathbf{P},\lambda)= E_{Q}[\Lambda]\geq E_{\pi_{n\#} g_\#Q}[\Psi_n]\geq u_n(\mathbf{P},\lambda),\]
    from where we conclude that both programs equal in value.

    It remains to show that the values are asymptotically closer. It suffices to demonstrate this occurs for only one class of problems, as both have the same value, whether discretized or not. Let $Q_n$ be an optimal latent law for each $n$ discretized partial identification problem. Note that any weakly convergent subsequence of these optimal distributions must converge to a feasible distribution of the original problem by Proposition~\ref{prop:nested-constraints}. Moreover, since $\W$ is compact, at least one convergent subsequence must exist. Also, the sequence of associated values belongs to the compact set $[0,1]$. Thus, it must admit at least two convergent subsequences of values, one adhering to the limsup and other to the liminf of this sequence.  By the previous argument, the respective subsequence of laws attaining these bounds must also admit a further convergent subsequence of distributions, whose respective limit distributions, $\overline{Q}$ and $\underline{Q}$,  must be a feasible in $\mathcal{Q}(\mathbf{P},\lambda)$. Since $\Lambda$ is continuous, we must have:
\begin{align*}
E_{\overline{Q}}[\Lambda]=\limsup_{n\in \N} E_{Q_n}[\Lambda]=\limsup_{n\in \N}v_n(\mathbf{P},\lambda)\geq \liminf_{n\in \N}v_n(\mathbf{P},\lambda)=\liminf_{n\in \N} E_{Q_n}[\Lambda]=E_{\underline{Q}}[\Lambda]\geq v(\mathbf{P},\lambda).
\end{align*}
    Now, let $Q ^\ast \in \mathcal{Q}(\mathbf{P},\lambda)$ be any optimal distribution. By Assumption~\ref{ass:constraints}, there is a sequence $Q_n^\ast\in \mathcal{Q}_n(\mathbf{P},\lambda)$ weakly converging to $Q^\ast$. This implies $E_{Q^\ast_n}[\Lambda]\geq v_n(\mathbf{P},\lambda)$ for any $n\in\N$. Because $\Lambda$ is continuous, we must have:
    \[v(\mathbf{P},\lambda) = E_{Q^\ast}[\Lambda]= \limsup_{n\in \N}E_{Q^\ast_n}[\Lambda]\geq\limsup_{n\in \N}v_n(\mathbf{P},\lambda)= E_{\overline{Q}}[\Lambda].\]
    Thus, the two above inequalities lead to:
    \[v(\mathbf{P},\lambda)= \lim_{n\in \N}v_n(\mathbf{P},\lambda),\]
    and the proof is complete.
\end{proof}

\subsection{Proofs of Section~\ref{sec:moment_penalization}}

\begin{proof}[Proof of Proposition~\ref{prop:support-penalization}]\label{app:proof_support_penalization}
    We begin with the first assertion. By Proposition~\ref{prop:feasible-sets}, $\tilde{\pi}_{n\#}$ maps $\mathcal{C}_n(\mathbf{P},\lambda)$ onto $\Pi_n(\mathbf{P},\lambda)$, and the cost functions are related by $\Psi_{\delta,n}(e)=c_{\delta,n}(\tilde{\pi}_n(e))$. Thus, for any $\mu\in \mathcal{C}_n(\mathbf{P},\lambda)$,
    \[E_{\mu}[\Psi_{\delta,n}]=E_{\mu}[c_{\delta,n}\circ\tilde{\pi}_n]=E_{\tilde{\pi}_{n\#}\mu}[c_{\delta,n}].\]
    Taking the infimum on both sides and using that $\tilde{\pi}_{n\#}\mathcal{C}_n(\mathbf{P},\lambda)=\Pi_n(\mathbf{P},\lambda)$ proves the first claim.

    Next, note that the infimum defining $u_{\delta,n}(\mathbf{P},\lambda)$ is attained in $\mathcal{C}_n(\mathbf{P},\lambda)$, because both the cost $\Psi_{\delta,n}$ and the constraint set $\mathcal{C}_n(\mathbf{P},\lambda)$ depend on a path only through its projection $\pi_n$, the infimum can equivalently be taken over $\pi_{n\#}\mathcal{C}_n(\mathbf{P},\lambda)$. By Proposition~\ref{prop:feasible-sets}, this set is compact, and since the cost is Lipschitz and bounded (Lemma~\ref{lemma:regularized-cost}), the minimum is attained there -- hence, by the same invariance, also in $\mathcal{C}_n(\mathbf{P},\lambda)$.

    From Lemma~\ref{lemma:regularized-cost}, the correspondence $G^\ast_{\delta,n}$ is non-empty, compact-valued, and upper hemicontinuous, so it admits a measurable selection (Theorems 18.13 and 18.20 in \cite{aliprantis2006infinite}); denote such a selection by $i_{\delta,n}:\mathbb{H}\rightarrow \W$.

    For each $n\in\N$ and $\delta>0$, let $\mu_{\delta,n}\in \mathcal{C}_n(\mathbf{P},\lambda)$ be optimal, and denote by $Q_{\delta,n}=i_{\delta,n\#}\mu_{\delta,n}\in\P(\W)$ the latent law obtained by pushing $\mu_{\delta,n}$ forward through $i_{\delta,n}$. We claim that the Wasserstein distance $W_1$ between $\pi_{n\#}g_\#Q_{\delta,n}$ and $\pi_{n\#}\mu_{\delta,n}$ is at most $\delta$. Indeed, for any $1$-Lipschitz function $\phi:\mathbb{H}\to[-1,1]$,
    \[\big|E_{\pi_{n\#}g_\# Q_{\delta,n}}[\phi]-E_{\pi_{n\#}\mu_{\delta,n}}[\phi]\big|=\Big|E_{\mu_{\delta,n}}\big[\phi(\pi_n(g(i_{\delta,n}(e))))-\phi(\pi_n(e))\big]\Big|\leq E_{\mu_{\delta,n}}\big[\|\pi_n(g(i_{\delta,n}(e)))-\pi_n(e)\|_{\mathbb{H}}\big].\]
    Since $i_{\delta,n}$ selects a latent type satisfying $\Lambda(i_{\delta,n}(e))+\delta^{-1}\|\pi_n(g(i_{\delta,n}(e)))-\pi_n(e)\|_{\mathbb{H}}=\Psi_{\delta,n}(e)$ for every $e\in\mathbb{H}$,
    \[\delta^{-1}E_{\mu_{\delta,n}}\big[\|\pi_n(g(i_{\delta,n}(e)))-\pi_n(e)\|_{\mathbb{H}}\big] = E_{\mu_{\delta,n}}[\Psi_{\delta,n}]-E_{Q_{\delta,n}}[\Lambda]\leq u_{n}(\mathbf{P},\lambda)\leq 1,\]
    where the last two inequalities use that $\Lambda\geq0$, that $\mu_{\delta,n}$ is optimal, so $E_{\mu_{\delta,n}}[\Psi_{\delta,n}]=u_{\delta,n}(\mathbf{P},\lambda)$, and that $u_{\delta,n}(\mathbf{P},\lambda)\leq u_n(\mathbf{P},\lambda)$. Multiplying both sides by $\delta$ and substituting above gives
    \[\big|E_{\pi_{n\#}g_\# Q_{\delta,n}}[\phi]-E_{\pi_{n\#}\mu_{\delta,n}}[\phi]\big|\leq \delta,\quad\text{ for every such }\phi,\]
    so $W_1(\pi_{n\#}g_\# Q_{\delta,n}, \pi_{n\#} \mu_{\delta,n})\leq \delta$. Since $T^\ast$ has Lipschitz constant at most $1$ with respect to $W_1$, and the marginals of $\mu_{\delta,n}$ are $P_n$,
    \[W_1(T_n^\ast(g_\# Q_{\delta,n}),P_n)=W_1(T^\ast\pi_{n\#}g_\# Q_{\delta,n}, T^\ast\pi_{n\#} \mu_{\delta,n})\leq W_1(\pi_{n\#}g_\# Q_{\delta,n}, \pi_{n\#} \mu_{\delta,n})\leq \delta.\]

    We can now conclude both limits. First, since $\Psi_{\delta,n}$ increases as $\delta\downarrow 0$ (Lemma~\ref{lemma:regularized-cost}) while the constraint set is unchanged for fixed $n$, the value $u_{\delta,n}(\mathbf{P},\lambda)$ is non-decreasing as $\delta\downarrow0$ and bounded above by $u_n(\mathbf{P},\lambda)$. Moreover, since $\Psi_{\delta,n}\geq \Lambda \circ i_{\delta,n}$,
    \[0\leq E_{Q_{\delta,n}}[\Lambda]\leq E_{\mu_{\delta,n}}[\Psi_{\delta,n}] =u_{\delta,n}(\mathbf{P},\lambda).\]
    Take a sequence $\delta\downarrow0$ along which the limsup is attained; by compactness of $\P(\W)$, $(Q_{\delta,n})$ admits a further subsequence converging weakly to some $Q\in \P(\W)$, and continuity of $\Lambda$ gives
    \[E_Q[\Lambda]=\limsup_{\delta\downarrow 0}E_{Q_{\delta,n}}[\Lambda]\leq \limsup_{\delta\downarrow 0} u_{\delta,n}(\mathbf{P},\lambda)\leq u_n(\mathbf{P},\lambda).\]
    Along this same subsequence, $\pi_{n\#}g_{\#}Q_{\delta,n}$ converges weakly to $\pi_{n\#}g_\#Q$, since $\pi_n\circ g$ is continuous. Combined with $W_1(T_n^\ast(g_\#Q_{\delta,n}),P_n)\leq \delta$ and the weak continuity of $T_n^\ast$, this gives $T_n^\ast(g_\#Q)=P_n$, so $Q\in \mathcal{Q}_n(\mathbf{P},\lambda)$. Proposition~\ref{prop:discretized-ot-pi} then implies $E_Q[\Lambda]\geq v_n(\mathbf{P},\lambda)=u_n(\mathbf{P},\lambda)$, which forces equality throughout and shows $u_{\delta,n}(\mathbf{P},\lambda)\uparrow u_n(\mathbf{P},\lambda)$ as $\delta\downarrow0$.

    Second, for the joint limit, note that
    \[0\leq \liminf_{\delta\downarrow 0,\,n\uparrow +\infty} E_{Q_{\delta,n}}[\Lambda]\leq \liminf_{\delta\downarrow 0,\,n\uparrow +\infty} u_{\delta,n}(\mathbf{P},\lambda)\leq \limsup_{\delta\downarrow 0,\,n\uparrow +\infty} u_{\delta,n}(\mathbf{P},\lambda)\leq \limsup_{n\uparrow +\infty} u_{n}(\mathbf{P},\lambda) =u(\mathbf{P},\lambda),\]
    where the last two relations follow from $u_{\delta,n}\leq u_n$ and Proposition~\ref{prop:discretized-ot-pi}. Take a subsequence of indices $(\delta,n)$ along which the leftmost liminf is attained, and along it extract a further subsequence such that $Q_{\delta,n}$ converges weakly to some $Q\in\P(\W)$. By item 5 of Proposition~\ref{prop:properties-projections}, since $g_\#Q_{\delta,n}$ is supported on $\mathbb{K}$,
    \[0\leq W_1(g_\#Q,\pi_{n\#}g_{\#} Q_{\delta,n})\leq \kappa_n +W_1(g_\#Q,g_\#Q_{\delta,n}),\]
    and, since $g$ is continuous, $Q_{\delta,n}\Rightarrow Q$ along the subsequence and $\kappa_n\to0$, so $\pi_{n\#}g_\# Q_{\delta,n}\Rightarrow g_{\#}Q$ along it as well. Since $T^\ast$ has Lipschitz constant at most $1$, $P_n\Rightarrow P$, and $W_1(T_n^\ast(g_\# Q_{\delta,n}),P_n)\leq \delta$, we conclude $W_1(T^\ast (g_\# Q),P)=0$, so $Q\in \mathcal{Q}(\mathbf{P},\lambda)$. By Theorem~\ref{thm:strong-duality} and continuity of $\Lambda$,
    \[u(\mathbf{P},\lambda)=v(\mathbf{P},\lambda)\leq E_Q[\Lambda]=\liminf_{\delta\downarrow0,\,n\uparrow +\infty}E_{Q_{\delta,n}}[\Lambda] \leq \liminf_{\delta\downarrow0,\,n\uparrow +\infty} u_{\delta,n}(\mathbf{P},\lambda)\leq\limsup_{\delta\downarrow0,\,n\uparrow +\infty} u_{\delta,n}(\mathbf{P},\lambda)\leq u(\mathbf{P},\lambda),\]
    which forces equality throughout and completes the proof.
\end{proof}

\subsection{Proofs of Section~\ref{sec:entropic_perturbation}}

\begin{proof}[Proof of Proposition~\ref{prop:duality}]\label{app:proof_duality}
By Lemma~\ref{lemma:c-is-cont}, the cost function $c_{\delta,n}:\X^n\to\R$ is bounded by $1+2\delta^{-1}\mathrm{diam}(\X)$. In particular, it is a bounded Borel cost on the Polish spaces $(\X^n,|\cdot|_{\X^n})$. Theorem 4.5 in \citet{dimarino2020schrodinger}, applied to the marginals $(P_{k,n})_{k\in[n]}$ and this cost, therefore gives
\[u_{\epsilon,\delta,n}(\mathbf{P},\lambda)=\sup_{\psi\in C_b(\X)^n} E_{r_n}\!\left[\sum_{k\in [n]}\psi_k- \epsilon\left(\exp\!\left({\frac{\sum_{k\in[n]} \psi_{k}-c_{\delta,n}}{\epsilon}}\right)-1\right)\right].\]
Substituting $\psi_k=\lambda_{k,n}\phi_k$ for each $k\in[n]$, a bijection of $C_b(\X)^n$ onto itself, since $\lambda_{k,n}>0$, turns $\sum_{k\in[n]}\psi_k$ into $\mathrm{T}_n\phi$, and the supremum above into the dual displayed in the proposition.

Now fix $\delta,n$. The value $u_{\epsilon,\delta,n}(\mathbf{P},\lambda)$ is monotone non-decreasing in $\epsilon$. This happens as the relative entropy is always non-negative, turning the objective of \eqref{eq:eot_value}, for every fixed feasible $\nu$, non-decreasing in $\epsilon$. As the constraint set remains the same, its infimum over $\nu\in\Pi_n(\mathbf{P},\lambda)$ is non-decreasing. Together with the rate bound of \citet{carlier2017convergence} above, whose upper envelope $u_{\delta,n}(\mathbf{P},\lambda)+d\,\epsilon\log\epsilon+\epsilon/\delta$ converges to $u_{\delta,n}(\mathbf{P},\lambda)$ as $\epsilon\downarrow0$, this forces $u_{\epsilon,\delta,n}(\mathbf{P},\lambda)\downarrow u_{\delta,n}(\mathbf{P},\lambda)$. If instead only $n$ is held fixed, the same rate bound, the fact that $\epsilon\log\epsilon$ and $\epsilon/\delta$ vanish as $\epsilon/\delta\to0$, and the first limit of Proposition~\ref{prop:support-penalization} together give $u_{\epsilon,\delta,n}(\mathbf{P},\lambda)\rightarrow u_n(\mathbf{P},\lambda)$. Finally, letting $n\uparrow+\infty$ as well, the same argument combined with the joint-limit statement of Proposition~\ref{prop:support-penalization} gives $u_{\epsilon,\delta,n}(\mathbf{P},\lambda)\rightarrow u(\mathbf{P},\lambda)$.
\end{proof}

\begin{proof}[Proof of Proposition~\ref{prop:entropic-ot-dual}]\label{app:proof_entropic_ot_dual}
   By Theorem 4.5(ii) and Proposition 4.6 in \citet{dimarino2020schrodinger}, the dual is attainable: a maximizer $\psi=(\psi_1,\ldots,\psi_n)\in C_b(\X)^n$ exists for the dual problem in the proof of Proposition~\ref{prop:duality}, unique up to the usual additive normalization, and $\phi_k:=\psi_k/\lambda_{k,n}$ is accordingly a maximizer of our dual, satisfying the first-order condition \eqref{eq:foc}. By Proposition 4.6(iii) in \citet{dimarino2020schrodinger}, translated through the same substitution, the corresponding optimizer $\nu^\ast_{\epsilon,\delta,n}\in\Pi_n(\mathbf{P},\lambda)$ of the primal problem \eqref{eq:eot_value} exists and has log-density
   \[\log\!\left(\frac{d\nu^\ast_{\epsilon,\delta,n}}{dr_n}\right)(x_1,\ldots,x_n)=\frac{\mathrm{T}_n\phi(x_1,\ldots,x_n)-c_{\delta,n}(x_1,\ldots,x_n)}{\epsilon}.\]
   It is the unique optimizer since \eqref{eq:eot_value} is a strictly convex minimization.

   For the Lipschitz continuity of potentials, fix $k\in[n]$, points $x_k,x_k'\in\X$, and let $x_{-k}\in\X^{n-1}$ be arbitrary. Write $x=(x_k,x_{-k})$, $x'=(x_k',x_{-k})$. Since $c_{\delta,n}$ is $\delta^{-1}$-Lipschitz with respect to $|\cdot|_{\X^n}$ (Lemma~\ref{lemma:c-is-cont}) and the distance $|x-x'|_{\X^n}=\lambda_{k,n}|x_k-x_k'|_\X$, we have $|c_{\delta,n}(x)-c_{\delta,n}(x')|\leq \delta^{-1}\lambda_{k,n}|x_k-x_k'|_\X$ uniformly in $x_{-k}$. Using the first-order equation \eqref{eq:foc} and the standard fact that $y\mapsto\log E[e^{y/\epsilon}]$ is $\epsilon^{-1}$-Lipschitz with respect to the sup norm,
   \begin{align*}
|\phi_k(x_k)-\phi_k(x_k')|&=\frac{\epsilon}{\lambda_{k,n}}\left|\log E_{r_{-k,n}}\!\left[e^{\frac{\mathrm{T}_{k,n}\phi_{-k}-c_{\delta,n}(x_k,\cdot)}{\epsilon}}\right]-\log E_{r_{-k,n}}\!\left[e^{\frac{\mathrm{T}_{k,n}\phi_{-k}-c_{\delta,n}(x_k',\cdot)}{\epsilon}}\right]\right|,\\
&\leq \frac{1}{\lambda_{k,n}}\sup_{x_{-k}}|c_{\delta,n}(x)-c_{\delta,n}(x')|\leq \delta^{-1}|x_k-x_k'|_\X,
   \end{align*}
   so $\phi_k$ is $\delta^{-1}$-Lipschitz. The same computation, using that $\psi_k$ is bounded by $\|c_{\delta,n}\|_\infty\leq 1+2\delta^{-1}\mathrm{diam}(\X)$ (Theorem 4.5(ii) in \citet{dimarino2020schrodinger} and Lemma~\ref{lemma:c-is-cont}), gives $\|\phi_k\|_\infty\leq(1+2\delta^{-1}\mathrm{diam}(\X))/\lambda_{k,n}$.

   Finally, the first-order condition \eqref{eq:foc} implies $E_{r_n}\!\left[\exp\left(\frac{\mathrm{T}_n\phi-c_{\delta,n}}{\epsilon}\right)\right]=1$. Conditioning on $x_k$, \eqref{eq:foc} gives $E_{r_{-k,n}}\!\left[\exp\left(\frac{\mathrm{T}_n\phi-c_{\delta,n}}{\epsilon}\right)\,\middle|\,x_k\right]=1$ for $P_{k,n}$-a.e.\ $x_k$, and taking $E_{P_{k,n}}$ of both sides gives the claim. Hence, by Proposition~\ref{prop:duality},
   \begin{align*}
    u_{\epsilon,\delta,n}(\mathbf{P},\lambda)&=E_{r_n}\!\left[\mathrm{T}_n\phi-\epsilon\left(e^{\frac{\mathrm{T}_n\phi-c_{\delta,n}}{\epsilon}}-1\right)\right],\\
    &= E_{r_n}[\mathrm{T}_n\phi]=\sum_{k\in[n]}\lambda_{k,n}\,E_{P_{k,n}}[\phi_k]=\sum_{k\in[n]}\int_{A_{k,n}\times\X}\phi_{k}(x)\,dP(z,x).
   \end{align*}
\end{proof}

\subsection{Proofs of Section~\ref{sec:x_discretization}}

\begin{proof}[Proof of Proposition~\ref{prop:fully-discretized-ot}]\label{app:proof_fully_discretized_ot}
The feasible set $\Pi_{n,m}(\mathbf{P},\lambda)$ is non-empty and, since $\X_m^n$ is finite, compact. Together with the boundedness and continuity of $c_{\delta,n}$ (Lemma~\ref{lemma:c-is-cont}), this places \eqref{eq:eot_value_discretized} within the scope of the same multi-marginal entropic optimal transport results of \citet{dimarino2020schrodinger} used in the proof of Proposition~\ref{prop:entropic-ot-dual}, now applied to the finitely supported marginals $(P_{k,n,m})_{k\in[n]}$: strong duality holds, the dual admits a maximizer $\phi$ pinned down on $\X_m$ by the same first-order argument, and the value and log-density formulas follow exactly as before, with $(P_{k,n},r_n)$ replaced throughout by $(P_{k,n,m},r_{n,m})$.

For the convergence rate, $c_{\delta,n}$ is $\delta^{-1}$-Lipschitz with respect to the weighted metric $|x-y|_{\X^n}=\sum_{k\in[n]}\lambda_{k,n}|x_k-y_k|_\X$ (Lemma~\ref{lemma:c-is-cont}), so Lemma~\ref{lemma:stability-eot} gives
\[|u_{\epsilon,\delta,n,m}(\mathbf{P},\lambda)- u_{\epsilon,\delta,n}(\mathbf{P},\lambda)|\leq\frac{1}{\delta} \sum_{k\in[n]}W_1(P_{k,n},P_{k,n,m}) \lambda_{k,n}.\]
Since the mass $P_{k,n}$ assigns to each cell $B_{j,m}$ moves a distance at most $\rho_m$ to reach its representative point $x_{j,m}$, the identity coupling within cells shows $W_1(P_{k,n},P_{k,n,m})\leq \rho_m$ for every $k\in[n]$. As $\sum_{k\in[n]}\lambda_{k,n}=1$, this gives
\[|u_{\epsilon,\delta,n,m}(\mathbf{P},\lambda)- u_{\epsilon,\delta,n}(\mathbf{P},\lambda)|\leq\frac{\rho_m}{\delta},\]
and, in particular, the stated convergence as $\rho_m\downarrow0$. The joint convergence statement follows by combining this rate with the joint-limit statement of Proposition~\ref{prop:duality}: as $\rho_m/\delta\to0$, $\epsilon/\delta\to0$, and $n\to\infty$, both $|u_{\epsilon,\delta,n,m}-u_{\epsilon,\delta,n}|\to0$ and $u_{\epsilon,\delta,n}\to v(\mathbf{P},\lambda)$, so that $u_{\epsilon,\delta,n,m}\to v(\mathbf{P},\lambda)$ as well.
\end{proof}

\begin{proof}[Proof of Proposition~\ref{prop:sinkhorn-convergence}]\label{app:proof_sinkhorn_convergence}
This is a restatement of Theorem 3.3 in \citet{carlier2022linear}, under the correspondence $\varphi_k:=\frac{\lambda_{k,n}\phi_k}{\epsilon}$ and Carlier's cost $c^{\mathrm{C}}:=\frac{c_{\delta,n}}{\epsilon}$, so that $\|c^{\mathrm C}\|_\infty=\epsilon^{-1}\|c_{\delta,n}\|_\infty$. Under this rescaling, Carlier's objective, his Equation (2.4), satisfies
\[\Phi(\phi,p,\lambda)=\epsilon\big(1-F(\varphi)\big),\]
a direct computation from the definitions of $\Phi$ and $F$. Since this map is an order-reversing affine bijection between $\phi\in\R^{n\times m}$ and $\varphi\in\R^{n\times m}$, and since the normalization defining $\mathcal{K}_p$ is invariant under the rescaling $\varphi_k=\frac{\lambda_{k,n}\phi_k}{\epsilon}$, the Sinkhorn iterates $(\phi^t)$ of Section~\ref{sec:x_discretization} correspond exactly, term by term, to the block-coordinate-descent iterates $(\varphi^t)$ of \citet[Eqs.~(2.7)--(2.14)]{carlier2022linear} for $F$, started from $\varphi^0=0$. Theorem 3.3 there gives
\[F(\varphi^t)-F(\varphi^\ast)\ \leq\ \left(1-\frac{e^{-(16n-8)\|c^{\mathrm C}\|_\infty}}{n}\right)^t\big(F(\varphi^0)-F(\varphi^\ast)\big),\]
where $\varphi^\ast$ is the unique minimizer of $F$. Substituting $\Phi=\epsilon(1-F)$ on both sides turns this into the stated bound. Finally, note that $\Phi(\phi^0,p,\lambda)=\epsilon>0$ and $\Phi(\phi^\ast,p,\lambda)\leq \|c_{\delta,n}\|_\infty\leq 1+\delta^{-1}\mathrm{diam}(\X)$, which completes the proof.
\end{proof}



\subsection{Proofs of Section~\ref{sec:consistency}}

\begin{proof}[Proof of Theorem~\ref{thm:consistency}]\label{app:proof_consistency}
    We bound the $L^1$ distance between $\hat{v}_N$ and $v(\mathbf{P},\lambda)$ in three steps, one for each source of approximation error. First, the absolute deviation of the estimator $\hat{v}_N$ from the plug-in value of the entropic optimal transport problem $u_{\epsilon,\delta,n,m}(\mathbf{P},\lambda)$, denoted by $\hat{u}_{\epsilon,\delta,n,m}$, is, by Proposition~\ref{prop:sinkhorn-convergence}, bounded by the Sinkhorn error:
    \[|\hat{v}_N-\hat{u}_{\epsilon,\delta,n,m}|\leq \exp\left(-Ie^{-\frac{16n}{\delta \epsilon}\mathrm{diam}(\X)}\right)=o(1),\]
    which, by assumption, vanishes as the sample size grows. Second, the absolute deviation of the plug-in value $\hat{u}_{\epsilon,\delta,n,m}$ from the population value $u_{\epsilon,\delta,n,m}(\mathbf{P},\lambda)$ is bounded by the $L^1$-error of the empirical marginals, an immediate consequence of the stability inequality of Lemma~\ref{lemma:stability-eot}:
    \[|\hat{u}_{\epsilon,\delta,n,m}-u_{\epsilon,\delta,n,m}(\mathbf{P},\lambda)|\leq \frac{1}{\delta}\sum_{k\in[n]} W_1(\hat{p}^N_{k,\cdot},P_{k,n,m})\lambda_{k,n}.\]
    By the dual representation of the Wasserstein distance in terms of Lipschitz functions, $W_1(\hat{p}^N_{k,\cdot},P_{k,n,m})$ is in turn bounded by
    \[W_1(\hat{p}^N_{k,\cdot},P_{k,n,m})\leq \sum_{\ell\in[m]}|\hat{p}^N_{k,\ell}-p_{k,\ell}|.\]
    Taking expectations on both sides and invoking the $L^1$-error bound of Lemma~\ref{lemma:sample-properties} for the empirical marginals gives
    \[E\left[|\hat{u}_{\epsilon,\delta,n,m}-u_{\epsilon,\delta,n,m}(\mathbf{P},\lambda)|\right]\leq \frac{1}{\delta}\sum_{k\in[n]} E\left[W_1(\hat{p}^N_{k,\cdot},P_{k,n,m})\right]\lambda_{k,n}\leq \frac{1}{\delta}\sum_{k\in[n],\ell\in[m]}\frac{4}{\sqrt{N\lambda_{k,n}}}\lambda_{k,n},\]
    which is bounded by $\frac{4nm\overline{\lambda}_n^{\frac{1}{2}}}{\delta N^{\frac{1}{2}}}$ and, by assumption, vanishes as the sample size grows. Third, and finally, the population value $u_{\epsilon,\delta,n,m}(\mathbf{P},\lambda)$ converges to the partially identified bound $v(\mathbf{P},\lambda)$ by Proposition~\ref{prop:fully-discretized-ot}. Combining the three steps gives $L^1$-consistency of $\hat{v}_N$.
\end{proof}

\subsection{Proofs of Section~\ref{sec:asymptotics}}

\begin{proof}[Proof of Theorem~\ref{thm:asymptotic_distribution}]\label{app:proof_asymptotic_distribution}
    We establish the result in three steps. First, the absolute deviation of the estimator $\hat{v}_N$ from the plug-in value of the entropic optimal transport problem $u_{\epsilon,\delta,n,m}(\mathbf{P},\lambda)$, denoted by $\hat{u}_{\epsilon,\delta,n,m}$, is, by Proposition~\ref{prop:sinkhorn-convergence}, bounded by the Sinkhorn error:
    \[|\hat{v}_N-\hat{u}_{\epsilon,\delta,n,m}|\leq \exp\left(-Ie^{-\frac{16n}{\delta \epsilon}\mathrm{diam}(\X)}\right)=o(N^{-\frac{1}{2}}),\]
    which, since $\frac{\ln N}{I}=o(1)$ by assumption, vanishes in probability faster than $N^{-\frac{1}{2}}$.

    Second, the values $\hat{u}_{\epsilon,\delta,n,m}$ and $u_{\epsilon,\delta,n,m}(\mathbf{P},\lambda)$ are, by Lemma~\ref{lemma:finite-support-eot-properties}, the value of a strongly concave and coercive objective function on a finite-dimensional vector subspace, given respectively by $V(\hat{p}^N,\hat{\lambda}^N)$ and $V(p,\lambda)$. By this same lemma, the value function $V$ is continuously differentiable at $(\hat p^N,\hat\lambda^N)$ provided $\hat\lambda^N_{k,n}>0$ for all $k\in[n]$ and, along every feasible direction $q\in\R^{n\times m}$ with $q_{k,\ell}=0$ whenever $\hat{p}^N_{k,\ell}=0$. With probability approaching one, the empirical marginal $\hat{\lambda}^N$ satisfies the first condition. Moreover, also with probability approaching one, $\hat{p}^N_{k,\ell}=0$ whenever $p_{k,\ell}=0$, and $q=\hat{p}^N-p$ is a feasible direction. As $\sqrt{N}\left((\hat{p}^N,\hat{\lambda}^N)-(p,\lambda)\right)\Rightarrow N(0,\Sigma)$ by Lemma~\ref{lemma:sample-properties}, the delta method gives
    \[\sqrt{N}\left(\hat{u}_{\epsilon,\delta,n,m}-u_{\epsilon,\delta,n,m}(\mathbf{P},\lambda)\right)\Rightarrow \mathcal{N}\left(0,\nabla V(p,\lambda)^T\Sigma \nabla V(p,\lambda)\right),\]
    and simplifying the asymptotic variance using the block structure of $\Sigma$ together with the characterization of the gradient in Lemma~\ref{lemma:finite-support-eot-properties} yields the weighted population variance $\textnormal{Var}_{(p,\lambda)}[\phi^\ast_\beta]$ defined above.

    Third, and finally, we relate the estimator's limit to that of $\phi^\ast_\beta$'s empirical counterpart. Since
    \[\hat{v}_N +o(1)= \Phi(\hat{\phi}^N_\beta,\hat{p}^N,\hat{\lambda}^N)+o(1)= V(\hat{p}^N,\hat{\lambda}^N)\]
    and $\hat{\phi}^N_\beta\in \mathcal{K}_{\hat{p}^N}$. As $(\hat{p}^N,\hat{\lambda}^N)$ converges to $(p,\lambda)$ almost surely, by the strong law of large numbers, Lemma~\ref{lemma:finite-support-eot-properties} further implies that $\hat{\phi}^N_\beta$ converges to $\phi^\ast_\beta$ almost surely. Thus, $\textnormal{Var}_{(\hat{p}^N,\hat{\lambda}^N)}[\hat{\phi}^N_\beta]$ is a consistent estimator of $\textnormal{Var}_{(p,\lambda)}[\phi^\ast_\beta]$, concluding the proof.
\end{proof}

\subsection{Proofs of Appendix~\ref{app:aux}}\label{app:proofs_aux}

\subsubsection{Proofs of Section~\ref{subsec:projections}}

\begin{proof}[Proof of Proposition~\ref{prop:properties-projections}]\label{app:proof_properties_projections}
    \begin{enumerate}
        \item Linearity is immediate from linearity of the Bochner integral defining $\tilde\pi_{k,n}$. For the norm bound, Jensen's inequality -- Proposition 1.2.11 in \cite{Hytonen2016} -- gives
        \[|\tilde\pi_{k,n}(e)|_\mathcal{H} \leq \lambda_{k,n}^{-1}\int_{A_{k,n}} |e_z|_\mathcal{H}\,d\lambda(z),\]
        so that $\|\pi_{k,n}(e)\|_{\mathbb{L}^1(\Z,\lambda;\mathcal{H})}=|\tilde\pi_{k,n}(e)|_\mathcal{H}\,\lambda_{k,n}\leq \int_{A_{k,n}}|e_z|_\mathcal{H}\,d\lambda(z)$. Summing over the disjoint cells leads to $\|\pi_n(e)\|_{\mathbb{L}^1(\Z,\lambda;\mathcal{H})}\leq \|e\|_{\mathbb{L}^1(\Z,\lambda;\mathcal{H})}$, so $\pi_n$, and by the same computation restricted to a single cell each of $\pi_{k,n}$, $\tilde\pi_{k,n}$, $\tilde\pi_n$, has operator norm at most $1$. In particular, all four are continuous.

        \item  Since $\pi_n(e)$ is constant, equal to $\tilde\pi_{k,n}(e)$, on $A_{k,n}$, its own conditional average over $A_{k,n}$ reproduces this constant:
        \[\tilde\pi_{k,n}(\pi_n(e))=\lambda_{k,n}^{-1}\int_{A_{k,n}}\tilde\pi_{k,n}(e)\,d\lambda(z)=\tilde\pi_{k,n}(e),\qquad\forall k\in[n].\]
        This gives $\tilde\pi_n\circ\pi_n=\tilde\pi_n$ and, summing the corresponding pieces over $k$, $\pi_n\circ\pi_n=\pi_n$. Convexity and compactness of the range $\pi_n(\mathbb{H})$ then follow from item 3: it is the image of the convex, compact set $\X^n$ under the inverse of an isometric bijection, hence itself convex and compact.

        \item  Every $e\in\pi_n(\mathbb{H})$ is, by construction, piecewise constant on $\mathcal{A}_n$ with values in $\X$ (using that $\X$ is closed and convex, so conditional averages of $\X$-valued paths remain in $\X$). Write $x_k=\tilde\pi_{k,n}(e)\in\X$ for its value on $A_{k,n}$. The map $x\in\X^n\mapsto \sum_{k\in[n]}x_k1_{A_{k,n}}$ is a two-sided inverse of $\tilde\pi_n$ on $\pi_n(\mathbb{H})$, so $\tilde\pi_n$ is a bijection onto $\X^n$. For $e,e'\in\pi_n(\mathbb{H})$ with values $x,x'\in\X^n$,
        \[\|e-e'\|_{\mathbb{H}}=\int_\Z |e_z-e_z'|_\mathcal{H}\,d\lambda(z)=\sum_{k\in[n]}\int_{A_{k,n}}|x_k-x_k'|_{\X}\,d\lambda(z)=\sum_{k\in[n]}|x_k-x_k'|_{\X}\lambda_{k,n}=|x-x'|_{\X^n},\]
        so $\tilde\pi_n$ is distance-preserving on $\pi_n(\mathbb{H})$, i.e.\ an isometry.

        \item Fix $k\in[n]$ and write $A_{k,n}=\bigcup_{j\in J_k}A_{j,n+1}$ as a disjoint union of cells of the finer partition, so $\lambda_{k,n}=\sum_{j\in J_k}\lambda_{j,n+1}$. Since $\pi_{n+1}(e)$ equals $\tilde\pi_{j,n+1}(e)$ on each $A_{j,n+1}$,
        \begin{align*}
            \tilde\pi_{k,n}(\pi_{n+1}(e)) &= \frac{1}{\lambda_{k,n}}\int_{A_{k,n}}\pi_{n+1}(e)(z)\,d\lambda(z) = \frac{1}{\lambda_{k,n}}\sum_{j\in J_k}\lambda_{j,n+1}\tilde\pi_{j,n+1}(e),\\
             &= \frac{1}{\lambda_{k,n}}\sum_{j\in J_k}\int_{A_{j,n+1}}e_z\,d\lambda(z) = \tilde\pi_{k,n}(e).
        \end{align*}
        As $k$ was arbitrary, $\tilde\pi_n\circ\pi_{n+1}=\tilde\pi_n$; summing the corresponding pieces over $k$ gives $\pi_n\circ\pi_{n+1}=\pi_n$.

        \item Fix $\delta>0$ and $e\in\mathbb{H}$. By density of continuous paths in $\mathbb{H}$ -- Lemma 1.2.31 in \cite{Hytonen2016} -- there is a continuous $e'\in\mathbb{H}$ with $\|e-e'\|_{\mathbb{H}}\leq\delta/3$. Because $\Z$ is compact, $e'$ is uniformly continuous, so there is $\rho>0$ such that $d_\Z(z,z')<\rho\Rightarrow|e'(z)-e'(z')|_\mathcal{H}<\delta/3$. Since $\rho_n\to 0$, choose $n_0$ with $\rho_n<\rho$ for all $n\geq n_0$; then for $n\geq n_0$, $k\in[n]$, and $z\in A_{k,n}$,
        \[|\pi_n(e')(z)-e'(z)|_\mathcal{H} = \Big|\frac{1}{\lambda_{k,n}}\int_{A_{k,n}}\big(e'(z')-e'(z)\big)\,d\lambda(z')\Big|_\mathcal{H} \leq \frac{1}{\lambda_{k,n}}\int_{A_{k,n}}|e'(z')-e'(z)|_\mathcal{H}\,d\lambda(z') < \frac{\delta}{3},\]
        using $\operatorname{diam}_\Z(A_{k,n})\leq\rho_n<\rho$. Integrating over $z\in\Z$ gives $\|\pi_n(e')-e'\|_{\mathbb{H}}\leq\delta/3$ for $n\geq n_0$, and the triangle inequality together with the operator-norm bound of item 1 gives
        \[\|\pi_n(e)-e\|_{\mathbb{H}}\leq \|\pi_n(e-e')\|_{\mathbb{H}}+\|\pi_n(e')-e'\|_{\mathbb{H}}+\|e'-e\|_{\mathbb{H}}\leq \delta,\quad\text{for all }n\geq n_0.\]

        For the uniform statement, fix $\delta>0$, and let $(e_j)_{j\in[M]}\subseteq \mathbb{K}$ be such that the open balls of radius $\frac{\delta}{4}$ centered at these points cover $\mathbb{K}$ -- which exist, since the set is compact. By the first part, applied to each of the finitely many $e_j$, there is $n_0\in\N$ such that
        \[\|\pi_n(e_j)-e_j\|_{\mathbb{H}}\leq \frac{\delta}{2},\quad\text{for every }j\in[M]\text{ and }n\geq n_0.\]
        For any $e\in\mathbb{K}$, pick $j$ with $\|e-e_j\|_{\mathbb{H}}\leq \frac{\delta}{4}$; then, using again that $\pi_n$ has operator norm at most $1$,
        \[\|\pi_n(e)-e\|_{\mathbb{H}} \leq \|\pi_n(e-e_j)\|_{\mathbb{H}} + \|\pi_n(e_j)-e_j\|_{\mathbb{H}} + \|e_j-e\|_{\mathbb{H}} \leq \|\pi_n(e_j)-e_j\|_{\mathbb{H}}+ 2 \|e-e_j\|_{\mathbb{H}}\leq \delta,\]
        for all $n\geq n_0$. Since $e\in\mathbb{K}$ was arbitrary, this proves the claim.
    \end{enumerate}
\end{proof}

\subsubsection{Proofs of Section~\ref{subsec:costs-correspondences}}

\begin{proof}[Proof of Lemma \ref{lemma:correspondences-objectives}]\label{app:proof_correspondences_objectives}
    The function $g:\W\rightarrow\mathbb{H}$ is closed, since it is continuous and $\W$ is compact; likewise, as $\pi_n$ is continuous, so is the composition $\pi_n\circ g:\W\rightarrow\mathbb{H}$. By Theorem 17.7.1 in \cite{aliprantis2006infinite}, the preimage correspondence of a closed function is upper hemicontinuous, so both $G$ and
    \[\tilde{G}_n(e)=\{\w\in\W: \pi_n(g(\w))=e\}\]
    are upper hemicontinuous. Since $G_n=\tilde{G}_n\circ\pi_n$ Theorem 17.23 in \cite{aliprantis2006infinite} (composition of a continuous function with an upper hemicontinuous correspondence is upper hemicontinuous) gives that $G_n$ is upper hemicontinuous as well.

    The same two facts give closed-valuedness: $G(e)=g^{-1}(e)$ and $\tilde{G}_n(e)=(\pi_n\circ g)^{-1}(e)$ are preimages of a point under a continuous map, hence closed, and $G_n(e)=\tilde{G}_n(\pi_n(e))$ inherits closedness from $\tilde{G}_n$. Closed subsets of the compact set $\W$ are compact, so $G$ and $G_n$ are compact-valued.

    To show that $\Psi,\Psi_n$ are lower semicontinuous, denote by $E_\alpha$ and $E_{\alpha,n}$ their level sets,
    \[E_\alpha=\{e\in\mathbb{H}:\Psi(e)\leq \alpha\}\quad\text{and}\quad E_{\alpha,n}=\{e\in\mathbb{H}:\Psi_n(e)\leq \alpha\},\]
    and let $M=\sup_{\w\in\W}\Lambda(\w)$. If $\alpha\geq M+1$, both level sets equal $\mathbb{H}$ by definition of $\Psi,\Psi_n$, hence are closed. Suppose instead $\alpha<M+1$, and let $(e_k)_{k\in\N}\subseteq E_\alpha$ (resp. $(e_k)_{k\in\N}\subseteq E_{\alpha,n}$) converge to some $e\in\mathbb{H}$; we show $e\in E_\alpha$ (resp. $e\in E_{\alpha,n}$).

    Since $\alpha<M+1$, each $e_k$ is achievable, $e_k\in g(\W)$ (resp. $\pi_n(e_k)\in\pi_n(g(\W))$), so $G(e_k)\neq\emptyset$ (resp. $G_n(e_k)\neq\emptyset$). As $\Lambda$ is lower semicontinuous and $G(e_k)$ (resp. $G_n(e_k)$) is non-empty and compact, it attains a minimizer $\w_k\in G(e_k)$ (resp. $\w_k\in G_n(e_k)$) with $\Psi(e_k)=\Lambda(\w_k)$ (resp. $\Psi_n(e_k)=\Lambda(\w_k)$). By compactness of $\W$, $(\w_k)_{k\in\N}$ has a subsequence $(\w_{k_l})_{l\in\N}$ converging to some $\w\in\W$, so $(e_{k_l},\w_{k_l})_{l\in\N}$ converges to $(e,\w)$ in $\mathrm{Graph}\,G$ (resp. $\mathrm{Graph}\,G_n$). Because $G$ (resp. $G_n$) is upper hemicontinuous and compact-valued, its graph is closed, so $(e,\w)\in\mathrm{Graph}\,G$ (resp. $\mathrm{Graph}\,G_n$), that is, $\w\in G(e)$ (resp. $\w\in G_n(e)$).

    Combining $\w\in G(e)$ with the lower semicontinuity of $\Lambda$ and $\Lambda(\w_{k_l})=\Psi(e_{k_l})\leq\alpha$ for every $l$ (resp. the analogous facts for $G_n,\Psi_n$) gives
    \begin{align*}
        \Psi(e)&\leq \Lambda(\w)\leq \liminf_{l\in\N}\Lambda(\w_{k_l})=\liminf_{l\in\N}\Psi(e_{k_l})\leq \alpha,\\
        (\text{resp. }\Psi_n(e)&\leq \Lambda(\w)\leq \liminf_{l\in\N}\Lambda(\w_{k_l})=\liminf_{l\in\N}\Psi_n(e_{k_l})\leq \alpha),
    \end{align*}
    so $e\in E_\alpha$ (resp. $e\in E_{\alpha,n}$). Every level set is therefore closed, so $\Psi,\Psi_n$ are lower semicontinuous.

    Finally,note that for any feasible path $e\in\mathbb{K}$ and for all $n\in\N$, we must have $G(e)\subseteq G_{n+1}(e)\subseteq G_n(e)$ -- where the last inclusion follows from item 4 in Proposition~\ref{prop:properties-projections}. Thus, as the value of the functions $\Psi,\Psi_{n+1},\Psi_n$ are the infimum of $\Lambda$ over these sets, we must obtain:
    \[\Psi_n(e)\leq \Psi_{n+1}(e)\leq\Psi(e).\]
    Now, since $\Lambda$ is lower semicontinuous, there exists a sequence $(\omega_{n})_{n\in\N}$ of optimal points, i.e. $\w_n\in G_n(e)$ and $\Lambda(\w_n)=\Psi_n(e)$. Because $\W$ is compact there exists a convergent subsequence, converging to some $\w\in\W$. Since $\Lambda$ is upper semicontinuous and the sequence $\Psi_n(e)$ is increasing, this implies that the entire sequence converges to the value $\Lambda(\w)$. Now, as $\w_n\in G_n(e)$ for all $n\in\N$, we must have $\pi_n(g(\w_n))=e$. The fact that $g(\w_n)\in\mathbb{K}$, $g$ is continuous, and $\pi_n$ converges uniformly to the identity map in this set implies that $g(\w)=e$. Hence,
    \[\Psi(e)\geq\lim_{n\in\N}\Psi_n(e)=\Lambda(\w)\geq \Psi(e),\]
    and the proof is complete.
\end{proof}

\begin{proof}[Proof of Lemma~\ref{lemma:regularized-cost}]\label{app:proof_regularized_cost}
    Let $e\in\mathbb{H}$ be a fixed path. Then, the function
    \[\w\in\W\mapsto \Lambda(\w)+\delta^{-1}\|\pi_n(g(\w)-e)\|_{\mathbb{H}}\in\R\]
    is strictly positive, lower semicontinuous, and bounded by $1+2\delta^{-1}\mathrm{diam}(X)$. Thus, its infimum is well-defined and attainable, with the correspondence $G^\ast_{\delta,n}$  being non-empty and compact valued -- since it is a non-empty closed subset of the compact set $\W$.

    To check the Lipschitz continuity, note that if $e,e'\in\mathbb{H}$, then:
    \begin{align*}
        |\Psi_{\delta,n}(e)-\Psi_{\delta,n}(e')|&\leq\left|\inf_{\w\in\W}\{\Lambda(\w)+\delta^{-1}\|\pi_n(g(\w)-e)\|_{\mathbb{H}}\}-\inf_{\w'\in\W}\{\Lambda(\w')+\delta^{-1}\|\pi_n(g(\w')-e')\|_{\mathbb{H}}\}\right|,\\
        &\leq \delta^{-1}\|\pi_n(e-e')\|_{\mathbb{H}}|\leq \delta^{-1}\|e-e'\|_{\mathbb{H}},
    \end{align*}
    where the last inequality uses the fact that the operator norm of $\pi_n$ is smaller than or equal to $1$.

    The graph of $G^\ast_{\delta,n}$ is closed: if $(e_k,\w_k)\in\mathrm{graph}(G^\ast_{\delta,n})$ converges to $(e,\w)\in\mathbb{H}\times\W$, continuity of $\Psi_{\delta,n}$ (shown above) and of $\w\mapsto\|\pi_n(g(\w)-e)\|_{\mathbb{H}}$, together with lower semicontinuity of $\Lambda$, give
    \[\Psi_{\delta,n}(e)=\lim_{k\in\N}\Psi_{\delta,n}(e_k)=\liminf_{k\in\N}\Big(\Lambda(\w_k)+\delta^{-1}\|\pi_n(g(\w_k)-e_k)\|_{\mathbb{H}}\Big)\geq \Lambda(\w)+\delta^{-1}\|\pi_n(g(\w)-e)\|_{\mathbb{H}}.\]
    Since $\w$ is itself a feasible candidate in the infimum defining $\Psi_{\delta,n}(e)$, the reverse inequality also holds. So, the two must be equal,  $\w\in G^\ast_{\delta,n}(e)$, and the graph is closed. By the Closed Graph Theorem -- Theorem 17.11 in \cite{aliprantis2006infinite} -- the correspondence is upper hemicontinuous.

    Also observe that if $\w\in G^\ast_{\delta,n}(e)$ with $\pi_n(e)\in\pi_n(\mathbb{K})$, then for any $\w'\in \W$ with $\pi_n(g(\w'))=\pi_n(e)$, the distance term in the definition of $\Psi_{\delta,n}(e)$ vanishes at $\w'$, so
    \[\Lambda(\w')\geq \Psi_{\delta,n}(e)=\Lambda(\w)+\delta^{-1}\|\pi_n(g(\w)-e)\|_{\mathbb{H}}\geq \delta^{-1}\|\pi_n(g(\w)-e)\|_{\mathbb{H}},\]
    using $\Lambda(\w)\geq0$. Since $\Lambda(\w')\leq1$, we conclude that $\|\pi_n(g(\w)-e)\|_{\mathbb{H}}\leq \delta$.

    Finally, that $\Psi_{\delta,n}(e)$ increases monotonically as $\delta$ decreases follows directly from the definition. Suppose $\pi_n(e)\in\pi_n(\mathbb{K})$, and let $\w_\delta\in G^\ast_{\delta,n}(e)$ for each $\delta>0$. We claim that any point of adherence $\w\in\W$ of $(\w_\delta)_{\delta>0}$ as $\delta\downarrow0$ satisfies $\Psi_n(e)=\Lambda(\w)$, i.e. is optimal for the unregularized cost. Such points of adherence exist because $\W$ is compact; fix one, along a subsequence $\delta\downarrow0$ with $\w_\delta\to\w$. Since $\|\pi_n(g(\w_\delta)-e)\|_{\mathbb{H}}\leq \delta\to0$ and $\w\mapsto\|\pi_n(g(\w)-e)\|_{\mathbb{H}}$ is continuous, $\|\pi_n(g(\w)-e)\|_{\mathbb{H}}=0$, so $\w$ is feasible for $\Psi_n(e)$ and
    \[\Lambda(\w)\geq \Psi_{n}(e)\geq \Psi_{\delta,n}(e)\geq \Lambda(\w_\delta),\]
    the first inequality by feasibility of $\w$, the second by the monotonicity just shown, and the third since $\Psi_{\delta,n}(e)=\Lambda(\w_\delta)+\delta^{-1}\|\pi_n(g(\w_\delta)-e)\|_{\mathbb{H}}\geq\Lambda(\w_\delta)$. As $\Lambda$ is lower semicontinuous and $\w_\delta\to\w$, taking $\delta\downarrow0$ along the subsequence gives
    \[\Lambda(\w)\geq\Psi_{n}(e)\geq\lim_{\delta\downarrow0}\Psi_{\delta,n}(e)\geq\liminf_{\delta\downarrow0}\Lambda(\w_\delta)\geq\Lambda(\w),\]
    so all four terms coincide. In particular $\Psi_n(e)=\Lambda(\w)$, as claimed, and since the monotone limit $\lim_{\delta\downarrow0}\Psi_{\delta,n}(e)$ equals $\Psi_n(e)$ as well, this proves $\Psi_{\delta,n}(e)\uparrow\Psi_n(e)$.
\end{proof}

\begin{proof}[Proof of Lemma~\ref{lemma:c-is-cont}]\label{app:proof_c_is_cont}
    By Proposition~\ref{prop:properties-projections} item 3, $\tilde{\pi}_n$ is an isometry onto $\X^n$, so $c_{\delta,n}(x):=\Psi_{\delta,n}(\tilde{\pi}_n^{-1}(x))$ is well-defined and measurable, with boundedness inherited directly from the bound on $\Psi_{\delta,n}$ in Lemma~\ref{lemma:regularized-cost}. For Lipschitz continuity, the isometry gives $\|\tilde{\pi}_n^{-1}(x)-\tilde{\pi}_n^{-1}(y)\|_{\mathbb{H}}=|x-y|_{\X^n}$ for any $x,y\in\X^n$, so the $\delta^{-1}$-Lipschitz continuity of $\Psi_{\delta,n}$ carries over directly:
    \[|c_{\delta,n}(x)-c_{\delta,n}(y)|=|\Psi_{\delta,n}(\tilde{\pi}_n^{-1}(x))-\Psi_{\delta,n}(\tilde{\pi}_n^{-1}(y))|\leq \delta^{-1}\|\tilde{\pi}_n^{-1}(x)-\tilde{\pi}_n^{-1}(y)\|_{\mathbb{H}}=\delta^{-1}|x-y|_{\X^n}.\]
\end{proof}

\begin{proof}[Proof of Lemma~\ref{lemma:log-sum-exp-cost}]\label{app:proof_log_sum_exp_cost}
    Under the stated hypotheses, let $e\in\mathbb{H}$ be fixed and denote by $\w^\ast\in\W$ a minimizer of the cost $\Psi_{\delta,n}(e)$, so that
    \[\Psi_{\delta,n}(e) = \Lambda(\w^\ast)+\delta^{-1} \|\pi_n(g(\w^\ast)-e)\|_{\mathbb{H}}.\]
    By the definition of $\Psi_{\delta,n,\zeta}$, restricting the expectation to the ball $\{d(\w,\w^\ast)<t\}$, for any $t>0$, we must have that:
    \[0\leq \Psi_{\delta,n,\zeta}(e) - \Psi_{\delta,n}(e) \leq -\zeta\log\left(E_{Q_0}\left[ e^{-\zeta^{-1} (\Lambda(\w)-\Lambda(\w^\ast)+\delta^{-1} (\|\pi_n(g(\w)-e)\|_{\mathbb{H}}-\|\pi_n(g(\w^\ast)-e)\|_{\mathbb{H}}))}1_{d(\w,\w^\ast)<t}\right]\right).\]
    Using the Lipschitz continuity of $\Lambda$ and $g$, we have that for any $\w$ such that $d(\w,\w^\ast)<t$:
    \[\Lambda(\w)+\delta^{-1} \|\pi_n(g(\w)-e)\|_{\mathbb{H}} \leq \Lambda(\w^\ast)+\delta^{-1} \|\pi_n(g(\w^\ast)-e)\|_{\mathbb{H}} + M(1+\delta^{-1})t.\]
    Therefore,
    \[0\leq \Psi_{\delta,n,\zeta}(e) - \Psi_{\delta,n}(e) \leq M(1+\delta^{-1})t-\zeta\log Q_0(d(\w,\w^\ast)<t) \text{ for all } t>0.\]
    In particular, taking the infimum over $t>0$ yields
    \[0\leq \Psi_{\delta,n,\zeta}(e) - \Psi_{\delta,n}(e) \leq \inf_{t>0}\left\{M(1+\delta^{-1})t-\zeta\log Q_0(d(\w,\w^\ast)<t)\right\} = \zeta s^\ast\left(\frac{M(1+\delta^{-1})}{\zeta},\w^\ast\right)\]
    establishing the desired bound.
\end{proof}

\subsubsection{Proofs of Section~\ref{subsec:marginal-operators}}

\begin{proof}[Proof of Lemma~\ref{lemma:marginal-operators}]\label{app:proof_marginal_operators}
    Linearity is immediate from linearity of the integral defining each operator. Positivity holds because each operator is defined by integrating $\phi$ against probability measures (namely $\lambda$, or the weights $\lambda_{k,n}$), so $\phi\geq0$ implies a non-negative output. The constant function $1$ is mapped to $1$ because integrating $1$ against a probability measure returns $1$: for instance $T1(e)=\int_\Z 1\,d\lambda(z)=1$, and likewise for $T_n$ and $\mathrm{T}_n$. Together, positivity and unit-preservation give the operator-norm bound: for any $\phi$ with $\|\phi\|_\infty\leq1$, both $1-\phi\geq0$ and $1+\phi\geq0$, so positivity gives $T(1-\phi)\geq0$ and $T(1+\phi)\geq0$, that is, $|T\phi|\leq T1=1$ (and likewise for $T_n$, $\mathrm{T}_n$); hence each operator has operator norm at most $1$ and, in particular, is continuous.
\end{proof}

\begin{proof}[Proof of Lemma~\ref{lemma:prob-to-prob}]\label{app:proof_prob_to_prob}
    Under the above conditions, the norm duals of $C_b(K)$ and $C_b(H)$ can be identified with the set of finite signed regular measures on $K$, $ca_r(K)$, and the set of normal signed charges of bounded variation on $H$, $ba_n(H)$, respectively -- see Theorems 14.10 and 14.14 in \cite{aliprantis2006infinite}. By the continuity and linearity of $S$, the adjoint operator $S^\ast: ba_n(H) \rightarrow ca_r(K)$ is well-defined, linear and continuous. Moreover, by the positivity of $S$, $S^\ast\mu$ is a positive regular measure on $K$ for any $\mu\in\P(H)$. Finally, because $S$ maps the constant function $1$ to $1$, the total mass of $S^\ast\mu$ is
    \[S^\ast\mu(K)=E_{S^\ast\mu}[1]=E_{\mu}[S1]=E_{\mu}[1]=1,\]
    for any $\mu\in\P(H)$. Hence $S^\ast\mu$ is a probability measure on $K$, and the adjoint operator $S^\ast:\P(H)\rightarrow \P(K)$ is well-defined, linear on convex combinations, continuous, and takes values in the set of probability measures.
\end{proof}

\begin{proof}[Proof of Proposition~\ref{prop:feasible-sets}]\label{app:proof_feasible_sets}
    By Lemma~\ref{lemma:marginal-operators} and \ref{lemma:prob-to-prob}, the adjoint operators $T^\ast$, $T_n^\ast$, and $\mathrm{T}_n^\ast$ are well-defined, linear on convex combinations, continuous, and map probability measures to probability measures ($n$-tuples in the case of the latter). This makes $\mathcal{Q}(\mathbf{P},\lambda)$, $\mathcal{Q}_n(\mathbf{P},\lambda)$, $\mathcal{M}(\mathbf{P},\lambda)$, and $\Pi_n(\mathbf{P},\lambda)$ convex and closed directly, as preimages of the single point $P$ (or $P_n$, or $(\mathbf{P}_{k,n})_{k\in[n]}$) under a continuous affine map, intersected where relevant with the convex closed support constraints $\{\mu:\mathrm{supp}(\mu)\subseteq\mathbb{K}\}$ or $\{\mu:\mathrm{supp}(\pi_n\#\mu)\subseteq\pi_n(\mathbb{K})\}$.

    The pushed-forward sets ${\pi_n}_\#\left(\mathcal{C}_n(\mathbf{P},\lambda)\right)$ and ${\pi_n}_\#\left(\mathcal{M}_n(\mathbf{P},\lambda)\right)$ require a separate argument, since a continuous image of a closed set need not itself be closed. Here, however, $T_n^\ast\mu$ depends on $\mu$ only through its projection $\pi_n\#\mu$ -- indeed $T_n\phi=T\phi\circ\pi_n$ gives $T_n^\ast=T^\ast\circ(\pi_n)_\#$ -- so the constraint $T_n^\ast\mu=P_n$ is equivalent to $T^\ast(\pi_n\#\mu)=P_n$, a condition on $\pi_n\#\mu$ alone. Hence ${\pi_n}_\#\left(\mathcal{C}_n(\mathbf{P},\lambda)\right)=\{\nu\in\P(\pi_n(\mathbb{H})): T^\ast\nu=P_n\}$ and ${\pi_n}_\#\left(\mathcal{M}_n(\mathbf{P},\lambda)\right)=\{\nu\in\P(\pi_n(\mathbb{H})): T^\ast\nu=P_n, \mathrm{supp}(\nu)\subseteq\pi_n(\mathbb{K})\}$ are themselves directly of the preimage form above, hence convex and closed by the same argument.

    Compactness follows from the fact that all these sets are equitight: $\mathcal{Q}(	\mathbf{P},\lambda)$ and $\mathcal{Q}_n(\mathbf{P},\lambda)$ are equitight because they are supported in the compact set $\W$; $\mathcal{M}(\mathbf{P},\lambda)$ and ${\pi_n}_\# \left(\mathcal{M}_n(	\mathbf{P},\lambda)\right)$ are equitight because they are supported in the compact sets $\mathbb{K}$ and $\pi_n(\mathbb{K})$, respectively; ${\pi_n}_\# \left(\mathcal{C}_n(\mathbf{P},\lambda)\right)$ is equitight as it is supported in the compact set $\pi_{n}(\mathbb{H})$; finally, $\Pi_n(\mathbf{P},\lambda)$ is equitight because it is supported in the compact set $\X^n$. By Prokhorov's theorem, these sets are relatively compact, and since they are also closed, they must be compact.
\end{proof}

\begin{proof}[Proof of Lemma~\ref{lemma:marginals-of-T}]\label{app:proof_marginals_of_T}
   It suffices to show the second claim, as the first one reduces to the case $B=\X$. Take $\phi(z,x)=\psi(z)\eta(x)$, where $\psi\in C_b(\Z)$, and $\eta\in C_b(\X)$. Then
    \[T_n\phi(e)=\int_\Z \mathbf \psi(z)\eta(\pi_n(e)_z)\,d\lambda(z)=\sum_{k\in[n]}\eta(\tilde\pi_{k,n}(e))\int_{A_{k,n}}\psi(z)\,d\lambda(z),\]
    using that $\pi_n(e)$ is constant, equal to $\tilde\pi_{k,n}(e)$, on each $A_{k,n}$. Taking expectations under $\mu$,
    \[T_n^\ast\mu(A\times B)=E_\mu[T_n\phi]=\sum_{k\in[n]}E_{\lambda}[\psi \mathbf{1}_{A_{k,n}}]\,E_\mu\big[\eta(\tilde\pi_{k,n}(e))\big]=\sum_{k\in[n]}E_{\lambda}[\psi \mathbf{1}_{A_{k,n}}]\,E_{\tilde\pi_{k,n\#}\mu}\big[\eta(e)\big].\]
    As this holds for any continuous functions of each variable, we conclude the claim. Finally, for the $\Z$-marginal of $T^\ast$, observe that for any $\psi\in C_b(\Z)$, then:
    \[E_{T^\ast\mu}[\psi]= E_\mu[T\psi]= \int_\Z \psi(z)d\lambda(z).\]
\end{proof}

\begin{proof}[Proof of Proposition~\ref{prop:nested-constraints}]\label{app:proof_nested_constraints}
    We first establish the characterization of $\mathcal{Q}(\mathbf{P},\lambda)$.
Let $Q_{n_k}\in\mathcal{Q}_{n_k}(\mathbf{P},\lambda)$ with $n_k\to\infty$ and $Q_{n_k}\Rightarrow Q$. Since $g$ is continuous, $g_\#Q_{n_k}\Rightarrow g_\#Q$. Fix $\phi\in C_b(\Z\times\X)$: uniform continuity of $\phi$ on the compact set $\Z\times\X$, together with the uniform convergence $\sup_{e\in\mathbb{K}}\|\pi_n(e)-e\|_{\mathbb{H}}\to0$ of Proposition~\ref{prop:properties-projections} item 5, gives $T_{n_k}\phi\to T\phi$ uniformly on $\mathbb{K}$, the compact set on which every $g_\#Q_{n_k}$ and $g_\#Q$ are supported. Hence
    \[\big|E_{g_\#Q_{n_k}}[T_{n_k}\phi]-E_{g_\#Q}[T\phi]\big|\leq \|T_{n_k}\phi-T\phi\|_{\infty,\mathbb{K}}+\big|E_{g_\#Q_{n_k}}[T\phi]-E_{g_\#Q}[T\phi]\big|\longrightarrow 0,\]
    the first term by uniform convergence, the second since $g_\#Q_{n_k}\Rightarrow g_\#Q$ and $T\phi$ is a fixed bounded continuous function. The left-hand quantity equals $E_{T_{n_k}^\ast(g_\#Q_{n_k})}[\phi]=E_{P_{n_k}}[\phi]$, which converges to $E_P[\phi]$ since $P_n\Rightarrow P$. Combining the two limits, $E_{g_\#Q}[T\phi]=E_P[\phi]=E_{T^\ast(g_\#Q)}[\phi]$ for every $\phi\in C_b(\Z\times\X)$, that is, $T^\ast(g_\#Q)=P$, so $Q\in\mathcal{Q}(\mathbf{P},\lambda)$.

    The converse inclusion is simply a consequence of Assumption~\ref{ass:constraints}: for $Q\in\mathcal{Q}(\mathbf{P},\lambda)$ it provides $Q_n\in\mathcal{Q}_n(\mathbf{P},\lambda)$ with $Q_n\Rightarrow Q$, so $Q\in\overline{\bigcup_{m\geq n}\mathcal{Q}_m(\mathbf{P},\lambda)}$ for every $n\in\N$. Together, these give $\mathcal{Q}(\mathbf{P},\lambda)=\bigcap_{n\in\N}\overline{\bigcup_{m\geq n}\mathcal{Q}_m(\mathbf{P},\lambda)}$.

    We now turn to $\mathcal{M}(\mathbf{P},\lambda)$. Let $\mu_{n_k}\in\mathcal{M}_{n_k}(\mathbf{P},\lambda)$ with $n_k\to\infty$ and $\mu_{n_k}\Rightarrow\mu$. The argument above, with $g_\#Q_{n_k}$ replaced by $\mu_{n_k}$ throughout, gives $T^\ast\mu=P$; it remains to show $\mathrm{supp}(\mu)\subseteq\mathbb{K}$.

    Let us denote by $\phi_n,\phi\in C_b(\mathbb{H})$ the functions measuring the distance to the sets $\pi_n(\mathbb{K})$ and $\mathbb{K}$:
    \[\phi_n(e)=\inf_{e'\in\mathbb{K}}\|e-\pi_n(e')\|_{\mathbb{H}},\quad\text{ and }\quad\phi(e)=\inf_{e'\in\mathbb{K}}\|e-e'\|_{\mathbb{H}}.\]
    Note that these functions are bounded since $\X$ is compact, hence, $\mathbb{H}$ is a bounded set. Moreover, by Proposition~\ref{prop:properties-projections} item 5, these functions are at maximum distance $\kappa_n$. Thus, $\phi_n$ converges uniformly to $\phi$. The condition that $\pi_{n_k\#}\mu_{n_k}$ is supported in $\pi_{n_k}$ implies $E_{\pi_{n_k\#}\mu_{n_k}}[\phi_{n_k}]=0$ for all $k\in \N$. Now, note that $\pi_{n_k}$ converges pointwise to the identity map in $\mathbb{H}$ and $\mu_{n_k}\Rightarrow \mu$. This implies that $\pi_{n_k\#}\mu_{n_l}\Rightarrow \mu$. This and the fact that $\pi_n$ converges uniformly to $\phi$ implies that:
    \[E_{\mu}[\phi]=\lim_{k\in \N}E_{\pi_{n_k}}[\phi_{n_k}]=0.\]
    This is possible only if $\mu(\mathbb{K})=1$, from where we obtain that $\mu\in \mathcal{M}(\mathbf{P},\lambda)$.

    Conversely, for $\mu\in\mathcal{M}(\mathbf{P},\lambda)$, Proposition~\ref{prop:co-mechanism} gives $Q\in\mathcal{Q}(\mathbf{P},\lambda)$ with $g_\#Q=\mu$; the $\mathcal{Q}$-characterization just proved gives $Q_n\in\mathcal{Q}_n(\mathbf{P},\lambda)$ with $Q_n\Rightarrow Q$, and since $g_\#\mathcal{Q}_n(\mathbf{P},\lambda)\subseteq\mathcal{M}_n(\mathbf{P},\lambda)$ and $g_\#$ is continuous, $\mu_n:=g_\#Q_n\in\mathcal{M}_n(\mathbf{P},\lambda)$ satisfies $\mu_n\Rightarrow g_\#Q=\mu$. Hence $\mu\in\overline{\bigcup_{m\geq n}\mathcal{M}_m(\mathbf{P},\lambda)}$ for every $n\in\N$, and $\mathcal{M}(\mathbf{P},\lambda)=\bigcap_{n\in\N}\overline{\bigcup_{m\geq n}\mathcal{M}_m(\mathbf{P},\lambda)}$.
\end{proof}

\subsubsection{Proofs of Section~\ref{subsec:eot_properties}}

\begin{proof}[Proof of Lemma~\ref{lemma:coupling-approximation}]\label{app:proof_coupling_approximation}
    The construction is a direct extension of the shadow coupling of Definition 3.1 in \citet{eckstein2023quantitative}. For each $k\in[n]$, let $\gamma_k\in \Pi(\rho_k,\tilde{\rho}_k)$ be an optimal coupling between $\rho_k$ and $\tilde{\rho}_k$ for the $W_1$-distance, and $K_k$ the corresponding kernel, i.e.\ $\gamma_k= K_k\otimes \rho_k$. Define the kernel $K:\X^n\rightarrow \P(\X^n)$ by $K(x_1,\ldots,x_n)=\bigotimes_{k\in[n]} K_k(x_k)$, and set $\gamma:=K\otimes \mu\in \P(\X^n\times \X^n)$. We use $\tilde{\mu}\in\P(\X^n)$ to denote its second marginal, the shadow of $\mu$ under $\tilde{\mathbf{\rho}}$. By construction, $\tilde{\mu}\in \Pi_n(\tilde{\mathbf{\rho}})$, and $\gamma$ is a coupling of $\mu$ and $\tilde{\mu}$. Using this coupling and the optimality of each $\gamma_k$,
    \[W_1(\mu,\tilde{\mu})\leq E_{\gamma}[|X-Y|_{\X^n}]= \sum_{k\in[n]} E_{\gamma_k}[|X-Y|_{\X}]\lambda_{k,n}=\sum_{k\in[n]} W_1(\rho_k,\tilde{\rho}_k)\lambda_{k,n}.\]
    Since $\mu$ and $\tilde\mu$ have marginals $\mathbf{\rho}$ and $\tilde{\mathbf{\rho}}$, respectively, the reverse inequality is exactly \eqref{eq:w1-marginal-lower-bound}, so the two combine into the claimed identity.

    For the entropy inequality, write $\bar{\gamma}:=\bigotimes_{k\in[n]}\gamma_k\in\P(\X^n\times\X^n)$ for the analogous gluing of the marginals, so that $\bar\gamma = K\otimes\big(\bigotimes_{k\in[n]}\rho_k\big)$. Such distribution shares the same kernel $K$ as $\gamma=K\otimes\mu$, only with $\mu$ replaced by $\bigotimes_{k\in[n]}\rho_k$. The chain rule for relative entropy therefore gives $H(\gamma|\bar\gamma)=H(\mu|\bigotimes_{k\in[n]}\rho_k)+E_\mu\big[H(K(x,\cdot)|K(x,\cdot))\big]=H(\mu|\bigotimes_{k\in[n]}\rho_k)$, the correction term vanishing as both couplings disintegrate through the same kernel $K$. Since $\tilde\mu$ and $\bigotimes_{k\in[n]}\tilde\rho_k$ are, respectively, the second marginals of $\gamma$ and $\bar\gamma$, the data-processing inequality for relative entropy under marginalization gives
    \[H\Big(\tilde{\mu}\,\Big|\bigotimes_{k\in[n]} \tilde{\rho}_k\Big)\leq H(\gamma|\bar\gamma)= H\Big(\mu\,\Big|\bigotimes_{k\in[n]} \rho_k\Big),\]
    which concludes the proof.
\end{proof}

\begin{proof}[Proof of Lemma~\ref{lemma:stability-eot}]\label{app:proof_stability_eot}
    The proof is a direct extension of Theorem 3.7(ii) in \citet{eckstein2023quantitative}. Since $c$ is bounded and Lipschitz and $\X^n$ is compact, $S(\mathbf\rho)$ equals the value of the (attained) primal problem $\inf_{\nu\in\Pi_n(\mathbf\rho)}E_\nu[c]+\epsilon H(\nu|\bigotimes_{k\in[n]}\rho_k)$, by the same duality result underlying Proposition~\ref{prop:entropic-ot-dual}. Let $\mu\in \Pi_n(\mathbf{\rho})$ be its optimizer, and let $\tilde{\mu}\in \Pi_n(\tilde{\mathbf{\rho}})$ be the coupling in Lemma~\ref{lemma:coupling-approximation}. Since $\tilde\mu$ is feasible, not necessarily optimal, for the EOT $S(\tilde{\mathbf\rho})$, and its relative entropy is no larger than that of $\mu$,
    \begin{align*}
        S(\tilde{\mathbf{\rho}})-S(\mathbf{\rho})&\leq \Big(E_{\tilde{\mu}}[c]+\epsilon H\big(\tilde\mu\,\big|\textstyle\bigotimes_{k\in[n]}\tilde\rho_k\big)\Big)-\Big(E_{\mu}[c]+\epsilon H\big(\mu\,\big|\textstyle\bigotimes_{k\in[n]}\rho_k\big)\Big)\leq E_{\tilde{\mu}}[c]-E_{\mu}[c]\\
        &\leq L\, W_1(\mu,\tilde{\mu})\\
        &= L\sum_{k\in[n]} W_1(\rho_k,\tilde{\rho}_k)\lambda_{k,n},
    \end{align*}
    where the second inequality follows from the $L$-Lipschitz continuity of $c$ and the Kantorovich-Rubinstein duality, and the last equality from Lemma~\ref{lemma:coupling-approximation}. By symmetry, exchanging the roles of $\mathbf{\rho}$ and $\tilde{\mathbf{\rho}}$, we obtain the desired result.
\end{proof}

\begin{proof}[Proof of Lemma~\ref{lemma:finite-support-eot-properties}]\label{app:proof_finite_support_eot_properties}
For the first assertion, note that, for a fixed dual potential $\phi\in C_b(\X)^n$, since each $\rho_k$ is supported on $\X_m$, the objective function of the entropic transport, in its dual formulation, reduces to the finite sum
\[E_{R}\left[\mathrm{T}_n\phi - \epsilon\left(e^{\frac{\mathrm{T}_n\phi -c}{\epsilon}}-1\right)\right]= \sum_{\ell\in [m]^n}\left( \sum_{k\in[n]}\phi_{k}(x_{\ell_k})\lambda_{k,n} - \epsilon\left(C(\ell)\,e^{\sum_{k\in[n]}\phi_{k}(x_{\ell_k})\lambda_{k,n}/\epsilon}-1\right)\right)\prod_{k\in[n]}p_{k,\ell_k}, \]
with $p_{k,\ell_k}=\rho_k(\{x_{\ell_k}\})$. Setting $\phi_{k,\ell_k}=\phi_k(x_{\ell_k})$, and noting that every $\phi\in\R^{n\times m}$ arises this way from some $\phi\in C_b(\X)^n$ (any function on the finite set $\X_m$ extends to a continuous function on $\X$), the right-hand side equals $\Phi(\phi,p,\lambda)$, so that taking the supremum over $\phi\in C_b(\X)^n$ on the left is equivalent to maximizing $\Phi(\cdot,p,\lambda)$ over $\phi\in\R^{n\times m}$ on the right, giving the desired equality $S(\mathbf\rho)=\sup_{\phi\in\R^{n\times m}}\Phi(\phi,p,\lambda)$; that the supremum is in fact attained is shown next.

As for the existence of a unique optimal solution in $\mathcal{K}_p$, note that this set is compact, and the objective function $\Phi(\cdot,p,\lambda)$ is continuous and concave -- concavity holds since $\phi\mapsto C(\ell)e^{\lambda^T\phi_{\cdot,\ell}/\epsilon}$ is convex for each $\ell$, and the weights $\prod_{k\in[n]}p_{k,\ell_k}\geq0$ -- so a maximizer exists by Weierstrass's theorem. Restricted to $\mathcal{K}_p$, $\Phi(\cdot,p,\lambda)$ is in fact strictly concave: writing $\Sigma:=\mathrm{diag}(\lambda_{k,n}p_{k,\ell})_{(k,\ell)\in[n]\times[m]}$, the Hessian of $\Phi$ satisfies
\[\eta^T D_\phi^2 \Phi(\phi,p,\lambda)\eta \leq - \kappa \left(\eta^T\Sigma \eta\right)=-\kappa \sum_{k\in[n]} p_{k,\cdot}^T\eta^2_{k,\cdot} \lambda_{k,n},\]
for every feasible direction $\eta\in\R^{n\times m}$, where $\kappa = \frac{1}{\epsilon}e^{-\frac{(6n+1)\|c\|_\infty}{\epsilon}}$. Since the right-hand side is always strictly negative for feasible directions -- at least one $p_{k,\ell_k}\lambda_{k,n}\eta^2_{k,\ell_k}>0$ -- $\Phi$ is strictly concave on $\mathcal{K}_p$. By Lemma 4.4 in \citet{dimarino2020schrodinger}, there is a solution to the dual of $S(\mathbf\rho)$ with $\|\phi\|_\infty\leq 3\|c\|_\infty$, so $|\phi_k(x_\ell)|\leq 3\|c\|_\infty$ for every $k\in[n]$ and $\ell\in[m]$. In the first-order optimality condition, only pairs $(k,\ell_k)$ with $p_{k,\ell_k}>0$ affect the value of the dual variables or of the program, so we may without loss of generality set $\phi_k(x_\ell)=0$ whenever $p_{k,\ell}=0$. Define $\phi^\ast\in \R^{n\times m}$ by setting
\[\phi^\ast_{k,\ell}=\begin{cases}
    \phi_{k,\ell}-p_k^T\phi_{k,\cdot},&\text{ for all }(k,\ell)\in[n-1]\times [m]\text{ so that }p_{k,\ell}>0,\\[2pt]
    \displaystyle\phi_{k,\ell}+\frac{1}{\lambda_{n,n}}\sum_{j\in[n-1]}\lambda_{j,n}\,p_j^T\phi_{j,\cdot},&\text{ if }k=n\text{ and }p_{n,\ell}>0,\\[6pt]
    0,&\text{ if }p_{k,\ell}=0.
\end{cases}\]
We claim that $\phi^\ast\in \mathcal{K}_p$ is a maximum. By construction, $p_k^T\phi^\ast_{k,\cdot}=0$ for all $k\in[n-1]$, and $\phi^\ast_{k,\ell}=0$ if $p_{k,\ell}=0$. Moreover, the $\lambda$-weighted shift added to $\phi_n$ is exactly what compensates the shifts subtracted from $\phi_1,\ldots,\phi_{n-1}$: a direct computation shows $\lambda^T\phi^\ast_{\cdot,\ell}=\lambda^T\phi_{\cdot,\ell}$ for every $\ell\in[m]^n$ with $\prod_{k\in[n]}p_{k,\ell_k}>0$. Since $\Phi(\cdot,p,\lambda)$ depends on $\phi$ only through these values at active $\ell$, this gives $\Phi(\phi^\ast,p,\lambda)=\Phi(\phi,p,\lambda)$, so $\phi^\ast$ satisfies the same first-order optimality condition as $\phi$ and is itself a maximizer. The remaining step is to confirm $\|\phi^\ast\|_\infty\leq 3n\|c\|_\infty$, uniformly over $(p,\lambda)$. For $k<n$, $\|\phi_k^\ast\|_\infty\leq 2\cdot3\|c\|_\infty=6\|c\|_\infty$ follows directly from $\|\phi\|_\infty\leq3\|c\|_\infty$. As for $k=n$, the added shift $\lambda_{n,n}^{-1}\sum_{j<n}\lambda_{j,n}p_j^T\phi_{j,\cdot}$ is bounded by $3\|c\|_\infty\frac{1-\lambda_{n,n}}{\lambda_{n,n}}\leq 3(\underline{\lambda}_n^{-1}-1)\|c\|_\infty$.

To show the differentiability, note first that we can always maximize our $\Phi$ in the compact set $\|\phi\|_{\infty}\leq 3(\underline{\lambda}_n^{-1}-1)\|c\|_{\infty}$, as we have showed the existence of solutions within this set. Then, applying Danskin's Theorem (Theorem 4.13 in \citet{Bonnans2000PerturbationAO}), we conclude that $V(p,\lambda)$ is directional differentiable with directional derivative at $(p,\lambda)$ in the feasible direction $(q,\mu)$ given by:
\[d V(p,\lambda)(q,\mu)=\max_{\phi\in Opt(p,\lambda)} D_{p,\lambda}\Phi(\phi,p,\lambda)(q,\mu),\]
where $Opt(p,\lambda)$ stands for the set of all maxima within the compact set of matrices $\|\phi\|_{\infty}\leq 3(\underline{\lambda}_n^{-1}-1)\|c\|_\infty$. Now, notice that for any maxima $\phi,\phi'$, $D_q\Phi(\phi,p,\lambda)=0=D_q\Phi(\phi',p,\lambda)$, by the first-order optimality condition, and $D_p\Phi(\phi,p,\lambda)=D_p\Phi(\phi',p,\lambda)$, since this latter quantity depends only on $\lambda^T\phi_{\cdot,\ell}$, which equals to $\lambda^T\phi'_{\cdot,\ell}$ for any active paths $\prod_{k\in[n]}p_{k,\ell_k}>0$. This implies that the directional derivative equals $D_{p,\lambda}\Phi(\phi,p,\lambda)(q,\lambda)$ for any optimum $\phi$. In particular, it holds for the $\phi^\ast$ obtained above.

Finally, for the stability of the solutions, note that if $\phi^N\in \mathcal{K}_{p^N}$ is such that $\Phi(\phi^N,p^N,\lambda^N)+ o(1)= V(p^N,\lambda^N)$, then $\phi^N$ is uniformly bounded, and thus has a convergent subsequence. Let $\phi^\infty$ be the limit of such a subsequence. By continuity of $\Phi$ and $V$, we have that $\Phi(\phi^\infty,p,\lambda)=V(p,\lambda)$, and thus $\phi^\infty$ is an optimal solution for $(p,\lambda)$. Moreover, such element must also belong in $\mathcal{K}_p$. By uniqueness of the optimal solution in $\mathcal{K}_p$, we conclude that $\phi^\infty=\phi^\ast$. Thus, as every convergent subsequence of $\phi^N$ converges to the same limit $\phi^\ast$, and there is at least one convergent subsequence, we conclude that $\phi^N\to \phi^\ast$.
\end{proof}

\subsubsection{Proofs of Section~\ref{subsec:sample_properties}}

\begin{proof}[Proof of Lemma~\ref{lemma:sample-properties}]\label{app:proof_sample_properties}
    The proof combines the strong law of large numbers, the central limit theorem, and the delta method, applied through the following conditional-independence property of multinomial sampling. For this, let us define the set $\mathrm{dom}(\Gamma)\subseteq \R^{n\times m}$ as the set of all $n\times m$ matrices with non-negative entries such that each row has a strictly positive sum. Set $\Gamma$ to be the function $\Gamma:\mathrm{dom}(\Gamma)\subseteq\R_+^{n\times m}\to (\Delta_m)^n\times \Delta_n$ defined by $\Gamma(y)=(\hat{p},\hat{\lambda})$, where
    \[\hat{p}_{k,\ell}=\frac{y_{k,\ell}}{\sum_{\ell'\in[m]}y_{k,\ell'}}\quad\text{ and }\quad\hat{\lambda}_k=\sum_{\ell\in[m]}y_{k,\ell}.\]

    For (i), $\bar{Y}^N\to E[Y_1]= (P(A_{k,n}\times B_{\ell,m}))_{k\in[n],\ell\in[m]}$ almost surely as $N\to\infty$, by the strong law of large numbers. Together with the fact that $P[N_k=0\text{ for some }k]\leq n e^{-N\underline{\lambda}_n}\to0$ as $N\to\infty$, $\Gamma$ is continuous in its domain, and the limiting point $E[Y_1]$ belongs to its domain, this gives $(\hat{p}^N,\hat{\lambda}^N)=\Gamma(\bar{Y}^N)\to \Gamma(E[Y_1])=(p,\lambda)$ almost surely.

    For (ii), note that the function $\Gamma$ is continuously differentiable in its domain. Moreover, by the central limit theorem, $\sqrt{N}(\bar{Y}^N-E[Y_1])$ is asymptotically normal with mean zero and covariance matrix $\mathrm{diag}(E[Y_1])-E[Y_1]E[Y_1]^T$. As the event $N_k>0$ for all $k>0$ tends to $0$ faster than $\sqrt{N}$, the delta method gives $\sqrt{N}(\hat{p}^N-p,\hat{\lambda}^N-\lambda)=\sqrt{N}(\Gamma(\bar{Y}^N)-\Gamma(E[Y_1]))\Rightarrow D\Gamma(E[Y_1])\cdot \mathcal{N}(0,\mathrm{diag}(E[Y_1])-E[Y_1]E[Y_1]^T)$, which is a multivariate normal with mean zero and covariance matrix $\Sigma$ as stated. An immediate calculation shows that the covariance matrix of the limiting distribution is indeed $\Sigma$.

    Finally, for (iii),
    \[E[|\hat{p}_{k,\ell}^N-p_{k,\ell}|]\leq E[|\hat{p}_{k,\ell}^N-p_{k,\ell}|\mid N_k>0]+ 2 P[N_k=0]\leq \frac{2}{\sqrt{N}}\left(\frac{1+\sqrt{p_{k,\ell}(1-p_{k,\ell})}}{\sqrt{\underline{\lambda}_n}}\right),\]
    where the last inequality follows since, conditional on $N_k$, $N_k\hat{p}^N_{k,\ell}$ is a binomial random variable with parameters $N_k$ and $p_{k,\ell}$, so $E[|\hat{p}^N_{k,\ell}-p_{k,\ell}|\mid N_k]\leq\sqrt{\mathrm{Var}(\hat{p}^N_{k,\ell} \mid N_k)}=\sqrt{p_{k,\ell}(1-p_{k,\ell})N_k^{-1}}$. Taking expectations and bounding the resulting reciprocal moment by $E[N_k^{-1/2}\mid N_k>0]\leq\sqrt{2(N\lambda_{k,n})^{-1}}\leq\sqrt{2(N\underline{\lambda}_n)^{-1}}$ gives the first term, while, when $N\underline{\lambda}_n>1$, the remaining term $2P[N_k=0]$ is bounded by $2(\sqrt{\underline{\lambda}_nN})^{-1}$. This concludes the proof.
\end{proof}

\end{document}